\documentclass[manuscript]{acmart}

\usepackage{subcaption}
\usepackage[utf8]{inputenc}
\usepackage{amsmath}
\ifdef{\Bbbk}{}{\usepackage{amssymb}}
\usepackage{amsthm}
\usepackage[capitalize]{cleveref}
\usepackage{todonotes}
\usepackage{csquotes}
\usepackage{enumitem}
\usepackage{placeins}
\usepackage{thm-restate}
\usepackage{mathtools}
\usepackage{multirow}
\usepackage{booktabs}
\usepackage{placeins}
\usepackage{scalerel}

\title{Efficient Normalization of Linear Temporal Logic}

\author{Javier Esparza}
\orcid{0000-0001-9862-4919}
\affiliation{\institution{Technical University of Munich}
\department{Fakultät für Informatik}
\streetaddress{Boltzmannstrasse 3}
\city{Garching bei München}
\postcode{85748}
\country{Germany}}
\email{esparza@in.tum.de}

\author{Rubén Rubio}
\orcid{0000-0003-2983-3404}
\affiliation{\institution{Universidad Complutense de Madrid}
\department{Facultad de Informática}
\streetaddress{Calle Profesor José García Santesmases 9}
\city{Madrid}
\postcode{28040}
\country{Spain}}
\email{rubenrub@ucm.es}

\author{Salomon Sickert}
\orcid{0000-0002-0280-8981}
\affiliation{\institution{The Hebrew University}
\department{School of Computer Science and Engineering}
\streetaddress{Rothberg Family Buildings, The Edmond J. Safra Campus}
\city{Jerusalem}
\postcode{91904}
\country{Israel}}
\email{salomon.sickert@mail.huji.ac.il}

\usetikzlibrary{automata,positioning}

\tikzset{state/.style={
  rectangle,
  rounded corners,
  draw=black,
  minimum height=2em,
  minimum width=2em,
  align=center}
}

\tikzset{every picture/.append style={initial text=}}
\tikzset{accepting/.style = {double}}
\tikzset{>=stealth}
\tikzset{parallel above/.append style={transform canvas={yshift= 1mm}}}
\tikzset{parallel below/.append style={transform canvas={yshift=-1mm}}}
\tikzset{parallel right/.append style={transform canvas={xshift= 1.3mm}}}
\tikzset{parallel left/.append style={transform canvas={xshift=-1.3mm}}}

\newcommand{\lang}[0]{\ensuremath{\mathcal{L}}}

\newcommand\mayrenewcommand[2]{\ifdefined#1\expandafter\renewcommand \else\expandafter\newcommand\fi{#1}{#2}}

\newcommand{\true}{{\ensuremath{\mathbf{t\hspace{-0.5pt}t}}}}
\newcommand{\false}{{\ensuremath{\mathbf{ff}}}}
\newcommand{\F}{{\ensuremath{\mathbf{F}}}}
\mayrenewcommand{\G}{{\ensuremath{\mathbf{G}}}}
\newcommand{\X}{{\ensuremath{\mathbf{X}}}}
\mayrenewcommand{\U}{\ensuremath{\mathbin{\mathbf{U}}}}
\newcommand{\W}{\ensuremath{\mathbin{\mathbf{W}}}}
\newcommand{\M}{\ensuremath{\mathbin{\mathbf{M}}}}
\newcommand{\R}{\ensuremath{\mathbin{\mathbf{R}}}}
\renewcommand{\P}{{\ensuremath{\mathbf{P}}}}
\newcommand{\Q}{{\ensuremath{\mathbf{Q}}}}

\newcommand{\subf}{\textit{sf}\,}

\newcommand{\setmu}{\ensuremath{M}}
\newcommand{\setnu}{\ensuremath{N}}

\newcommand{\eval}[2]{{  {#1}{\llangle{#2}\rrangle} }}

\newcommand{\evalgf}[2]{{  {#1}\llbracket{#2}\rrbracket    }}

\newcommand{\flatten}[2]{{#1\langle#2\rangle}}

\newcommand{\mubasis}{B_{\G\F}}
\newcommand{\nubasis}{B_{\F\G}}
\newcommand{\Univ}{\mathcal{U}}
\renewcommand{\flat}[2]{#1\langle #2 \rangle}
\newcommand{\con}{C}
\newcommand{\context}[2]{\langle {#1} \!\mid\! {#2} \rangle}
\newcommand{\flatn}[3]{#1\langle #2 \!\mid \! #3 \rangle}
\newcommand{\llangle}{\langle\hspace{-0.075cm}\langle}
\newcommand{\rrangle}{\rangle\hspace{-0.075cm}\rangle}
\newcommand{\flatnn}[3]{#1\llangle #2 \!\mid \! #3 \rrangle}
\newcommand{\langn}[2]{L_{#1 | #2}}

\newcommand{\strict}[1]{{#1}\mathchoice{\!\!\downarrow}{\!\!\downarrow}{\downarrow}{\downarrow}}

\newcommand\N{\ensuremath{\mathbb{N}}}

\newcommand{\trans}[1]{\overset{#1}{\longrightarrow}}

\newcommand{\Op}{\mathbin{op}}

\AtBeginDocument{\newtheorem{remark}[theorem]{Remark}}

\newcommand{\aww}{\textnormal{AWW}}
\newcommand{\alw}{\textnormal{A1W}}
\newcommand{\dbw}{\textnormal{DBW}}
\newcommand{\dcw}{\textnormal{DCW}}
\newcommand{\drw}{\textnormal{DRW}}

\newcommand{\A}{\textbf{A}}

\newcommand{\GF}{\ensuremath{\mathbf{G\hspace{-0.030cm}F}\,}}
\newcommand{\FG}{\ensuremath{\mathbf{F\hspace{-0.030cm}G}\,}}

\newcommand{\fnf}{1-form}
\newcommand{\snf}{1-2-form}

\newcommand{\nnodes}[1]{|#1|}
\newcommand{\gfba}[1]{n_{\text{lim}}({#1})}
\newcommand{\ubw}[1]{n_{u}({#1})}

\newcommand{\hole}{[ \quad ]}
\newcommand{\rank}[1]{\mathit{rank}({#1})}

\newif\ifarxiv
\arxivtrue

\acmJournal{JACM}
\newgeometry{hmargin=2.5cm}

\makeatletter
\renewcommand\@copyrightpermission{This is the authors' version of the work submitted to the Journal of the ACM.}
\makeatother

\begin{abstract}
In the mid 80s, Lichtenstein, Pnueli, and Zuck proved a classical theorem stating that every formula of Past LTL (the extension of LTL with past operators) is equivalent to a formula of the form $\bigwedge_{i=1}^n \G\F \varphi_i \vee \F\G \psi_i $, where $\varphi_i$ and $\psi_i$ contain only past operators. Some years later, Chang, Manna, and Pnueli built on this result to derive a similar normal form for LTL. Both normalization procedures have a non-elementary worst-case blow-up, and follow an involved path from formulas to counter-free automata to star-free regular expressions and back to formulas. We improve on both points. We present direct and purely syntactic normalization procedures for LTL, yielding a normal form very similar to the one by Chang, Manna, and Pnueli, that exhibit only a single exponential blow-up. As an application, we derive a simple algorithm to translate LTL into deterministic Rabin automata. The algorithm normalizes the formula, translates it into a special very weak alternating automaton, and applies a simple determinization procedure, valid only for these special automata.

\end{abstract}



\begin{document}
\maketitle
\section{Introduction}

In the late 1970s, Amir Pnueli introduced Linear Temporal Logic (LTL) into
computer science as a framework for specifying and verifying concurrent programs \cite{Pnueli77,Pnueli81}, 
a contribution that earned him the 1996 Turing Award.  During the 1980s and the early 1990s, Pnueli et al.\ proceeded to study the properties expressible in LTL. In 1985, Lichtenstein, Pnueli and Zuck introduced a classification of LTL properties \cite{LPZ85}, later described in detail by Manna and Pnueli in two famous monographs \cite{MannaP89,MannaPnueli91}, where also gave it its current name,  the \emph{safety-progress} hierarchy (see also \cite{PitermanP18} for a brief account).  The safety-progress hierarchy consists of a \emph{safety} class of formulas, and five \emph{progress} classes. The classes are defined semantically in terms of their models, and the largest class, called the \emph{reactivity} class in \cite{MannaP89,MannaPnueli91}, contains all properties expressible in LTL.  Manna and Pnueli provide syntactic characterizations of each class. In particular, they state a fundamental theorem showing that every reactivity property is expressible as a conjunction of formulas of the form $\G\F\varphi \vee \F\G\psi$, where $\F \chi$ and $\G \chi$ mean that $\chi$ holds at some and at every point in the future, respectively, and  $\varphi, \psi$ only contain past operators. Manna and Pnueli call this the normal form for Past LTL.

Technically, the works above consider Past LTL, an extension of LTL with past operators.
In 1992, Chang, Manna, and Pnueli presented a different and very elegant characterization of the safety-progress hierarchy in terms of standard LTL, the logic containing only the future operators $\X$ (next), $\U$ (until), and $\W$ (weak until) \cite{ChangMP92}.  They showed that every reactivity formula is equivalent to an LTL formula in negation normal form, such that every path through the syntax tree contains at most one alternation of $\U$ and $\W$. We call this fundamental result the Normalization Theorem. In the notation of \cite{CernaP03,PelanekS05,SickertE20}, which mimics the definition of the $\Sigma_i$, $\Pi_i$, and $\Delta_i$ classes of the arithmetical and polynomial hierarchies, they proved that every LTL formula is equivalent to a $\Delta_2$-formula.

While these normal forms for LTL have had large conceptual impact in model checking, automatic synthesis, and deductive verification (see e.g.\ \cite{PitermanP18} for a recent survey), the normalization \emph{procedures} for LTL have had none. The reason is that many of the known procedures are not direct, meaning that they require to translate formulas into automata and back, all have high complexity, and their correctness proofs are involved.  Let us elaborate on this. 
\begin{itemize}
\item As mentioned above, the normal form for Past LTL is first stated by Lichtenstein, Pnueli, and Zuck in page 208 of \cite{LPZ85}. It is a consequence of Theorem 2, whose proof  is not given in the paper it is only said that the proof is based on many previous results, including results from five different papers \cite{BuechiWeak,MP71,Choueka74,Thomas81,GabbayPSS80}.
\item In three chapters of her PhD thesis~\cite{Zuck86}, Zuck gives a detailed description of the normalization procedure of \cite{LPZ85}. First, Zuck translates the initial Past LTL formula into a counter-free automaton, then applies the Krohn-Rhodes cascade decomposition of finite automata and other results to translate the automaton into a star-free regular expression, and finally translates this expression into a reactivity formula with a non-elementary blow-up. 
\item The normal form for LTL is stated in section 4.2, page 296, of Manna and Pnueli's monograph \cite{DBLP:books/daglib/0077033}, and in  their PODC 1990  survey paper \cite{MannaP89} (page 399 of the proceedings). The proof is said to be out of scope and no normalization procedure is presented.
\item In \cite{ChangMP92}, Chang, Manna and Pnueli state the Normalization Theorem for LTL. They give a rough sketch of the normalization procedure, and do not give a proof. The normalization procedure uses the translation of \cite{LPZ85,Zuck86} from star-free regular expressions to Past LTL as a subroutine, and so it is at least as complex as the one of \cite{LPZ85,Zuck86}. 
\item A normalization procedure for LTL of elementary complexity can be obtained by combining work by Maler and Pnueli on the Krohn-Rhodes decomposition \cite{MP90,MP94,Mal10} with a recent result by Boker, Lehtinen and Sickert on translating automata back into LTL (when possible) \cite{BokerLS22}. The procedure has three steps. First, the  formula is translated into a deterministic $\omega$-regular automaton using e.g.\ Safra's double-exponential construction \cite{Safra88}. Second, this automaton is translated into an equivalent deterministic and counter-free automaton, using the single-exponential construction of \cite{MP90}. Finally, this automaton is translated into a $\Delta_2$-formula using the triple-exponential construction of \cite{BokerLS22}. This combination of constructions yields an elementary normalization procedure, which moreover produces an equivalent formula at the smallest possible level of the safety-progress hierarchy. However, the procedure is indirect and has high complexity. In particular, while future work may improve the blow-up of the last two steps, the double exponential bound on the translation of LTL into deterministic $\omega$-automata is tight. Therefore, any procedure that constructs a deterministic $\omega$-automaton as intermediate step must have at least double-exponential complexity.
\item In \cite{Reynolds00} and \cite{Guelev08}, Reynolds and Guelev give direct proofs of the normalization theorem for Past LTL that do not require to translate formulas into automata. Both proofs rely on different versions of Gabbay's famous separation theorem, stating that every formula of Past LTL is equivalent to a Boolean combination of past and future formulas \cite{Gabbay89,GabbayHR94}. Gabbay's theorem is proved by means of equivalence preserving syntactic transformations, and so are the proofs of 
\cite{Reynolds00,Guelev08}. However, the only known upper bound on the blow-up in the size of the formula produced by Gabbay's separation procedure is non-elementary. Oliveira and Rasga give a double-exponential separation algorithm for a fragment of Past LTL that restricts the nesting of the since and until operators \cite{OR20}.
\end{itemize}

It is remarkable that, despite the prominence of the safety-bounded hierarchy in the work of Manna and Pnueli, the complexity of normalization procedures has not been studied further, even though no lower bound for the blow-up it involves is known. In particular, and contrary to the case of propositional and first-order logic, where efficient normalization algorithms for conjunctive and clausal normal form are essential part of SAT or first-order theorem provers, normalization has not been used in LTL to obtain more efficient algorithms for satisfiability, model-checking, or synthesis tasks.  This paper contains three main results that, in our opinion, completely change this situation:
\begin{enumerate}
\item A simple proof of the Normalization Theorem for LTL. The proof is direct (it does not require any knowledge of automa\-ta or regular expressions), gives a clear intuitive explanation of why normalization is possible, and yields a closed form for the normalized formula with \emph{single-exponential} blow-up. 
\item An efficient normalization algorithm. The normalization procedure given in (1) has exponential \emph{best-case} complexity, and is not goal-oriented, in the sense that it does not only concentrate on those parts of the formula that do not belong to $\Delta_2$. We provide a normalization algorithm consisting of six rewrite rules that solves these problems. In particular, the rewrite rules can be applied locally to ``offending subformulas''.
\item A novel translation of LTL into deterministic Rabin automata (\drw) that exploits the results of (1) and (2). The translation normalizes the formula and then transforms it into an alternating Rabin automaton with at most one alternation between accepting and non-accepting states (\alw{1}). This automaton is determinized by means of a novel, dedicated algorithm  for \alw{1}, with better properties than Safra's construction \cite{Safra88}. In particular, the states of the \drw\ are pairs of sets of states of the \alw{1}, instead of trees of sets of states, as would be the case with Safra's construction. This simpler state structure leads to smaller \drw.
\end{enumerate}

At the heart of our first result is a novel technique, interesting in its own,  called \emph{contextual normalization}. Loosely speaking, in order to normalize a formula $\varphi$ interpreted on infinite words over some given alphabet $\Sigma$, the technique constructs a finite cover of the set $\Sigma^\omega$. This is achieved by carefully selecting a set of formulas $B$, called a \emph{basis}. The cover contains one set of words $\langn{C}{B}$ for every $C \subseteq B$, defined as the set of the words satisfying all formulas of $C$ and none of $B \setminus C$.  For every $C \subseteq B$, we find a formula $\flatn{\varphi}{\con}{B}$ of $\Delta_2$ equivalent to $\varphi$ over all words of $\langn{C}{B}$ (that is, every word of  $\langn{C}{B}$ satisfies either both formulas or none). Intuitively, $\flatn{\varphi}{\con}{B}$ is equivalent to $\varphi$ in the context of $C$. Then we ``patch together'' these formulas to obtain a formula equivalent to $\varphi$ on all words.

The second result shows that LTL formulas can also be normalized by means of a rewrite system, just as one brings a Boolean formula in CNF; the only difference is the need for \emph{contextual} rewrite rules, specifying that a subformula of a formula can only be rewritten if the formula has a certain form.

The third result shows that, on top of the central role it plays in the work of Manna and Pnueli, normalization can lead to novel algorithms for questions in the theory and applications of LTL that continue to be investigated today. In particular, efficient translations from LTL to automata exhibiting different degrees of nondeterminism are being intensely studied, due to their applications to controller synthesis and probabilistic model checking, among others (see e.g.\ \cite{EsparzaKS20,CasaresCF21,CasaresDMRS22,DuretLutzRCRAS22,JohnJBK21,JohnJBK22}). The translation of (3) based on normalization has already become part of the \textsc{Owl} library for $\omega$-automata \cite{KretinskyMS18} and it is used in the \textsc{Strix} synthesis tool \cite{MeyerSL18}. 

The paper is structured as follows. Section \ref{sec:prelims} introduces the syntax and semantics of LTL. Section \ref{sec:hierarchy} introduces Manna and Pnueli's safety-progress hierarchy---following the notation of Cern{\'{a}} and
 Pel{\'{a}}nek \cite{CernaP03}---and recalls the Normalization Theorem presented in \cite{LPZ85,MannaP89,ChangMP92}.
 Section \ref{sec:firstproof} presents our novel proof of the theorem based on contextual equivalence, and in particular Theorem \ref{thm:normthm} on page \pageref{thm:normthm}, our first normalization procedure. Section \ref{sec:main} describes a normalizing rewrite system for LTL consisting of the rules presented in Tables \ref{tab:allrules} and \ref{tab:allrulesRM} in pages \pageref{tab:allrules} and \pageref{tab:allrulesRM}, respectively. Section \ref{sec:LTLtoDRW} introduces a translation from LTL to deterministic Rabin automata based on normalization, summarized in  \Cref{thm:mainLTLDRW} in page \pageref{thm:mainLTLDRW}. Section \ref{sec:hierarchy-automata} shows a tight correspondence between the classes of the safety-progress hierarchy and weak alternating automata. Section \ref{sec:concl} contains some conclusions.

\paragraph{Remark.} This paper is a revised and extended version of previous work by the authors published in \cite{SickertE20,henzingerarticle}. Results  (1) and (3) above were first presented in \cite{SickertE20}, and result (2) in \cite{henzingerarticle}.

 \section{Preliminaries}
\label{sec:prelims}

Let $\Sigma$ be a finite alphabet. A \emph{word} $w$ over $\Sigma$ is an infinite sequence of letters $a_0 a_1 a_2 \dots$ with $a_i \in \Sigma$ for all $i \geq 0$. A \emph{finite word} is a finite sequence of letters. The set of all words (finite words) is denoted $\Sigma^\omega$ ($\Sigma^*$). We let $w[i]$ (starting at $i=0$) denote the $i$-th letter of a word $w$. The finite infix $w[i]w[i+1]\dots w[j - 1]$ is abbreviated with $w_{ij}$ and the infinite suffix $w[i] w[i+1] \dots$ with $w_{i}$. We denote the infinite repetition of a finite word $a_0 \dots a_n$ by $(a_0 \dots a_n)^\omega = a_0 \dots a_n a_0 \dots a_n a_0 \dots$. A set of (finite or infinite) words is called a language.

\subsection{Syntax and semantics of Linear Temporal Logic}
\label{sec:syse}

Formulas of Linear Temporal Logic (LTL) over a finite set $Ap$ of atomic propositions are constructed by the following syntax:
\begin{align}
\label{syntax}
\varphi \Coloneqq \; & \true \mid \false \mid a \mid \neg \varphi \mid \varphi \wedge \varphi \mid \varphi\vee\varphi 
                      \mid \X\varphi \mid \varphi\U\varphi \mid \varphi\W\varphi \mid \varphi\R\varphi \mid \varphi\M\varphi 
\end{align}
\noindent where $a \in Ap$ is an atomic proposition and $\X$, $\U$, $\W$, $\R$, and $\M$ 
are the next, (strong) until, weak until, (weak) release, and strong release operators, respectively. Further, we use the standard abbreviations 
$\F \varphi \coloneqq \true \, \U \, \varphi$ (eventually) and $\G \varphi \coloneqq \varphi \, \W \, \false$ (always). The \emph{size} of a formula is the number of nodes of its syntax tree.

Formulas are interpreted on words over the alphabet $\Sigma \coloneq 2^{Ap}$. Let $w$ be such a word  and let $\varphi$ be a formula. The satisfaction relation $w \models \varphi$ is inductively defined as the smallest relation satisfying:
{\arraycolsep=1.8pt\[\begin{array}[t]{lclclcl}
w \models \true               & &  \mbox{ for every $w$ } & &                                                   \\
w \not \models \false         & &  \mbox{ for every $w$ }                                                      \\
w \models a & \mbox{ if{}f }    & a \in w[0]                                                     \\
w \models  \neg \varphi   & \mbox{ if{}f } & w \not\models \varphi                                   \\
w \models \varphi \wedge \psi & \mbox{ if{}f } & w \models \varphi \text{ and } w \models \psi   \\
w \models \varphi \vee \psi   & \mbox{ if{}f } & w \models \varphi \text{ or } w \models \psi    \\
\end{array} \qquad
\begin{array}[t]{lcl}
w \models \X \varphi      & \mbox{ if{}f } & w_1 \models \varphi \\
w \models \varphi \U \psi & \mbox{ if{}f } & \exists k. \, w_k \models \psi \text{ and } \forall j < k. \, w_j \models \varphi \\
w \models \varphi \M \psi & \mbox{ if{}f } & \exists k. \, w_k \models \varphi \text{ and } \forall j \leq k. \, w_j \models \psi \\
w \models \varphi \R \psi & \mbox{ if{}f } & \forall k. \, w_k \models \psi \text{ or } w \models \varphi\M \psi \\
w \models \varphi \W \psi & \mbox{ if{}f } & \forall k. \, w_k \models \varphi \text{ or } w \models \varphi\U \psi
\end{array}\]}We let $\lang(\varphi) \coloneqq \{ w \in \Sigma^\omega : w \models \varphi\}$ denote the language of $\varphi$.
Two formulas $\varphi$ and $\psi$ are \emph{equivalent}, denoted $\varphi \equiv \psi$, if $\lang(\varphi) = \lang(\psi)$.
We overload the definition of $\models$ and write $\varphi \models \psi$ as a shorthand for $\lang(\varphi) \subseteq \lang(\psi)$.

\paragraph{Monotonic and dual operators.} An operator $\P$  of arity $k \geq 0$ is  \emph{monotonic} if $(\varphi_1 \models \psi_1) \wedge \cdots \wedge (\varphi_k \models \psi_k)$ implies
$\P(\varphi_1, \ldots, \varphi_k) \models \P(\psi_1, \ldots, \psi_k)$. It is easy to see that all operators of the syntax, with the exception of negation, are monotonic.  Two operators $\P$ and $\Q$ of arity $k \geq 0$ are \emph{dual} if $\P(\varphi_1, \ldots, \varphi_k) \equiv \neg{\Q(\neg{\varphi}_1, \ldots, \neg{\varphi}_k)}$ and $\Q(\varphi_1, \ldots, \varphi_k) \equiv \neg{\P(\neg{\varphi}_1, \ldots, \neg{\varphi}_k)}$ for all formulas $\varphi_1, \ldots, \varphi_k$. It is easy to see that $\true$ and $\false$, $\vee$ and $\wedge$, $\U$ and $\R$, and $\W$ and $\M$ are dual, and $\X$ is self dual. 

\paragraph{Negation normal form.} A formula is in \emph{negation normal form} if negations appear only in front of atomic propositions. In other words, the syntax of formulas for negation normal form is obtained by substituting $\neg \varphi$ for $\neg a$ in (\ref{syntax}).
 
The following result is folklore:
\begin{proposition}
\label{prop:nnf}
Every formula of size $n$ has an equivalent formula in negation normal form of size $O(n)$. Further, for every formula $\varphi$ in negation normal form of size $n$ there exists a formula $\overline{\varphi}$ in negation normal form of size $n$ such that $\neg \varphi \equiv \overline{\varphi}$. 
\end{proposition}
\begin{proof}
Formulas are put in negation normal form by ``pushing negations'' across all operators using the duality relations, e.g.\ $\neg(\varphi_1 \U \varphi_2)$ is replaced by $\neg \varphi_1 \R \neg \varphi_2$. These transformations only increase the size of the formula by one unit, and so the formula in negation normal form has size $O(n)$. The formula $\overline{\varphi}$ is defined inductively using the duality relations, e.g.\ 
$\overline{\varphi_1 \U \varphi_2} \coloneq \overline{\varphi_1} \R \overline{\varphi_2}$, $\overline{\X\varphi} \coloneq \X \overline{\varphi}$, etc.; these transformations preserve the size of the formula. 
\end{proof}

\begin{quote}
{\bf Convention}: In the rest of the paper we assume, without explicit mention, that formulas are in negation normal form. Abusing language, we call $\overline{\varphi}$ the negation of $\varphi$.
\end{quote}

 \section{The Safety-Progress Hierarchy}
\label{sec:hierarchy}

We recall the definition of the safety-progress hierarchy, the hierarchy of temporal properties studied by Manna and Pnueli \cite{MannaP89}. We follow the  formulation of {\v{C}}ern{\'{a}} and Pel{\'{a}}nek \cite{CernaP03}. The definition of the hierarchy formalizes the intuition that e.g.\ a safety property is
violated by an execution  if{}f one of its finite prefixes is ``bad'' or, equivalently, satisfied by an execution if{}f all its finite prefixes belong to a language of good prefixes.  In the ensuing sections we describe structures that have a direct correspondence to this hierarchy and in this sense the hierarchy provides a map to navigate the results of this paper.

\begin{definition}[\cite{MannaP89,CernaP03}]Let $P \subseteq \Sigma^\omega$ be a property over $\Sigma$.
\begin{itemize}
\item $P$ is a safety property if there exists a language of finite words $L \subseteq \Sigma^*$ such that $w \in P$ if{}f all finite prefixes of $w$ belong to $L$.
\item $P$ is a guarantee property if there exists a language of finite words $L \subseteq \Sigma^*$ such that $w \in P$ if{}f there exists a finite prefix of $w$ which belongs to $L$.
\item $P$ is an obligation property if it can be expressed as a positive Boolean combination of safety and guarantee properties.
\item $P$ is a recurrence property if there exists a language of finite words $L \subseteq \Sigma^*$ such that $w \in P$ if{}f infinitely many prefixes of $w$ belong to $L$.
\item $P$ is a persistence property if there exists a language of finite words $L \subseteq \Sigma^*$ such that $w \in P$ if{}f all but finitely many prefixes of $w$ belong to $L$.
\item $P$ is a reactivity property if $P$ can be expressed as a positive Boolean combination of recurrence and persistence properties.
\end{itemize}
\end{definition}

The inclusions between these classes are shown in \Cref{fig:temporal_hierarchy}. 
Chang, Manna, and Pnueli give in \cite{ChangMP92} a syntactic characterization of the classes in terms of the following fragments of LTL:

\begin{definition}[Adapted from \cite{CernaP03}]
\label{def:future_hierarchy}
We define the following classes of LTL formulas:
\begin{itemize}
	\item The class $\Sigma_0 = \Pi_0 = \Delta_0$ is the least set of formulas containing $\true$, $\false$, all atomic propositions and their negations, and is closed under the application of conjunction and disjunction.
	\item The class $\Sigma_{i+1}$ is the least set of formulas containing $\Pi_i$ that is closed under the application of conjunction, disjunction, and the $\X$, $\U$, and $\M$ operators.
	\item The class $\Pi_{i+1}$ is the least set of formulas containing $\Sigma_i$ that is closed under the application of conjunction, disjunction, and the $\X$, $\R$, and $\W$ operators.
	\item The class $\Delta_{i+1}$ is the least set of formulas containing $\Sigma_{i+1}$ and $\Pi_{i+1}$ that is closed under the application of conjunction and disjunction.
\end{itemize}
\end{definition}

Observe the behavior of the classes under negation.  Given a set of formulas $F$, let $\overline{F} = \{\overline{\varphi} \mid \varphi \in F \}$. By the definition of $\overline{\phi}$ we have:
\begin{proposition}
\label{prop:dualityhierarchy}
For every $i \geq 0$ we have $\overline{\Sigma}_i=\Pi_i$, $\overline{\Pi}_i = \Sigma_i$, and $\overline{\Delta}_i = \Delta_i$.
\end{proposition}
In particular, \Cref{prop:dualityhierarchy} shows that a formula $\varphi$ is equivalent to a formula of $\Delta_2$ if{}f $\overline{\varphi}$ is. 

The following result, a corollary of the proof of \cite[Thm. 8]{ChangMP92}, shows that the safety-progress hierarchy and the syntactic hierarchy of Definition \ref{def:future_hierarchy} coincide:

\begin{theorem}[Adapted from \cite{CernaP03}]\label{thm:hierarchy:correspondence}
A property that is specifiable in LTL is a guarantee (safety, obligation, persistence, recurrence, reactivity, respectively) property if and only if it is specifiable by a formula from the class $\Sigma_1$, $(\Pi_1$, $\Delta_1$, $\Sigma_2$, $\Pi_2$, $\Delta_2$, respectively$).$
\end{theorem}

\begin{figure}
\begin{subfigure}[c]{0.51\columnwidth}
  \begin{center}
  \small
	\begin{tikzpicture}[x=1cm,y=0.75cm,outer sep=2pt]

    \node (padding1) at ( 0,0.3) {};
	\node (1) at ( 0, 0) {reactivity};
	\node (2) at ( 1,-1) {recurrence};
	\node (3) at (-1,-1) {persistence};
    \node (4) at ( 0,-2) {obligation};
	\node (5) at ( 1,-3) {safety};
	\node (6) at (-1,-3) {guarantee};
	\node (padding2) at ( 0,-3.2) {};

	\path
	(2) edge[draw=none] node[sloped]{$\supset$} (1)
    (3) edge[draw=none] node[sloped]{$\subset$} (1)
    
    (4) edge[draw=none] node[sloped]{$\subset$} (2)
    (4) edge[draw=none] node[sloped]{$\supset$} (3)
    
    (5) edge[draw=none] node[sloped]{$\supset$} (4)
    (6) edge[draw=none] node[sloped]{$\subset$} (4);

	\end{tikzpicture}
  \end{center}
\subcaption{Safety-progress hierarchy \cite{MannaP89}}
\label{fig:temporal_hierarchy}
\end{subfigure}\begin{subfigure}[c]{0.51\columnwidth}
  \begin{center}
  \small
	\begin{tikzpicture}[x=1cm,y=0.75cm,outer sep=2pt]

    \node (padding1) at ( 0,0.3) {};
	\node (1) at ( 0, 0) {$\Delta_2$};
	\node (2) at ( 1,-1) {$\Pi_2$};
	\node (3) at (-1,-1) {$\Sigma_2$};
    \node (4) at ( 0,-2) {$\Delta_1$};
	\node (5) at ( 1,-3) {$\Pi_1$};
	\node (6) at (-1,-3) {$\Sigma_1$};
    \node (padding2) at ( 0,-3.2) {};
    
    \path
	(2) edge[draw=none] node[sloped]{$\supset$} (1)
    (3) edge[draw=none] node[sloped]{$\subset$} (1)
    
    (4) edge[draw=none] node[sloped]{$\subset$} (2)
    (4) edge[draw=none] node[sloped]{$\supset$} (3)
    
    (5) edge[draw=none] node[sloped]{$\supset$} (4)
    (6) edge[draw=none] node[sloped]{$\subset$} (4);
	\end{tikzpicture}
 \end{center}
\subcaption{Syntactic-future hierarchy}
\label{fig:syntactic-future}
\end{subfigure}
\caption{Both hierarchies, side-by-side, indicating the correspondence of \Cref{thm:hierarchy:correspondence}} for properties specifiable in LTL.
\label{fig:hierarchies}
\end{figure}
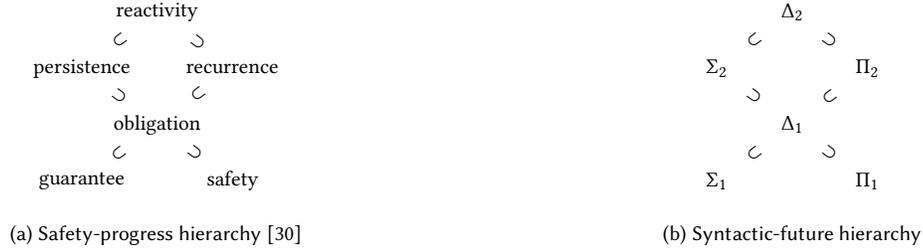

Together with the result of \cite{LPZ85}, stating that every formula of LTL is equivalent to a reactivity formula, Chang, Manna, and Pnueli obtain:

\begin{theorem}[Normalization Theorem \cite{LPZ85,MannaP89,ChangMP92}]
Every LTL formula is equivalent to a formula of $\Delta_2$.
\end{theorem}

\section{A simple proof of the Normalization Theorem} \label{sec:firstproof}

We present a simple proof of the Normalization Theorem. \Cref{subsec:overview} introduces our proof technique, the \emph{Contextual Equivalence Lemma}.
The lemma shows that every formula $\varphi$ is normalizable if certain formulas related to $\varphi$ exist. \Cref{subsec:cnormFGGF,subsec:cnormforall} prove the existence of those formulas. Finally, \cref{subsec:fullnorm} instantiates the Contextual Equivalence Lemma with the formulas of 
\Cref{subsec:cnormFGGF,subsec:cnormforall}, which yield the theorem.

\subsection{Contextual equivalence}
\label{subsec:overview}

Consider the formula $\varphi = \F\G (a \U b)$. It does not belong to $\Delta_2$, because of the alternation $\F$-$\G$, followed by the alternation $\G$-$\U$. Let us see how to find an equivalent formula of $\Delta_2$. We consider the formula $\G\F b$, and argument (informally!) as follows:

\begin{itemize}
\item For words that do not satisfy $\G\F b$, $\varphi$ is equivalent to $\false$. Indeed, if $w \not\models \G\F b$, then $b$ holds for only finitely many suffixes of $w$, and so  $a \U b$ also holds for finitely many suffixes only. So  $w$ cannot satisfy $\F\G (a \U b)$.
\item For words that satisfy $\G\F b$, $\varphi$ is equivalent to $\F\G (a \W b)$. Indeed, if a word $w$ satisfies $\G\F b$, then every suffix $w_i$ satisfies $\F b$, and so $w_i \models a  \U b$ if{}f $w_i \models a \W b$ holds for every $i \geq 0$. 
\end{itemize}

We say that $\varphi$ is equivalent to $\F\G (a \W b)$ \emph{in the context of $\G\F b$}, and equivalent to $\false$ \emph{in the context of $\overline{\G\F b}$}. 
We get $\F\G (a \U b) \equiv (\G\F b  \wedge \F\G (a \W b)) \vee (\overline{\G\F b} \wedge \false) \equiv (\G\F b  \wedge \F\G (a \W b))$. Due to the elimination of $\U$, this is a formula of $\Delta_2$. 

We can proceed dually with the formula $\varphi = \G\F (a \W b)$. We consider the formula $\F\G a$, and argument as follows:

\begin{itemize}
\item In the context of $\F\G a$, $\varphi$ is equivalent to $\true$. Indeed, if $w \models \F\G a$, then there is a suffix $w_i$ such that $a \W b$ holds for all suffixes of $w_i$, and so  in particular for infinitely many suffixes.
\item In the context of $\overline{\F\G a}$, $\varphi$ is equivalent to $\G\F (a \U b)$. Recall that $a \W b \equiv  a \U b \vee \G a$. If a word $w$ does not satisfy $\F\G a$, then no suffix satisfies $\G a$; so $w$ satisfies $\varphi$ if{}f infinitely many prefixes satisfy $a \U b$. 
\end{itemize}
Now we have $\G\F (a \W b) \equiv (\F\G a  \wedge \true) \vee (\overline{\F\G a} \wedge \G\F (a \U b)) \equiv \F\G a \vee (\overline{\F\G a} \wedge \G\F (a \U b)) \equiv \F\G a \vee (\G\F (\neg a) \wedge \G\F (a \U b)) \equiv \F\G a \vee \G\F (a \U b)$. Due to the elimination of $\W$, this is a formula of $\Delta_2$. 

Our normalization strategy follows this pattern. Given a formula $\varphi$, we define a \emph{basis} of formulas $B$, and consider all \emph{contexts} $\context{\con}{B}$, where $\con \subseteq B$. Intuitively, a context $\context{\con}{B}$ corresponds to a region of the set of all words containing the words satisfying all formulas of $C$ and none of $B \setminus C$, or, equivalently, all dual formulas of the formulas in $B \setminus C$. For every context $\context{\con}{B}$ we find a formula $\flatn{\varphi}{\con}{B}$ of $\Delta_2$ equivalent to $\varphi$ over all words of the corresponding region. Then we ``patch together'' these formulas to obtain a formula of $\Delta_2$ that is equivalent to $\varphi$ everywhere.

We formalize this idea in two steps. First we state and prove the Weak Contextual Equivalence Lemma, which is enough to normalize simple formulas like $\F\G (a \U b)$ and $\G\F (a \W b)$. In the second step we state and prove the Contextual Equivalence Lemma, which presents a technique to normalize arbitrary formulas.

Fix a set of atomic propositions $Ap$, and let $\Univ$ be the set of all words over $2^{Ap}$.

\begin{definition}[Basis and equivalence under context]
\label{def:basis}
Let $B$ be a finite set of formulas over $Ap$, called a \emph{basis}. A \emph{context} is a partition of $B$ into a set $\con \subseteq B$ and the set $B \setminus \con$. We denote a context by $\context{\con}{B}$. The language $\langn{C}{B} \subseteq \Univ$ of $\context{\con}{B}$ is the set of words that satisfy every formula of $\con$ and no formula of $B \setminus \con$. Two formulas $\varphi_1, \varphi_2$ are \emph{equivalent under context $\context{\con}{B}$} if $(w \models \varphi_1 \Leftrightarrow w \models \varphi_2)$ holds for every $w \in \langn{C}{B}$.
\end{definition}

We collect some simple properties for later use:

\begin{lemma}
\label{rem}
We have:
\begin{enumerate}
\item $\langn{\emptyset}{\emptyset} = \Univ$. In particular, two formulas are equivalent if{}f they are equivalent under context $\context{\emptyset}{\emptyset}$.
\item Let $B \subseteq B'$. Then $\langn{\con'}{B'} \subseteq \langn{(\con'\cap B)}{B}$ for every $\con' \subseteq B'$ .\label{rem:lsubcont}
\item If two formulas are equivalent under context $\context{\con}{B}$, then they are also equivalent under any context $\context{\con'}{B'}$ such that $\con = \con' \cap B$ and $B \subseteq B'$.\label{rem:extend}
\item Let $\overline{B}:= \{\overline{\psi} \mid \psi \in B\}$ be the \emph{dual} basis of $B$. For every $C \subseteq B$, define $C_d:= \{\overline{\psi} \mid \psi \in B \setminus C\}$. Two formulas are equivalent under context $\context{\con}{B}$ if{}f they are equivalent under context $\context{C_d}{\overline{B}}$. \label{rem:dualequiv} \end{enumerate}
\end{lemma}
\begin{proof} We proceed as follows:
\begin{enumerate}
\item Follows immediately from the definition.
\item Let $w  \in \langn{\con'}{B'}$. By definition we have $w \models \psi$ for all $\psi \in \con'$ and $w \not\models \psi$ for all $\psi \in B' \setminus \con'$. We have $B \setminus (\con' \cap B) = B \setminus \con' \subseteq B' \setminus \con'$. So $w \models \psi$ for all $\psi \in \con' \cap B$ and $w \not\models \psi$ for all $\psi \in B \setminus (\con' \cap B)$. By definition, $w \in \langn{(\con' \cap B)}{B}$. 
\item Let $\varphi_1, \varphi_2$ be equivalent under $\context{\con}{B}$. We have
$w \models \varphi_1 \Leftrightarrow$ if{}f $w \models \varphi_2$ for every $w \in \langn{\con}{B}$.
By $B \subseteq B'$, $\con = \con' \cap B$, and (2), we get $\langn{\con'}{B'} \subseteq \langn{\con}{B}$, and so $\varphi_1, \varphi_2$ are equivalent under $\context{\con'}{B'}$.
\item It suffices to prove $\langn{C_d}{\overline{B}} = \langn{C}{B}$. This is an easy  consequence of $\overline{\psi} \equiv \neg \psi$. Indeed, for any $w \in \Univ$, we have
\begin{center}
$w \models \psi$ for all $\psi \in C$ 
\; if{}f  \; $w \not\models \overline{\psi}$ for all $\psi \in C$ 
\; if{}f  \;   $w \not\models \overline{\psi}$ for all $\psi \in B \setminus (B \setminus C)$ 
\; if{}f  \;   $w \not\models \overline{\psi}$ for all $\psi \in \overline{B} \setminus C_d$.
\end{center}
Similarly, $w \not\models \psi$ for all $\psi \in B\setminus C $ if{}f $w \models \overline{\psi}$ for all $\psi \in C_d$. Hence, $w \in \langn{C}{B}$ if{}f $w \in \langn{C_d}{\overline{B}}$.
\end{enumerate}
\end{proof}

\subsubsection{Weak Contextual Equivalence Lemma}

Intuitively, the Weak Contextual Equivalence Lemma shows how to patch together formulas $\flatn{\varphi}{\con}{B}$ that are equivalent to $\varphi$ for each context $\context{C}{B}$.

\begin{lemma}[Weak Contextual Equivalence Lemma]
\label{lem:conteq}
Let $\varphi$ be a formula, and let $B$ be a basis of formulas. Assume that for every context $\context{\con}{B}$ there exists a formula $\flatn{\varphi}{\con}{B}$ satisfying 
\begin{itemize}
\item[(i)] $\varphi$ and $\flatn{\varphi}{\con}{B}$ are equivalent under context $\context{\con}{B}$, and 
\item[(ii)] for every $C, C' \subseteq B$: if $C \subseteq C'$ then $\flatn{\varphi}{\con}{B} \models \flatn{\varphi}{C'}{B}$.  
\end{itemize}
Then we have
$$ \varphi \equiv \bigvee_{C \subseteq B} \bigg( \flatn{\varphi}{\con}{B} \wedge \bigwedge_{\psi \in C} \psi \bigg)\ .$$
\end{lemma}
\begin{proof}
Let $w \in \Univ$. Let $C_w$ be the (possibly empty) set of formulas of the basis $B$ satisfied by $w$.

Assume $w \models \varphi$.  We have 
$w \models \bigwedge_{\psi \in C_w} \psi$ and  $w \in \langn{C_w}{B}$ by definition of $C_w$, and so $w \models \flatn{\varphi}{C_w}{B}$ by property (i). So $w \models \flatn{\varphi}{C_w}{B} \wedge \bigwedge_{\psi \in C_w} \psi$,
and thus $w \models \bigvee_{C \subseteq B} \big( \flatn{\varphi}{\con}{B} \wedge \bigwedge_{\psi \in C} \psi \big)$.

Assume $w \models \bigvee_{C \subseteq B} \big( \flatn{\varphi}{\con}{B} \wedge \bigwedge_{\psi \in C} \psi \big)$.  Then there exists a context $C\subseteq B$ such that $w \models \flatn{\varphi}{\con}{B} \wedge \bigwedge_{\psi \in C} \psi$. By the definition of $C_w$ we have $C \subseteq C_w$, and so $w \models \flatn{\varphi}{C_w}{B}$ by property (ii). Since, by property (i),
$\varphi$ and $\flatn{\varphi}{C_w}{B}$ are equivalent under context $\context{C_w}{B}$, we get $w \models \varphi$.
\end{proof}

\begin{example}
Consider the formula $\varphi \coloneqq \F\G ( a \U b)$. It is easy to see that  $B \coloneqq \{ \G\F b \}$, $\flatn{\varphi}{\emptyset}{B} \coloneqq \false$ and $\flatn{\varphi}{B}{B} \coloneqq \F\G(a \W b)$ satisfy the conditions of the lemma. So we have 
$$\varphi \equiv (\flatn{\varphi}{\emptyset}{B} \wedge \true) \vee (\flatn{\varphi}{B}{B} \wedge \G\F b) =(\false \wedge \true) \vee (\F\G(a \W b) \wedge \G\F b) \equiv \F\G(a \W b) \wedge \G\F b \in \Delta_2$$
Consider now $\varphi \coloneqq \G\F ( a \W b)$. Taking $B \coloneqq \{ \F\G a \}$,  $\flatn{\varphi}{\emptyset}{B} \coloneqq \G\F (a \U b)$,
and $\flatn{\varphi}{B}{B} \coloneqq \true$ also satisfies the conditions, and we get
$$\varphi \equiv (\flatn{\varphi}{\emptyset}{B} \wedge \true) \vee (\flatn{\varphi}{B}{B} \wedge \F\G a) =(\G\F (a \U b) \wedge \true) \vee (\true \wedge \F\G a) \equiv \G\F (a \U b) \vee \F\G a \in \Delta_2$$
\end{example}

The Weak Contextual Equivalence Lemma only shows how to normalize formulas under the assumption that \emph{one already knows how to normalize the formulas of the basis $B$}. Indeed, if $B$ contains formulas that are not in $\Delta_2$, then for any $C\subseteq B$ containing such formulas the expression $\bigwedge_{\psi \in C} \psi$ is also not in $\Delta_2$. But how can one normalize these formulas of $B$? The obvious idea is to recursively apply the lemma: find a second basis $B'$ of ``simpler'' formulas, reduce the problem of normalizing $B$ to normalizing $B'$, and iterate until a  basis containing only formulas of $\Delta_2$ is reached. The Contextual Equivalence Lemma follows this idea, but also improves on it; instead of a sequence of bases, it only assumes a unique \emph{well-founded} basis.

\subsubsection{Contextual Equivalence Lemma}

We introduce well-founded bases.

\begin{definition}
Let $B$ be a basis and $(\prec) \in B \times B$ a well-founded order on $B$. We say that $(B, \prec)$ is a \emph{well-founded basis}. Given a basis formula $\psi \in B$, we let $\strict{\psi}$ denote the set $\{ \chi \in B \mid \chi \prec \psi \}$.
\end{definition}

Intuitively, after applying weak contextual equivalence to $\varphi$ with basis $B$, we may need to apply it again to $\psi$. But with which basis? The answer is with the (well-founded) basis $\strict{\psi}$. Well-foundedness guarantees that this process terminates. The Contextual Equivalence Lemma also adds to the assumptions (i) and (ii)  on $\varphi$ similar assumptions (iii) and (iv) on the formulas $\psi \in B$.

\begin{lemma}[Contextual Equivalence Lemma]
\label{lem:conteq-ext}
Let $\varphi$ be a formula and let $(B, \prec)$ be a well-founded basis. Assume that for every $\con \subseteq B$ there exists a formula $\flatn{\varphi}{\con}{B}$ satisfying 
\begin{itemize}
\item[(i)] $\varphi$ and $\flatn{\varphi}{\con}{B}$ are equivalent under context $\context{\con}{B}$, and 
\item[(ii)] for every two sets $\con, \con' \subseteq B$: if $\con \subseteq \con'$ then $\flatn{\varphi}{\con}{B} \models \flatn{\varphi}{C'}{B}$.
\end{itemize}
Assume further that for every basis formula $\psi \in B$ and every $\con \subseteq \strict{\psi}$ there exists a formula 
$\flatnn{\psi}{\con}{\strict{\psi}}$ satisfying
\begin{itemize}
\item[(iii)] $\psi$ and $\flatnn{\psi}{\con}{\strict{\psi}}$ are equivalent under context $\context{\con}{\strict{\psi}}$, and 
\item[(iv)] for every two sets $\con, \con' \subseteq \strict{\psi}$: if $\con \subseteq \con'$ then $\flatnn{\psi}{\con}{\strict{\psi}} \models \flatnn{\psi}{\con'}{\strict{\psi}}$.
\end{itemize}
Then we have:
$$\varphi \equiv \bigvee_{\con \subseteq B} \bigg( \flatn{\varphi}{\con}{B} \wedge \bigwedge_{\psi \in \con} \flatnn{\psi}{\con}{\strict{\psi}} \bigg) \ . $$
\end{lemma}
\begin{proof}
Let $w$ be a word.  We prove $w \models  \varphi$ if{}f $w \models \bigvee_{\con \subseteq B} \big( \flatn{\varphi}{\con}{B} \wedge \bigwedge_{\psi \in C} \flatnn{\psi}{\con}{\strict{\psi}} \big)$. Let $C_w \subseteq B$ be the (possibly empty) set of formulas of $B$ satisfied by $w$.

\noindent ($\Rightarrow$): Assume $w \models \varphi$. Observe that $w \in \langn{C_w}{B}$ by definition of $C_w$. We prove that
$w \models \flatn{\varphi}{\con_w}{B}$ and $w \models \flatnn{\psi}{\con_w}{\strict{\psi}}$ for every $\psi \in C_w$.
\begin{itemize}
\item $w \models \flatn{\varphi}{C_w}{B}$. Since $w \in \langn{C_w}{B}$ and $w \models \varphi$ by hypothesis, we have $w \models \flatn{\varphi}{C_w}{B}$ by property (i).

\item $w \models \flatnn{\psi}{C_w}{\strict{\psi}}$ for every $\psi \in C_w$.  Let $\psi \in C_w$. We have $w \in \langn{C_w}{B} \subseteq \langn{C_w}{\strict{\psi}}$ by \cref{rem}(\ref{rem:lsubcont}) because $\strict{\psi} \subseteq B$. Hence, by property (iii), $w \models \flatnn{\psi}{C_w}{\strict{\psi}}$ if{}f $w \models \psi$. Since $\psi \in C_w$ we have $w \models \psi$, and we are done.
\end{itemize}

\noindent ($\Leftarrow$): Assume there is $C \subseteq B$ such that $w \models \flatn{\varphi}{\con}{B}$ and $w \models \flatnn{\psi}{C}{\strict{\psi}}$ for every $\psi \in C$.
We prove $w \models \varphi$. 

\smallskip \noindent\textit{Claim.} $w \models \psi$ for every $\psi \in C$. \\
\noindent\textit{Proof of the claim.} We proceed by induction on $(C, \prec)$. \\
\noindent \textit{Basis.} $\strict{\psi}=\emptyset$. 
Then $w \models \flatnn{\psi}{\emptyset}{\emptyset}$. By \cref{rem} and property (iii), $\psi$ and $\flatnn{\psi}{\emptyset}{\emptyset}$ are equivalent, and so $w \models \psi$. \\
\noindent \textit{Step.} $\strict{\psi} \neq \emptyset$. 
For every $\chi \in C \cap \strict{\psi}$, we have $\chi \prec \psi$ by definition of $\strict{\psi}$, so $w \models \chi$ by induction hypothesis. Hence, $w \models \bigwedge_{\chi \in C \cap \strict{\psi}} \chi$, so $C \cap \strict{\psi} \subseteq C_w \cap \strict{\psi}$. Since 
$w \models \flatnn{\psi}{C}{\strict{\psi}}$ by assumption, we get $w \models \flatnn{\psi}{C_w}{\strict{\psi}}$
by property (iv). Further, by property (iii) we have $w \models \flatnn{\psi}{C_w}{\strict{\psi}}$ if{}f $w \models \psi$. So $w \models \psi$.

\smallskip
By the claim we have $w \models \bigwedge_{\psi \in C} \psi$ and so $C \subseteq C_w$. It follows $w \models \flatn{\varphi}{C_w}{B}$ by property (ii). By property (i) we have 
$w \models \varphi$ if{}f $w \models \flatn{\varphi}{C_w}{B}$, and so $w \models \varphi$.
\end{proof}

\begin{remark}
The Weak Contextual Equivalence Lemma is a corollary of the Contextual Equivalence Lemma. Indeed, the weak statement is obtained by choosing $\prec$ as the empty order (no two basis formulas are ordered), which is trivially well-founded. This implies $\strict{\psi} = \emptyset$ and, by \cref{rem}, we are forced to take $\flatnn{\psi}{\con}{\strict{\psi}} = \flatnn{\psi}{\con}{\emptyset} = \psi$ for any $\psi \in B$ and context $\con \subseteq \emptyset$.
\end{remark}

In the rest of the section we prove the Normalization Theorem by instantiating the Contextual Equivalence Lemma as follows:
\begin{itemize}
\item In Section \ref{sec:contnorm}, we define a well-founded basis $(B,\prec)$ for a given formula $\varphi$. Loosely speaking, $B$ contains formulas of the form $\G\F \psi$ and $\F\G \psi$  for certain subformulas $\psi$ of $\varphi$. Further, $\prec$ is  induced by the subformula order: given basis formulas $\Op(\psi)$ and $\Op'(\psi')$,
where $\Op, \Op' \in \{\G\F, \F\G \}$, we say $\Op(\psi) \prec \Op'(\psi')$ if $\psi$ is a proper subformula of $\psi'$.
\item In Section \ref{subsec:cnormFGGF}, we define formulas $\flatnn{(\F\G \psi)}{\con}{\strict{\psi}}$ and $\flatnn{(\G\F \psi)}{\con}{\strict{\psi}}$ of $\Delta_2$  for all basis formulas $\F\G \psi, \G\F \psi \in B$, and for every $C \subseteq B$. We prove that these formulas satisfy conditions (iii) and (iv) of the Contextual Equivalence Lemma. 
\item In Section \ref{sec:flatnn}, we define a formula $\flatn{\varphi}{\con}{B} \in \Delta_2$ for every $C \subseteq B$, and prove that it satisfies conditions (i) and (ii) of the Contextual Equivalence Lemma.
\end{itemize}
The Contextual Equivalence Lemma yields that $\varphi$ is equivalent to $\bigvee_{\con \subseteq B} \big( \flatn{\varphi}{\con}{B} \wedge \bigwedge_{\psi \in \con} \flatnn{\psi}{\con}{\strict{\psi}} \big)$, a Boolean combination of formulas of  $\Delta_2$. Since $\Delta_2$ is closed under Boolean combinations, we obtain a \emph{contextual normalization} procedure. 

\subsection{Contextual Normalization I: The well-founded basis $(B,\prec)$}
\label{sec:contnorm}

In our introductory example at the beginning of Section \ref{subsec:overview}, we normalize the formula $\F\G (a \U b)$ with the help of a case distinction: for words satisfying $\G\F b$ the formula is equivalent to $\false$, and for words that do not satisfy $\G\F b$ the formula is equivalent to $\F\G (a \U b)$. This holds because $b$ is the right child of the until operator; indeed, the case distinction with $\G\F a$ does not work. Similarly, we normalize the formula $\G\F (a \W b)$ by a case distinction on $\F\G a$, and we succeed because $a$ is the left child of the weak until. This motivates the following definition of a well-founded basis for a given formula:

\begin{definition}
Let $\varphi$ be a formula and let $\subf(\varphi)$ denote the set of subformulas of $\varphi$.
We define $B \coloneqq \mubasis  \cup \nubasis $, where $\mubasis  := \{ \G\F \psi \mid \exists \chi. ~ \chi \U \psi \in \subf(\varphi) \vee  \psi \M \chi \in \subf(\varphi) \}$
and $\nubasis  := \{ \F\G \psi \mid \exists \chi. ~ \psi \W \chi \in \subf(\varphi) \vee  \chi \R \psi \in \subf(\varphi) \}$.
Further, we define a well-founded partial order $(\prec) \subseteq B \times B$ as follows: for every  $\Op, \Op' \in \{\G\F, \F\G \}$, $\Op(\psi) \prec \Op'(\psi')$ if{}f $\psi \in \subf(\psi')$.
\end{definition}

\begin{example}
\label{ex:running1}

Consider the formula $\varphi = ((a \W b) \U c) \W d$. We have $\mubasis  = \{ \G\F c \}$,  $\nubasis  = \{ \F\G ((a \W b) \U c)), \F\G a \}$ and $B = \{\G\F c, \F\G ((a \W b) \U c), \F\G a\}$. Further, $\G\F c  \prec \F\G ((a \W b) \U c)$ and $\F\G a \prec \F\G ((a \W b) \U c)$. Observe that the basis formula $\F\G ((a \W b) \U c)$ is not in $\Delta_2$.
\end{example}

\subsection{Contextual Normalization II: The formulas $\flatnn{(\F\G \psi)}{\con}{\strict{\psi}}$ and $\flatnn{(\G\F\psi)}{\con}{\strict{\psi}}$}
\label{subsec:cnormFGGF}

For all basis formulas $\F\G \psi, \G\F \psi \in B$ and for all $C \subseteq \strict{\psi}$, we define formulas $\flatnn{(\F\G \psi)}{\con}{\strict{\psi}}$ and $\flatnn{(\G\F \psi)}{\con}{\strict{\psi}}$ and prove that the definitions satisfies conditions (iii) and (iv) of the Contextual Equivalence Lemma. We define $\flatnn{(\F\G \psi)}{\con}{\strict{\psi}}$ and prove its properties in Section \ref{sec:flatnnfg}, and do the same for $\flatnn{(\G\F \psi)}{\con}{\strict{\psi}}$  in Section \ref{sec:flatnn}.

\subsubsection{The formula $\flatnn{\F\G\psi}{\con}{\strict{\psi}}$}
\label{sec:flatnnfg}
Loosely speaking,  we define $\flatnn{(\F\G \psi)}{\con}{\strict{\psi}}$ in two steps: first, we assign to $\psi$ a formula $\eval{\psi}{\con}$ of $\Pi_1$, and then set $\flatnn{(\F\G \psi)}{\con}{\strict{\psi}}:= \F\G(\eval{\psi}{\con})$. Since $\eval{\psi}{\con} \in \Pi_1$, we have $\flatnn{(\F\G \psi)}{\con}{\strict{\psi}} \in \Delta_2$. 

\begin{definition}\label{def:evalnu}
Let $\F\G\psi \in \nubasis$ be a formula, and let $\con \subseteq \strict{\psi}$. Further, let $\eval{\psi}{\con}$ be the formula inductively defined as follows. If $\psi = \true, \false, a, \neg a$ then $\eval{\psi}{\con} := \psi$. For the operators $\vee$, $\wedge$, $\X$, $\W$, and $\R$, the formula $\eval{\psi}{\con}$ is defined homomorphically.\footnote{Given an operator $\P$  of arity $k \geq 0$, we say that $\flat{\psi}{\con}$ is defined \emph{homomorphically} for $\P$ if $\flat{(\P(\psi_1, \ldots, \psi_k))}{\con}$ = $\P(\flat{\psi_1}{\con}, \ldots, \flat{\psi_k}{\con})$. So, for example, $\flat{\X\psi_1}{\con}= \X (\flat{\psi_1}{\con})$ and $\flat{(\psi_1 \wedge \psi_2)}{\con} = \flat{\psi_1}{\con} \wedge \flat{\psi_2}{\con}$} For the operators $\U$ and $\M$ we set
\begin{align*}
\eval{(\psi_1 \U \psi_2)}{\con}  \coloneqq  \begin{cases} \eval{\psi_1}{\con} \W \eval{\psi_2}{\con}\hphantom{\R} & \mbox{if $\G\F \psi_2 \in \con$} \\ \false & \mbox{otherwise,} \end{cases} 
\quad \mbox{and} \quad
\eval{(\psi_1 \M \psi_2)}{\con}  \coloneqq \begin{cases} \eval{\psi_1}{\con} \R \eval{\psi_2}{\con} \hphantom{\W} & \mbox{if $\G\F \psi_1 \in \con$} \\ \false & \mbox{otherwise.}\end{cases} 
\end{align*} We define $\flatnn{(\F\G\psi)}{\con}{\strict{\psi}} \coloneqq \F\G(\eval{\psi}{\con})$. 
\end{definition}

Loosely speaking, in $\eval{\psi}{\con}$ all occurrences of $\U$ in $\psi$ are either changed into $\W$ or removed. Therefore, 
$\eval{\psi}{\con}$ does not contain any occurrence of $\U$ anymore. The same happens with $\M$. So $\eval{\psi}{\con}$ is a formula of $\Pi_1$, and 
$\flatnn{(\F\G\psi)}{\con}{\strict{\psi}}$ is a formula of $\Delta_2$.

\begin{example}
\label{ex:evalnu}
Consider the basis formula $\F\G \psi = \F\G ((a \W b) \U c)$ of \cref{ex:running1}.  We have $\strict{\psi} = \{ \G\F c, \F\G a\}$.
For $\con := \{ \G\F c \}$ and $\con := \{\G\F c, \F\G a\}$ we get $\flatnn{(\F\G\psi)}{\con}{\strict{\psi}} = \F\G ((a \W b) \W c)$. For $\con := \emptyset$ and $\con :=\{ \F\G a \}$
we get $\flatnn{(\F\G\psi)}{\con}{\strict{\psi}} = \F\G \false \equiv \false$.
\end{example}

We prove in \cref{lem:mainlemma} below that  $\flatnn{(\F\G\psi)}{\con}{\strict{\psi}}$ satisfies conditions (iii) and (iv) of the Contextual Equivalence Lemma. First we need a definition and a technical lemma.

\begin{definition}
Let $\con \subseteq \strict{\psi}$ and $w \in \langn{\con}{\strict{\psi}}$. The \emph{stabilization index} of $w$ with respect to $\con$ is the least index $I \geq 0$ such that $w_{I+j} \not\models \chi$ for every $j \geq 0$ and every formula $\G\F\chi \in \strict{\psi} \setminus \con$.
\end{definition}

The stabilization index exists because, by the definition of $\langn{\con}{ \strict{\psi}}$, we have $w \not\models \G\F \chi$ for every $\G\F\chi \in \strict{\psi}\setminus \con$, and so there exists only finitely many suffixes of $w$ that satisfy each $\chi$. Since $\strict{\psi}$ is finite, we can choose the stabilization index as the least index $i$ such that no suffix $w_{i+j}$ satisfies $\chi$ for any $\G\F\chi \in \strict{\psi} \setminus \con$. The following lemma explains the relation between $\psi$ and $\eval{\psi}{\con}$.

\begin{lemma}
\label{lem:technicallem}
Let $\con \subseteq \strict{\psi}$, let $w \in \langn{\con}{\strict{\psi}}$, and let $I$ be the stabilization index of $w$ with respect to $\con$. 
\begin{enumerate}
\item For every $i \geq I \colon w_ {i} \models \psi \Rightarrow w_{i} \models \eval{\psi}{\con} $.
\item For every $i \geq 0 \colon w_{i} \models \eval{\psi}{\con}  \Rightarrow w_ {i} \models \psi$.
\end{enumerate}
\end{lemma}
\begin{proof}
(1):  Fix an $i \geq I$ such that  $w_i \models \psi$.  We prove $w_{i} \models \eval{\psi}{\con}$ by structural induction on $\psi$.  If $\psi= \true, \false, a, \overline{a}$ then $\eval{\psi}{\con}= \psi$, and we are done. For all cases defined homomorphically the result follows immediately from the induction hypothesis. 
Assume now $\psi = \psi_1 \U \psi_2$.  We claim $\G\F\psi_2 \in \con$. Observe first that, by the definition of the basis, we have  $\G\F\psi_2 \in \strict{\psi}$.  Assume that $\G\F\psi_2 \notin \con$. Then, since $w \in \langn{\con}{\strict{\psi}}$ and $i \geq I$, we have $w_{i+j} \not\models \psi_2$ for every $j \geq 0$, which implies $w_i \not\models  \psi_1 \U\psi_2$ by the semantics of LTL. But this contradicts the assumption $w_i \models \psi$, and the claim is proved. Now we proceed as follows:

$$\begin{array}{rcll}
w_{i}\models \psi_1 \U \psi_2 &  \implies & w_{i}\models \psi_1 \W \psi_2 & \text{(semantics of LTL)}\\
& \iff    &   \forall j. \, (w_{i+j} \models \psi_1 \vee \exists k \leq j.\, w_{i+k} \models \psi_2 ) & \text{(semantics of LTL)} \\
& \implies  & \forall j. \, (w_{i+j} \models \eval{\psi_1}{\con} \vee \exists k \leq j.\, w_{i+k} \models \eval{\psi_2}{\con})  & \text{(induction hypothesis)} \\
& \implies  & w_{i} \models (\eval{\psi_1}{\con}) \W (\eval{\psi_2}{\con}) & \text{(semantics of LTL)} \\
& \iff    &   w_{i} \models \eval{(\psi_1 \U \psi_2)}{\con} & \text{(definition of  $\eval{ }{}$ and claim)}
\end{array}$$

Finally, assume $\psi = \psi_1 \M \psi_2$. Using  $\psi_1 \M \psi_2 \equiv \psi_1 \U (\psi_1 \wedge \psi_2)$, 
$\psi_1 \R \psi_2 \equiv \psi_1 \W (\psi_1 \wedge \psi_2)$, and the previous case we get
\begin{align*}
w_{i}\models \psi_1 \M \psi_2  & \iff w_{i}\models \psi_1 \U (\psi_1 \wedge \psi_2) \implies
w_{i} \models \eval{\psi_1}{\con} \W (\eval{\psi_1}{\con} \wedge \eval{\psi_2}{\con}) \\
& \iff w_{i} \models \eval{\psi_1}{\con} \R  \eval{\psi_2}{\con} \iff w_{i} \models \eval{(\psi_1 \M \psi_2)}{\con}
\end{align*}

(2):  Fix an $i \geq 0$ such that $w_i \models \eval{\psi}{\con}$. We prove $w_ {i} \models \psi$
by structural induction on $\psi$. All cases but $\psi = \psi_1 \U \psi_2$ are proved as in (1). Assume now
$\varphi = \psi_1 \U \psi_2$. By the definition of the basis, we have  $\G\F\psi_2 \in \strict{\psi}$. We claim $\G\F \psi_2 \in \con$. Assume the contrary. Then, by the definition of $\eval{}{}$ we have $\eval{\psi}{\con} = \false$, contradicting that $w_i \models \eval{\psi}{\con}$, and the claim is proved. Since $\G\F \psi_2 \in \con$ and $w \in \langn{\con}{\strict{\psi}}$, we have $w \models \G \F \psi_2$.  We show $w_i \models \psi_1 \U \psi_2$:
$$\begin{array}{rcll}
w_i \models \eval{(\psi_1 \U \psi_2)}{\con} & 
\iff     & w_i \models (\eval{\psi_1}{\con}) \W (\eval{\psi_2}{\con}) & \text{(definition of  $\eval{ }{}$ and claim)}\\
& \iff     & \forall j. \, w_{i+j} \models \eval{\psi_1}{\con} \vee \exists k \leq j.\, w_{i+k} \models \eval{\psi_2}{\con} & 
\text{(semantics of LTL)}\\
& \implies & \forall j. \, w_{i+j} \models \psi_1 \vee \exists k \leq j.\, w_{i+k} \models \psi_2 & \text{(induction hypothesis)} \\
& \iff     & w_i \models \psi_1 \W \psi_2 & \text{(semantics of LTL)}\\
& \iff     & w_i \models \psi_1 \U \psi_2  & \text{(because $w_i \models \G\F \psi_2$)}
\end{array}$$

The case $\psi = \psi_1 \M \psi_2$ is proved as in (1).

\end{proof}

\begin{proposition}
\label{lem:mainlemma}
Let $\F\G \psi$ be a basis formula. For every $C \subseteq \strict{\psi}$, the formula $\F\G \left(\eval{\psi}{\con}\right)$ belongs to $\Delta_2$. Further, $\F\G \psi$ and $\F\G \left(\eval{\psi}{\con}\right)$ are equivalent under context $\context{C}{\strict{\psi}}$. Finally, if $\con \subseteq \con'$ then  $\F\G \left(\eval{\psi}{\con}\right) \models \F\G \big(\eval{\psi}{\con'}\big)$.
\end{proposition}
\begin{proof}
Fix  $C \subseteq \strict{\psi}$. We prove the three parts of the proposition separately.

\smallskip\noindent  (1)  The formula $\F\G \left(\eval{\psi}{\con}\right)$ belongs to $\Delta_2$. 
Follows immediately from $\eval{\psi}{\con} \in \Pi_1$.

\smallskip\noindent  (2) $\F\G \psi$ and $\F\G \left(\eval{\psi}{\con}\right)$ are equivalent with respect to $\context{C}{\strict{\psi}}$. If $\langn{\con}{\strict{\psi}} = \emptyset$, both formulas are trivially equivalent, so let $w \in \langn{\con}{\strict{\psi}}$. Let $I$ be the stabilization index of $w$ with respect to $C$. By \cref{lem:technicallem},  for every $i \geq 0$ we have $w_{i+I} \models  \psi$ if{}f $w_{i+I} \models \eval{\psi}{\con}$.  So, in particular, 
$w \models\F\G \psi$ if{}f $w \models \F\G \left(\eval{\psi}{\con}\right)$.

\smallskip\noindent (3) If $\con \subseteq \con'$ then 
$\F\G\left(\eval{\psi}{\con}\right) \models \F\G \left(\eval{\psi}{\con'}\right)$. Assume $\con \subseteq \con'$. By the monotonicity of $\F$ and $\G$ it suffices to prove $\eval{\psi}{\con} \models \eval{\psi}{\con'}$. We proceed by structural induction on $\psi$. The cases  $\psi = \true, \false, a, \overline{a}$ are trivial.
We now consider the case  $\psi = \psi_1 \W \psi_2$, as a representative of the cases where $\eval{\psi}{\con}$ is defined homomorphically. We have $\eval{\psi}{\con} = \eval{\psi_1}{\con} \W \eval{\psi_2}{\con}$  and  $\eval{\psi}{\con'} = \eval{\psi_1}{\con'} \W \eval{\psi_2}{\con'}$ (definition of $\eval{}{\con}$). So it suffices to show $\eval{\psi_1}{\con} \W  \eval{\psi_2}{\con}  \models \eval{\psi_1}{\con'} \W \eval{\psi_2}{\con'}$. But this follows immediately from $\eval{\psi_1}{\con} \models \eval{\psi_1}{\con'}$ and $\eval{\psi_2}{\con} \models \eval{\psi_2}{\con'}$ (induction hypothesis), and the monotonicity of $\W$. Finally we consider the two remaining cases:

\smallskip\noindent Case $\psi = \psi_1 \U \psi_2$.  We have $\G\F \psi_2 \in \strict{\psi}$ (definition of $\strict{\psi}$). If $\G\F\psi_2 \notin \con$ then $\eval{\psi}{\con} = \false$ (definition of $\eval{}{\con}$), which implies $\eval{\psi}{\con} \models \eval{\psi}{\con'}$, and we are done. If $\G\F\psi_2 \in \con$ then $\G\F\psi_2 \in \con'$ ($\con \subseteq \con'$), and so we have $\eval{\psi}{\con} = \eval{\psi_1}{\con} \W \eval{\psi_2}{\con}$  and  $\eval{\psi}{\con'} = \eval{\psi_1}{\con'} \W \eval{\psi_2}{\con'}$ (definition of $\eval{}{\con}$). Proceed now as in the previous case.

\smallskip\noindent Case $\psi = \psi_1 \M \psi_2$. We have $\G\F \psi_1 \in \strict{\psi}$ (definition of $\strict{\psi}$). If $\G\F\psi_1 \notin \con$ then we proceed as for $\varphi = \psi_1 \U \psi_2$. If $\G\F\psi_1 \in \con$ then $\G\F\psi_1 \in \con'$ ($\con \subseteq \con'$). Proceed as in the previous case, replacing $\W$ by $\R$.
\end{proof}

\subsubsection{The formula $\flatnn{\G\F\psi}{\con}{\strict{\psi}}$ }
\label{sec:flatnn}

We define the formula $\flatnn{(\G\F \psi)}{\con}{\strict{\psi}}$ and prove that it satisfies conditions (iii) and (iv) of the Contextual Equivalence Lemma. We use the 
duality between $\F\G$ and $\G\F$. Recall from \cref{sec:syse} that the dual $\overline{\varphi}$ of a formula $\varphi$ in negation normal form replaces $\U$ by $\R$, $\W$ by $\M$, and vice versa. In particular we have $\overline{\F\G \varphi} = \G\F \, \overline{\varphi}$ and $\overline{\G\F \varphi} = \F\G \, \overline{\varphi}$, $\overline{\Pi_1} = \Sigma_1$, and $\overline{\Delta_2} = \Delta_2$ by \Cref{prop:dualityhierarchy}. Moreover, by \cref{rem}(\ref{rem:dualequiv}) two formulas are equivalent under context $\context{C_d}{\overline{B}}$ if{}f they are equivalent under the dual context $\context{C}{B}$. This is all we need.

\begin{definition}\label{def:evalmu}
Let $\G\F\psi \in B$ and let $\con \subseteq \strict{\psi}$. 
We define $\evalgf{\psi}{C} \coloneqq \overline{\eval{\overline{\psi}}{C_d}}$ and $\flatnn{(\G\F\psi)}{\con}{\strict{\psi}} \coloneqq \G\F(\evalgf{\psi}{C})$.
\end{definition}

Since $\eval{\overline{\psi}}{C_d} \in \Pi_1$, we have $\evalgf{\psi}{\con} \in \overline{\Pi_1} = \Sigma_1$. We now prove the counterpart of \cref{lem:mainlemma}:

\begin{proposition}
\label{lem:mainlemma-nu}
Let $\G\F \,\psi$ be a basis formula. For every $\con \subseteq \strict{\psi}$, the formula $\G\F \left(\evalgf{\psi}{\con}\right)$ belongs to $\Delta_2$. Further, $\G\F\, \psi$ and $\G\F \left(\evalgf{\psi}{\con}\right)$ are equivalent under context $\context{C}{\strict{\psi}}$. Finally, if $\con \subseteq \con'$ then  $\G\F \left(\evalgf{\psi}{\con}\right) \models \G\F \big(\evalgf{\psi}{\con'}\big)$.
\end{proposition}

\begin{proof}
Each of the three statements follows by duality from \cref{lem:mainlemma}. Given any formula $\G\F\psi$ and any two sets $C, C' \subseteq \strict{\psi}$, we instantiate \cref{lem:mainlemma} with $\overline{\G\F\psi}$, $C_d$, and $C_d'$. For the first statement, we have $\G\F (\evalgf{\psi}{\con}) = \overline{\F\G (\eval{\overline{\psi}}{C_d})} \in \overline{\Delta_2} = \Delta_2$. Let us now prove that $\G\F\, \psi$ and $\G\F \left(\evalgf{\psi}{\con}\right)$ are equivalent under context $\context{C}{\strict{\psi}}$. 
We have  $\G\F \psi = \overline{\F\G \overline{\psi}}$ and $\G\F (\evalgf{\psi}{\con}) = \overline{\F\G (\eval{\overline{\psi}}{C_d})}$ by definition. 
So it suffices to show that $\F\G \overline{\psi}$ and $\F\G (\eval{\overline{\psi}}{C_d})$ are equivalent under $\context{C}{\strict{\psi}}$, and, by
\cref{rem}(\ref{rem:dualequiv}), that they are equivalent under context $\context{C_d}{\strict{\overline{\psi}}}$. But this follows from \cref{lem:mainlemma-nu}.
Finally, let $C \subseteq C'$. We show $\G\F \left(\evalgf{\psi}{\con}\right) \models \G\F \big(\evalgf{\psi}{\con'}\big)$. Since $C \subseteq C'$,
we have $C_d' \subseteq C_d$ and so $\neg \, \GF (\evalgf{\psi}{C_d'}) \equiv \FG (\eval{\overline{\psi}}{C'_d}) \models \FG (\eval{\overline{\psi}}{C_d}) \equiv \overline{\GF (\evalgf{\psi}{C_d})} \equiv \neg \G\F(\evalgf{\psi}{C_d})$, where the $\models$-step applies \cref{lem:mainlemma}. By \emph{modus tollens}, we are done.
\end{proof}

\subsection{Contextual normalization III: The formula $\flatn{\varphi}{\con}{B}$}
\label{subsec:cnormforall}

We search for a formula $\flatn{\varphi}{\con}{B}$ that satisfies conditions (i) and (ii) of the Contextual Equivalence Lemma, and belongs to $ \Delta_2$.
\cref{lem:flatten:correct-local} gives such a formula for the case in which $\varphi=\G \psi$ for some formula $\psi$. Let us first explain the idea. By the equivalence $\G \psi \equiv \psi \U \G\psi$, given a word $w$ and $C \subseteq B$, the following is true for \emph{every} index $j \geq 0$:  $w \models \psi \U \G\psi$ if{}f  $w_i \models \psi$ for every $i  \leq j$ and $w_j \models \G\psi$. Observe that the form of this statement is already close to condition (i) of the Contextual Equivalence Lemma.  The final trick is to choose $j$ as the stabilization index $I$ of $w$ with respect to $C$. Applying \Cref{lem:technicallem} we obtain that $\G\psi$ is equivalent to $\psi \U \G \left(\eval{\psi}{\con}\right)$ under
context $\context{\con}{B}$, and also to $\left(\G\left(\eval{\psi}{\con}\right)\right)\R\psi$.

\begin{lemma}\label{lem:flatten:correct-local}
If $\varphi = \G\psi$, then for every $\con \subseteq B$ the formulas $\G\psi$,
$\psi \U \G\left(\eval{\psi}{\con}\right)$, and  $\G\left(\eval{\psi}{\con}\right)\R\psi$ are equivalent under context $\context{\con}{B}$.
\end{lemma}

\begin{proof}
By \cref{rem}(\ref{rem:extend}), it suffices to show that they are equivalent under context $\context{\con}{\mubasis}$ since $\mubasis \subseteq B$. 
We prove $w \models \G\psi$ if{}f $w \models \psi \U \G \left(\eval{\psi}{\con}\right)$ 
if{}f $w \models \G\left(\eval{\psi}{\con}\right)\R\psi$ for every $w \in \langn{\con}{\mubasis}$. So fix an arbitrary $w \in \langn{\con}{\mubasis }$.

\smallskip\noindent 
($\Rightarrow$)  Assume $w \models \G \psi$. Then, in particular, we have 
$w \models \F\G \psi$. By \cref{lem:technicallem}, $(w_i \models \psi \Leftrightarrow 
\eval{\psi}{\con})$ holds for every $i$ beyond the stabilization index of $w$ with respect to $C$. So, in particular, we have $w \models \F\G \left(\eval{\psi}{\con}\right)$.
It follows that $w \models \G\psi \wedge \F\G \left(\eval{\psi}{\con}\right)$ which, by the semantics of LTL, implies 
$w  \models \psi \U \G \left(\eval{\psi}{\con}\right)$ and $w  \models \G \left(\eval{\psi}{\con}\right) \R \psi$. 

\smallskip\noindent
($\Leftarrow$) Assume $w \models \psi \U \G\left(\eval{\psi}{\con}\right)$. By the semantics of LTL there exists $i \geq 0$ such that $w_j \models \psi$ for every $j < i$ and $w_j \models \eval{\psi}{\con}$ for every $j \geq i$. By \cref{lem:technicallem}, since $w \in \langn{\con}{\mubasis }$, we have $w_j \models \psi$ for every $j \geq i$. So $w_j \models \psi$ holds for every $j \geq 0$, and so $w \models \G\psi$. For $\G\left(\eval{\psi}{\con}\right)\R\psi \equiv \G\left(\eval{\psi}{\con}\right)\M\psi \vee \G\psi$, either $\G\psi$ holds and we are already done, or we can apply the previous argument with only $j \leq i$ instead of $j < i$.
\end{proof}

It is easy to extend \Cref{lem:flatten:correct-local} to all $\W$- and $\R$-formulas.

\begin{proposition}\label{prop:flatten:correct-local}
Let $\psi_1$ and $\psi_2$ be formulas.
\begin{enumerate}
\item 
For every $\con \subseteq \mubasis$, the formulas $\psi_1 \W \psi_2$ and $\psi_1  \U \left(\psi_2 \vee \G \left(\eval{\psi_1}{\con}\right)\right)$ are equivalent under context $\con$.
\item For every $\con \subseteq \mubasis$, the formulas $\psi_1 \R \psi_2 $ and $\left(\psi_1 \vee \G \left(\eval{\psi_2}{\con}\right)\right) \M  \psi_2$ are equivalent under context $\con$.
\end{enumerate}
\end{proposition}
\begin{proof}
We only prove the first case. If $\langn{\con}{\mubasis } = \emptyset$, the proposition trivially holds. Let $w \in \langn{\con}{\mubasis }$. 

\smallskip\noindent 
($\Rightarrow$) Assume $w \models \psi_1 \W \psi_2$, and consider two cases. If $w \models \G \psi_1$, then 
$w \models \psi_1 \U \G(\eval{\psi_1}{\con})$ (\cref{lem:flatten:correct-local}), and so
$w \models \psi_1 \; \U  \; \left(\psi_2 \vee \G \left(\eval{\psi_1}{\con}\right)\right)$.
If $w \not \models \G \psi_1$, then we have $w \models \psi_1 \W \psi_2 $ if{}f $w \models \psi_1 \U \psi_2$ by the semantics of LTL, and 
then $w \models \psi_1 \U \psi_2$ implies $w \models \psi_1 \U \left(\psi_2 \vee \G \left(\eval{\psi_1}{\con}\right)\right)$, also by the semantics of LTL.

\smallskip\noindent 
($\Leftarrow$)  Assume $w \models \psi_1 \U \left(\psi_2 \vee \G\left(\eval{\psi_1}{\con}\right)\right)$. Then
$w \models \psi_1 \U (\psi_2 \vee \G \psi_1)$ (\Cref{lem:flatten:correct-local}, after distributing $\U$ over $\vee$), and since $\psi_1 \U (\psi_2 \vee \G \psi_1) \equiv \psi_1 \W \psi_2$ we are done.
\end{proof}

\Cref{prop:flatten:correct-local} suggests the following definition for the formula $\flatten{\varphi}{\con}$:

\begin{definition}\label{def:evalmunu}
For every $\con \subseteq B$, the formula $\flatten{\varphi}{\con}$ is inductively defined as follows. If $\varphi = \true, \false, a, \overline{a}$ then $\flatten{\varphi}{\con} = \varphi$. For the operators $\vee$, $\wedge$, $\X$, $\U$, and $\M$, we define $\flatten{\varphi}{\con}$ homomorphically. Finally
\begin{align*}
\flatten{(\psi_1 \W \psi_2)}{\con}  \coloneq \flatten{\psi_1}{\con} \U \big( \flatten{\psi_2}{\con} \vee \G \left(\eval{\psi_1}{\con}\right)\big)  
\quad \mbox{ and } \quad
\flatten{(\psi_1 \R \psi_2)}{\con}  = \left(\flatten{\psi_1}{\con} \vee \G \big(\eval{\psi_2}{\con}\big)\right) \M \flatten{\psi_2}{\con} \ .
\end{align*} \end{definition}

\begin{example}
\label{ex:norm2}
Let $\varphi = ((a \W b) \U c) \W d$ and $\con \subseteq B$. With $\eval{((a \W b) \U c)}{\con} = (a \W b) \W c $ we get:
\begin{align*}
\flatten{\varphi}{\con} & =  \flatten{\big((a \W b) \U c\big)}{\con} \, \U \, \bigg(  d \vee  \G \big((a \W b) \W c\big) \bigg) \\
& =  \big( \flatten{(a \W b)}{\con} \U  c \big)  \, \U \, \bigg(  d \vee  \G \big((a \W b) \W c\big) \bigg) \\
& =  \bigg( \big( a \U ( b \vee \G a) \big)  \U  c \bigg)  \, \U \, \bigg(  d \vee  \G \big((a \W b) \W c\big) \bigg) 
\end{align*}
\end{example}

We prove that setting $\flatn{\varphi}{\con}{B} := \flatten{\varphi}{\con}$ satisfies all conditions of the Contextual Equivalence Lemma.

\begin{proposition}
\label{prop:maingeneral}
Let $\varphi$ be a formula.
For every $\con \subseteq B$, the formula $\flatten{\varphi}{\con}$ belongs to $\Sigma_2$. Further, $\varphi$ and $\flatten{\varphi}{\con}$ are equivalent under context $\context{\con}{B}$. Finally, if $\con \subseteq \con'$ then $\flatten{\varphi}{\con} \models \flatten{\varphi}{\con'}$.
\end{proposition}
\begin{proof}

\smallskip\noindent (1) The formula $\flatten{\varphi}{\con}$ belongs to $\Sigma_2$. We show it by structural induction on $\varphi$. The only non-trivial cases are $\varphi = \varphi_1\W \varphi_2$ and $\varphi = \varphi_1\R \varphi_2$. For either of the two cases, by induction hypothesis, we have $\flatten{\varphi_k}{\con} \in \Sigma_2$ for $k=1,2$. Further, since $\eval{\varphi_k}{\con}\in \Pi_1$, we have $\G \left(\eval{\varphi_k}{\con}\right)\in \Pi_1$ as well. So $\flatten{\varphi_1}{\con} \U \big( \flatten{\varphi_2}{\con} \vee \G \left(\eval{\varphi_1}{\con}\right)\big) \in \Sigma_2$ and $\left(\flatten{\varphi_1}{\con} \vee \G \big(\eval{\varphi_2}{\con}\big)\right) \M \flatten{\varphi_2}{\con} \in \Sigma_2$, and we are done.

\smallskip\noindent (2) The formulas $\varphi$ and $\flatten{\varphi}{\con}$ are equivalent under context $\context{\con}{B}$. By \cref{rem}(\ref{rem:extend}), it suffices to show 
that $\left( w \models \varphi \iff w \models \flatten{\varphi}{\con} \right)$ holds for every $C \subseteq \mubasis $ and for every $w  \in \langn{\con}{\mubasis }$. We proceed by structural induction on $\varphi$.  We use the  identity
\begin{align}
\flatten{\psi}{\con} = \flatten{\psi}{\con \cap \mubasis} \label{flatten:restrict}
\end{align}
\noindent which holds because the case distinctions in \cref{def:evalnu} only involve formulas $\G\F\chi$ such that $\chi$ is a subformula of $\psi$, and those formulas belong to $\mubasis$.

The base of the induction is $\varphi \in \{ \true, \false, a, \neg a\}$. In all these cases we have $\varphi = \flatten{\varphi}{\con}$ by definition, and we are done. All cases in which $\flatten{\varphi}{\con}$ is defined homomorphically are handled in the same way, and so we consider only the case $\varphi = \psi_1 \U \psi_2$. Fix $\con \subseteq \mubasis$ and $w \in \langn{\con}{\mubasis}$. We prove $\left( w \models \varphi \iff w \models \flatten{\varphi}{\con} \right)$. Applying the induction hypothesis to $\psi_1$ and $\psi_2$, $\psi_j$ and $\flatten{\psi_j}{\con}$ are equivalent under context $\context{\con}{\mubasis}$. Moreover, $w_i \in \langn{C}{\mubasis}$ for all $i \geq 0$ by the limit properties of $\G\F$-subformulas, and this yields:
\begin{align}
\forall i. &  (w_i \models \psi_j \iff w_i \models \flatten{\psi_j}{\con}) & \mbox{ for $j=1,2$.} \label{eq:wi}
\end{align}

\noindent We get
\[\arraycolsep=4.1pt
\begin{array}{lllll}
w \models \psi_1 \U \psi_2 & \iff & \exists k. ~ w_k \models \psi_2 \wedge (\forall \ell < k. ~ w_\ell \models \psi_1)  & \mbox{(semantics of LTL)}\\
&\iff & \exists k. ~ w_k \models \flatten{\psi_2}{\con} \wedge (\forall \ell < k. ~ w_\ell \models \flatten{\psi_1}{\con})  & \mbox{(\ref{eq:wi})}\\ 
&\iff & w \models \flatten{\psi_1}{\con} \U \flatten{\psi_2}{\con} & \mbox{(semantics of LTL)} \\
&\iff & w \models \flatten{(\psi_1 \U \psi_2)}{\con} & \mbox{(definition of $\flatten{}{}$)}
\end{array}\]
\noindent which concludes the proof. 

The remaining cases are $\varphi = \psi_1 \W \psi_2$ and $\varphi = \psi_1 \R \psi_2$. Fix $\con \subseteq \mubasis $ and $w \in \langn{\con}{\mubasis }$. By induction hypothesis, (\ref{eq:wi}) holds. We derive:
\[\def\arraystretch{1.2}
\begin{array}{rlll}
w \models \psi_1 \W \psi_2  &\iff       & w \models \psi_1 \U (\psi_2 \vee \G (\eval{\psi_1}{\con})) & \text{(\Cref{prop:flatten:correct-local})} \\
&\iff       & w \models \; \flatten{\psi_1}{\con} \U (\flatten{\psi_2}{\con} \vee \G(\eval{\psi_1}{\con})) & \text{(\ref{eq:wi})} \\
&\iff   & w \models \; \flatten{(\psi_1 \W \psi_2)}{\con}  & \mbox{(definition of $\flatten{}{}$)}
\end{array}\]
The case $\varphi = \psi_1 \R \psi_2$ is analogous.

\smallskip\noindent (3) If $\con \subseteq \con'$ then $\flatten{\varphi}{\con} \models \flatten{\varphi}{\con'}$. Follows from a straightforward structural induction, using $\eval{\varphi}{\con} \models \eval{\varphi}{\con'}$, and the monotonicity of all operators (recall that formulas are in negation normal form).
\end{proof}

\subsection{Contextual normalization IV: The Normalization Theorem}
\label{subsec:fullnorm}

We insert Propositions \ref{lem:mainlemma}, \ref{lem:mainlemma-nu}, and  \ref{prop:maingeneral} into the Contextual Equivalence Lemma.
This yields a closed expression for a formula of $\Delta_2$ equivalent to a given formula. The normalized formula is exponentially larger than the original one.

\begin{theorem}[Normalization Theorem]
\label{thm:normthm}
Let $\varphi$ be a formula of length $n$. We have
$$ \varphi \equiv \bigvee_{\substack{M \subseteq \mubasis  \\ N \subseteq \nubasis }} \bigg( \flatten{\varphi}{M} \wedge \bigwedge_{\F\G \psi \in N} \F\G (\eval{\psi}{M}) \wedge \bigwedge_{\G\F \psi \in M} \G\F (\evalgf{\psi}{N})  \bigg)\ .$$
Moreover, the right-hand side formula is in $\Delta_2$ and has length $2^{O(n)}$.
\end{theorem}
\begin{proof}
By the Contextual Equivalence Lemma, we have
$$\varphi \equiv \bigvee_{\con \subseteq B} \bigg( \flatn{\varphi}{\con}{B} \wedge \bigwedge_{\psi \in \con} \flatnn{\psi}{\con}{\strict{\psi}} \bigg)$$
for any well-founded basis $(B, \prec)$, formulas $\flatn{\varphi}{\con}{B}$ satisfying conditions (i)-(ii), and formulas $\flatnn{\psi}{\con}{\strict{\psi}}$ satisfying 
conditions (iii)-(iv) of the lemma.

We choose $(B, \prec)$ as in Definition \ref{def:basis}. We have $B = \mubasis  \uplus \nubasis $. For every $\con \subseteq B$, we split $\con = M \uplus N$,
where $M := \con \cap \mubasis $ and $N := \con \cap \nubasis $. For every $\F\G \psi \in N$ (recall that $N$ only contains $\F\G$-formulas) and $\con \subseteq \strict{\psi}$, we define $\flatnn{(\F\G\psi)}{\con}{\strict{\psi}} := \F\G ( \eval{\psi}{M} )$,
and for every $\G\F \psi \in M$ and $\con \subseteq \strict{\psi}$, we define $\flatnn{(\G\F\psi)}{\con}{\strict{\psi}} := \G\F ( \evalgf{\psi}{N} )$, as in Definitions \ref{def:evalnu} and \ref{def:evalmu}. By Propositions \ref{lem:mainlemma} and \ref{lem:mainlemma-nu},  $\flatnn{\psi}{\con}{\strict{\psi}}$ satisfies conditions (iii)-(iv) of the Contextual Equivalence Lemma for every $\psi \in B$ and $\con \subseteq \strict{\psi}$. Finally, for every $\con \subseteq B$
we define $\flatn{\varphi}{\con}{B} := \flatten{\varphi}{\con}$, as in \cref{def:evalmunu}. By \cref{prop:maingeneral}, $\flatn{\varphi}{\con}{B}$ satisfies conditions (i)-(ii) of the  
Contextual Equivalence Lemma. Applying the lemma we obtain:

\begin{align*}
\varphi \equiv \bigvee_{\con \subseteq B} \bigg( \flatn{\varphi}{\con}{B} \wedge \bigwedge_{\psi \in \con} \flatnn{\psi}{\con}{\strict{\psi}} \bigg)\;
& \equiv \bigvee_{\substack{M \subseteq \mubasis  \\ N \subseteq \nubasis }}   \bigg( \flatn{\varphi}{\con}{B} \wedge \bigwedge_{\psi \in N} \flatnn{\psi}{\con}{\strict{\psi}} \wedge \bigwedge_{\psi \in M} \flatnn{\psi}{\con}{\strict{\psi}}\bigg) \\
& \equiv \bigvee_{\substack{M \subseteq \mubasis  \\ N \subseteq \nubasis }}   \bigg( \flatten{\varphi}{M} \wedge \bigwedge_{\F\G \psi \in N} \F\G (\eval{\psi}{M}) \wedge \bigwedge_{\G\F \psi \in M} \G\F (\evalgf{\psi}{N})\bigg) 
\end{align*}

By Propositions \ref{lem:mainlemma}, \ref{lem:mainlemma-nu}, and  \ref{prop:maingeneral}, the right-hand-side is a Boolean combination of formulas of $\Delta_2$, and so in $\Delta_2$ itself. For the size,
observe first that $B$ contains at most $n$ formulas, and so the right-hand-side is a disjunction of at most $2^{n}$ formulas. We bound the length of a disjunct. First we observe that $|\flatten{\varphi}{C}| \leq O(n^2)$, which follows easily from \Cref{def:evalmunu}. (For the cases defined homomorphically, the proof is a direct application of the induction hypothesis; for the case $\varphi = \psi_1 \W \psi_2$ we have  $|\flatten{\varphi}{C}| \leq |\flatten{\psi_1}{C}| + |\flatten{\psi_2}{C}| + |\eval{\psi_2}{C}| +2$, and the result follows from $|\eval{\psi_k}{C}| \leq n$; the case $\varphi = \psi_1 \R \psi_2$ is analogous.) From the definition of $\eval{\varphi}{\con}$ we immediately obtain $|\F\G\left(\eval{\varphi}{\con}\right)| \leq n$, and, since for every formula $\psi$ we have $|\psi| = |\overline{\psi}|$, the same holds for $\F\G\left(\eval{\varphi}{\con}\right)$. So each disjunct has length at most $O(n^2)+ n \cdot n \in O(n^2)$. So we obtain $2^n \cdot O(n^2) \in 2^{O(n)}$ for the final bound.
\end{proof}

\begin{example}
Consider the formula $\varphi = ((a \W b) \U c) \W d$ of \cref{ex:running1}. We have $\mubasis  = \{ \G\F c \}$ and $\nubasis  = \{ \F\G ((a \W b) \U c)), \F\G a \}$. 
Let us compute the disjunct of the right-hand-side of \cref{thm:normthm} for $M := \{\G\F c\}$ and $N: = \{\F\G ((a \W b) \U c)), \F\G a\}$. We have:
\begin{itemize}
\item $\flatten{\varphi}{M} = ((a \U (b \vee \G a)) \U  c) \U (d \vee  \G ((a \W b) \W c)) \in \Delta_2$, as shown in \cref{ex:norm2}.
\item $\F\G (\eval{((a \W b) \U c)}{M}) = \F\G ((a \W b) \W c) \in \Delta_2$, as shown in \cref{ex:evalnu}, and $\F\G (\eval{a}{M}) = \F\G a$.
\item $\G\F (\evalgf{c}{N}) = \G\F c$. 
\end{itemize}
So the disjunct is  $((a \U (b \vee \G a)) \U  c) \U (d \vee  \G ((a \W b) \W c)) \wedge \F\G ((a \W b) \W c)) \wedge \F\G \, a \wedge \G\F \, c$.
\end{example}

We conclude with the following lemma that bounds the number of formulas in the normal form of~\cref{thm:normthm}. It will be used later in \cref{subsec:det} to bound the size of some automata.

\begin{lemma}
\label{lem:subfbound}
For any formula $\varphi$, $\setmu \subseteq \mubasis$ and $\setnu \subseteq \nubasis$, let $\varphi_{\setmu,\setnu}$ be the corresponding clause of the disjunction in \cref{thm:normthm}. Then $|\subf(\varphi_{\setmu,\setnu}) | \in O(|\subf(\varphi)|)$.
\end{lemma}

\begin{proof}
This follows from the following claims:

\begin{enumerate}
	\item $|\bigcup \{\subf(\eval{\psi}{\setmu}) : \psi \in \subf(\varphi)\} | \leq |\subf(\varphi)|$ and $|\bigcup \{\subf(\evalgf{\psi}{\setnu}) : \psi \in \subf(\varphi)\} | \leq |\subf(\varphi)|$ \label{lem:size:sf:c1}
	\item $|\subf(\flatten{\varphi}{\setmu})| \leq 4 |\subf(\varphi)|$ \label{lem:size:sf:c2}
\end{enumerate}

To prove (\ref{lem:size:sf:c1}), observe that $\subf(\eval{\psi}{\setmu}) \subseteq \{ \eval{\psi'}{\setmu} \mid \psi' \in \subf(\psi) \}$. This can be proven by a straightforward induction, which is trivial for the cases where $\eval{\psi}{\setmu}$ is defined homomorphically. For $\psi = \psi_1 \U \psi_2$, according to \cref{def:evalnu}, either $\eval{\psi}{\setmu} = \false$ or $\eval{\psi}{\setmu} = \eval{\psi_1}{\setmu} \W \eval{\psi_2}{\setmu}$. Since $\eval{\psi}{\setmu}$ is in the right-hand set by definition with $\psi' = \psi$, it remains to consider the subformulas of $\eval{\psi_k}{\setmu}$ for $k=1,2$. However, by induction hypothesis, for any $\psi'' \in \subf(\eval{\psi_k}{\setmu})$ there is a $\psi' \in \subf(\psi_k) \subseteq \subf(\psi)$ such that $\psi'' = \eval{\psi'}{\setmu}$, and we are done. Using the just proven inclusion, we have
\[ {\textstyle \bigcup} \{\subf(\eval{\psi}{\setmu}) : \psi \in \subf(\varphi)\} \subseteq \{ \eval{\psi'}{\setmu} : \psi' \in \subf(\psi), \psi \in \subf(\varphi) \} = \{ \eval{\psi}{\setmu} : \psi \in \subf(\varphi) \} \]
and the cardinality of the latter set is clearly no more than $|\subf(\varphi)|$. The case of $\evalgf{\psi}{\setnu}$ follows by duality.

For (\ref{lem:size:sf:c2}), we prove the stronger claim $|\subf(\flatten{\psi}{\setmu}) \cup \subf(\eval{\psi}{\setmu})| \leq 4 |\subf(\psi)|$ by induction on $\psi$. The base cases and those in which both $\flatten{\psi}{\setmu}$ and $\eval{\psi}{\setmu}$ are defined homomophically are straightforward. Let us show only the cases of $\U$ and $\W$, since $\M$ and $\R$ are completely analogous. For $\psi = \psi_1 \U \psi_2$, we have $\flatten{\psi}{\setmu} = \flatten{\psi_1}{\setmu} \U \flatten{\psi_2}{\setmu}$ and either $\eval{\psi}{\setmu} = \false$ or $\eval{\psi}{\setmu} = \eval{\psi_1}{\setmu} \W \eval{\psi_2}{\setmu}$. In the first case, by induction hypothesis, $|\subf(\flatten{\psi_k}{\setmu})| \leq |\subf(\flatten{\psi_k}{\setmu}) \cup \subf(\eval{\psi_k}{\setmu})| \leq 4 |\subf(\psi_k)|$ for $k = 1,2$. The remaining subformulas in the left-hand side are only two, $\false$ and $\flatten{\psi}{\setmu}$ itself, so they are at most $4 |\subf(\psi)|$. In the second case,
\begin{equation}\label{eq:subf}
\subf(\flatten{\psi}{\setmu}) \cup \subf(\eval{\psi}{\setmu}) = (\subf(\flatten{\psi_1}{\setmu}) \cup \subf(\eval{\psi_1}{\setmu})) \cup (\subf(\flatten{\psi_2}{\setmu}) \cup \subf(\eval{\psi_2}{\setmu})) \cup \{\flatten{\psi}{\setmu}, \eval{\psi}{\setmu}\},
\end{equation}
so we have the bound $4 |\subf(\psi_1)| + 4 |\subf(\psi_2)| + 2 \leq 4 |\subf(\psi)|$ by induction hypothesis. For $\psi = \psi_1 \W \psi_2$, we have $\flatten{\psi}{\setmu} = \flatten{\psi_1}{\setmu} \U (\flatten{\psi_2}{\setmu} \vee \G (\eval{\psi_1}{\setmu}))$ and $\eval{\psi}{\setmu} = \eval{\psi_1}{\setmu} \W \eval{\psi_2}{\setmu}$. The calculation (\ref{eq:subf}) also holds with the last set replaced by $\{\flatten{\psi}{\setmu}, \eval{\psi}{\setmu}, \flatten{\psi_2}{\setmu} \vee \G (\eval{\psi_1}{\setmu}), \G (\eval{\psi_1}{\setmu}) \}$, so we obtain the bound $4 |\subf(\psi_1)| + 4 |\subf(\psi_2)| + 4 \leq 4 |\subf(\psi)|$.
\end{proof}

\subsection{Evaluation}
\label{subsec:eval}

Our proof of the Normalization Theorem gives a good explanation of \emph{why} the result holds. Roughly speaking, given a formula $\varphi$, the set of all words can be partitioned into contexts, such that within each context the formula is equivalent to a formula of $\Sigma_2$ derived from $\varphi$ by means of a very simple syntactic transformation. The procedure also provides a single exponential asymptotic upper bound for the length of the equivalent $\Delta_2$-formula.
However, the proof does not yield a normalization procedure efficient in practice. Indeed, a procedure based on
\cref{thm:normthm} has exponential \emph{best}-case complexity, since it requires to iterate over all subsets of $\mubasis $ and $\nubasis $. Also, the procedure works top-down, and so it is particularly bad for families of formulas where, loosely speaking, the double alternation occurs near the bottom of the syntax tree.

\begin{example} \label{example:exp}
Consider the family of formulas 
$$\varphi_n =  (\cdots ((((a_0 \U a_1) \W a_2) \U a_3) \U a_4) \cdots \U a_n)$$
\noindent for $n \geq 3$. The sets $\mubasis$ and $\nubasis$ have size $n-1$ and 1, respectively. The normalization procedure based on the Normalization Theorem yields a disjunction of $2^{n+1}$ formulas, and so it takes exponential time in $n$. However, reasoning as at the beginning of Section \ref{subsec:overview}, when we proved  $\F\G (a \U b) \equiv  (\G\F b  \wedge \F\G (a \W b))$, one can show that $\varphi_n$ is equivalent to a $\Delta_2$-formula of length $\Theta(n)$:
$$\begin{array}{rcl}
\varphi_n & \equiv & \bigg(\G\F a_1 \wedge   (\cdots (((a_0 \U a_1) \U ( a_2 \vee \G(a_0 \vee a_1))) \U a_3)  \cdots \U a_n\bigg) \; \vee \; \bigg(\cdots (((a_0 \U a_1) \U a_2) \U a_3)  \cdots \U a_n\bigg)
\end{array}$$
Intuitively, in order to normalize $\varphi_n$ it suffices to solve the ``local'' problem caused by  the subformula $((a_0 \U a_1) \W a_2) \U a_3$ of $\varphi_n$, which is in $\Sigma_3$; however, the procedure of \cref{sec:firstproof} is blind to this  fact, and generates $2^{n+1}$ formulas.
\end{example}

 \section{A Normalizing Rewrite System} \label{sec:main}

We present a system of rewrite rules that allow us to normalize every LTL formula. As a corollary, we obtain an alternative proof of the Normalization Theorem. The normalization algorithms derived from the rewriting rules are more efficient than the ones presented in the previous section. 

We first introduce a rewrite system for the fragment of the syntax without the operators $\R$ and $\M$. 
The reasons are purely expository. The restriction to this fragment makes the proofs shorter, and the extension to a rewrite system for the full syntax, presented in \cref{sec:extensions}, is routine. 

\begin{remark}
Recall also that the fragment without the operators $\R$ and $\M$ is as expressive as the full syntax. Indeed, exhaustively applying the well-known equivalences
$$  \varphi_1 \R \varphi_2 \equiv \varphi_1 \W (\varphi_1 \wedge \varphi_2) \quad \text{and} \quad  
\varphi_1 \M \varphi_2 \equiv \varphi_2 \U (\varphi_1 \wedge \varphi_2) $$
\noindent to a formula yields an equivalent formula of the fragment. However, in the worst case the formula of the fragment may be exponentially larger than the original one, and so eliminating the $\R$ and $\M$ operators and then applying the rewrite system for the fragment may be less efficient than a direct application of the rewrite system for the full syntax given in \cref{sec:extensions}.
\end{remark}

The key idea leading to the rewrite system is to treat the combinations $\G\F$ (infinitely often) and $\F\G$ (almost always) of temporal operators as \emph{atomic} operators $\GF$ and $\FG$ (notice the typesetting with the two letters touching each other). We call them the \emph{limit operators}; intuitively, whether a word satisfies a formula $\GF \varphi$ or $\FG \varphi$ depends only on its behavior ``in the limit'', in the sense that $w'w$ satisfies $\GF \varphi$ or $\FG \varphi$ if{}f $w$ does.\footnote{Limit operators are called \emph{suspendable} in \cite{babiak13}.}
So we add the limit operators to the fragment, yielding the following syntax:

\begin{definition}
Extended LTL formulas over a set $Ap$ of atomic propositions are generated by the syntax:
\begin{align*}
\varphi \Coloneqq \; & \true \mid \false \mid a \mid \neg a \mid \varphi \wedge \varphi \mid \varphi\vee\varphi
                       \mid \X\varphi \mid \varphi\U\varphi \mid \varphi\W\varphi \mid \GF \varphi \mid \FG \varphi
\end{align*}
\end{definition}

\noindent When determining the class of a formula in the syntactic future hierarchy, $\GF$ and $\FG$  are implicitly replaced by $\G\F$ and $\F\hspace{0.02cm}\G$. For example, $\F \GF a$ is rewritten into $\F\G\F a$, and so it is a formula of $\Sigma_3$.

\begin{quote}
\underline{Convention}: In the rest of the section we only consider extended formulas, which are by construction in negation normal form, and call them just formulas. 
\end{quote}

Let us now define our precise target normal form.  Formulas of the form $\varphi \U \psi$, $\varphi \W \psi$, $\X\varphi$, $\GF\varphi$, and $\FG\varphi$ are called  \U-, \W-, \X-, \GF-, and \FG-formulas, respectively.
We refer to these formulas as \emph{temporal} formulas. 
The syntax tree $T_\varphi$ of a formula $\varphi$ is defined in the usual way, and $\nnodes{\varphi}$ denotes the number of nodes of $T_\varphi$.
A node of $T_\varphi$ is a \emph{\U-node} if the subformula rooted at it is a \U-formula. \W-, \GF-, \FG- and temporal nodes are defined analogously. 

\newcommand\normalFormConds{\begin{enumerate}
\item No \U-node is under a \W-node.
\item No limit node is under another temporal node.
\item No \W-node is under a \GF-node, and no \U-node is under a \FG-node.
\end{enumerate}}

\begin{definition}
\label{def:normalform}
Let $\varphi$ be an LTL formula. 
A node of $T_\varphi$ is a \emph{limit node} if it is either a \GF-node or a \FG-node. The formula
$\varphi$ is in \emph{normal form} if $T_\varphi$ satisfies the following properties:
\normalFormConds
\end{definition}

\begin{restatable}{remark}{normalFormCharacterization}
Observe that formulas in normal form belong to $\Delta_2$. Even a slightly stronger statement holds: a formula in normal form is a positive Boolean combination of formulas of $\Sigma_2$ and formulas of the form $\GF \psi$ such that $\psi \in \Sigma_1$ (and so $\GF \psi \in \Pi_2$). 

\smallskip

There is a dual normal form in which property (1) is replaced by ``no \W-node is under a \U-node'', and the other two properties do not change. Formulas in dual normal form are positive Boolean combination of formulas of $\Pi_2$ and formulas of the form $\FG \psi$ such that $\psi \in \Pi_1$. Once the Normalization Theorem for the primal normal form is proved, a corresponding theorem for the dual form follows as an easy corollary (see Section \ref{sec:extensions}).
\end{restatable}

We incrementally normalize formulas by dealing with the three requirements of the normal form, one by one,
as follows:

\begin{enumerate}
	\item We remove all \U-nodes that are under some \W-node, but not under any  limit node. 
	\item We remove all limit nodes under some other temporal node. 
	\item We remove all \W-nodes under some \GF-node, and all \U-nodes under some \FG-node.
\end{enumerate}
The resulting formula is in normal form (\Cref{def:normalform}).

It is convenient to introduce intermediate normal forms corresponding to the results of applying the first and the first two steps. 
For this, let us introduce the following parameters of a formula $\varphi$:
\begin{itemize}
\item $\ubw{\varphi}$ denotes the number of \U-nodes in $T_\varphi$ that are under some \W-node, but not under any limit node of $T_\varphi$. For example,  if $\varphi = (a \U b) \W (\FG (c \U d))$ then $\ubw{\varphi}=1$.
\item $\gfba{\varphi}$ denotes  the number of distinct limit subformulas under some temporal operator. Formally,
$\gfba{\varphi}$ is the number of limit formulas $\psi'$ such that $\psi'$ is a proper subformula of a temporal subformula (proper or not) of $\varphi$. For example, if $\varphi = (\FG a \, \U \, \GF b) \vee (\GF b \, \W \, \FG a)$ then $\gfba{\varphi}=2$.
\end{itemize}
\noindent We define the normal forms as follows:

\begin{definition}
\label{def:1stnormalform}
An LTL formula $\varphi$ is in \emph{\fnf} if $\ubw{\varphi} = 0$, and in \emph{\snf} if $\ubw{\varphi} = 0$ and  
$\gfba{\varphi} = 0$.
\end{definition}

\noindent Observe that after step (1) we obtain a formula in \fnf, and after steps (1) and (2) a formula in \snf.

The rest of the section is structured as follows. Sections \ref{subsec:stage1}-\ref{subsec:stage3} present rewrite rules to conduct steps (1)-(3), respectively. Section \ref{sec:overview-rs} proves the Normalization Theorem, and gives an upper bound in the size of the final formula. Section \ref{sec:summary} summarizes the normalization procedure. Finally, Section \ref{sec:extensions} discusses some extensions and special cases.

\subsection{Stage 1: Removing \U-nodes under \W-nodes.}
\label{subsec:stage1}

We consider formulas $\varphi$ with placeholders, i.e., ``holes'' that can be filled with a formula. 
Formally, let $\hole$ be a symbol denoting a special atomic proposition. 
A formula with placeholders is a formula with one or more occurrences of $\hole$, all of them positive (i.e., the formula has  no occurrence of $\neg\hole$. We denote by $\varphi[\psi]$ the result of filling each placeholder of $\varphi$ with an occurrence of $\psi$; formally, $\varphi[\psi]$ is the result of substituting $\psi$ for $\hole$ 
in $\varphi$. For example, if $\varphi\hole = (\hole \W (a \U \hole))$, then $\varphi[\X b] = (\X b) \W (a \U \X b)$. We assume that $\hole$ binds more strongly than any operator, e.g.\ $\varphi_1 \W \varphi_2[\psi] = \varphi_1 \W (\varphi_2[\psi])$.
We write $\varphi \equiv^w \psi$ as a shorthand for $w_k \models \varphi$ if{}f $w_k \models \psi$ for all $k \in \N$, the following lemma states two straightforward properties of formulas with placeholders:

\begin{lemma} \label{lemma:placeholder}
	For every formula with placeholder $\varphi$ and every two formulas $\psi$ and $\psi'$,
\begin{enumerate}
	\item If $\varphi$ is in negative normal form (as every LTL formula considered in this article), $\psi \models \psi'$ implies $\varphi[\psi] \models \varphi[\psi']$. \label{lemma:nnf}
	\item $\psi \equiv^w \psi'$ implies $\varphi[\psi] \equiv^w \varphi[\psi']$. \label{lemma:replace}
\end{enumerate}
\end{lemma}

\smallskip
The following lemma allows us to pull \U-subformulas out of \W-formulas:

\begin{lemma}
\label{lem:UWandWU}
\begin{eqnarray}
\varphi_1 \W \varphi_2[\psi_1 \U \psi_2] & \equiv & (\varphi_1 \U \varphi_2[\psi_1 \U \psi_2]) \vee \G \varphi_1 \label{eqWU} \\
\varphi_1[\psi_1 \U \psi_2] \W \varphi_2 & \equiv & (\GF \psi_2 \wedge \varphi_1[\psi_1 \W \psi_2] \W \varphi_2) \vee \varphi_1[\psi_1 \U \psi_2] \U (\varphi_2 \vee (\G \varphi_1[\false])) \label{eqUW}
\end{eqnarray}
\end{lemma}

\begin{proof}

For Equation (\ref{eqWU}) observe that, by the definition of the semantics of LTL, $\varphi_1 \W \varphi_2 \equiv  \varphi_1 \U \varphi_2 \vee \G \varphi_2$ holds for arbitrary formulas $\varphi_1, \varphi_2$. For Equation (\ref{eqUW}) we proceed by invoking the Weak Contextual Equivalence Lemma (Lemma~\ref{lem:conteq}) with basis $B = \{\GF \psi_2\}$, main formula $\varphi = \varphi_1[\psi_1 \U \psi_2] \W \varphi_2$, and contextual formulas $\flatn{\varphi}{\con}{B} = \varphi_1[\psi_1 \W \psi_2] \W \varphi_2$ for $C = \{\GF \psi_2\}$ and $\flatn{\varphi}{\con}{B} = \varphi_1[\psi_1 \U \psi_2] \U (\varphi_2 \vee (\G \varphi_1[\false]))$ for  $C = \emptyset$. The conclusion of that lemma is exactly (\ref{eqUW}), but we have to prove its two premises:
{\renewcommand\theenumi{\roman{enumi}}

\smallskip\noindent \textbf{First premise}:  $\varphi$ and $\flatn{\varphi}{\con}{B}$ are equivalent under context $\context{\con}{B}$. 
	
For the case $C =  \{ \GF \psi_2 \}$, let $w \models \GF \psi_2$. We have to show $\varphi_1[\psi_1 \U \psi_2] \W \varphi_2  \equiv^w \varphi_1[\psi_1 \W \psi_2] \W \varphi_2$. By the semantics of LTL, we have $\psi_1 \U \psi_2 \equiv \psi_1 \W \psi_2 \wedge \F \psi_2$; further, $w \models \GF \psi_2$ implies  $\F \psi_2 \equiv^w \true$, and so we get $\psi_1 \U \psi_2 \equiv^w \psi_1 \W \psi_2$. Applying \cref{lemma:placeholder}(\ref{lemma:replace}) we obtain $\varphi_1[\psi_1 \U \psi_2] \W \varphi_2 \equiv^w \varphi_1[\psi_1 \W \psi_2] \W \varphi_2$, as desired.

\newcommand{\sledom}{\Relbar\joinrel\mathrel{|}}

For the case $C =  \emptyset$, let $w \not\models \GF \psi_2$. We have to show $\varphi_1[\psi_1 \U \psi_2] \W \varphi_2  \equiv^w \varphi_1[\psi_1 \U \psi_2] \U (\varphi_2 \vee (\G \varphi_1[\false]))$. For $\sledom^w$,  we prove the stronger $\sledom$:
\[\begin{array}{lr}
		\varphi_1[\psi_1 \U \psi_2] \U (\varphi_2 \vee (\G \varphi_1[\false])) \\
		\kern2em \models \varphi_1[\psi_1 \U \psi_2] \U (\varphi_2 \vee (\G \varphi_1[\psi_1 \U \psi_2]) &\text{(\cref{lemma:placeholder}(\ref{lemma:nnf}) and $\false \models \psi_1 \U \psi_2$)} \\
		\kern2em \equiv \varphi_1[\psi_1 \U \psi_2] \W \varphi_2 &\text{(semantics of LTL)}
	\end{array}\]

\noindent For $\models^w$, assume  $w \models \varphi_1[\psi_1 \U \psi_2] \W \varphi_2$. We have to show $w \models \varphi_1[\psi_1 \U \psi_2] \U (\varphi_2 \vee (\G \varphi_1[\false]))$. Consider two cases:

\begin{itemize}
	\item $w_k \models \varphi_2$ for some $k \geq 0$. Then, since $w \models \varphi_1[\psi_1 \U \psi_2] \W \varphi_2$ by hypothesis, we have $w \models \varphi_1[\psi_1 \U \psi_2] \U \varphi_2 $, and so $w \models  \varphi_1[\psi_1 \U \psi_2] \U (\varphi_2 \vee (\G \varphi_1[\false]))$.
	\item $w_k \not\models \varphi_2$ for every $k \geq 0$.  Since $w \not\models \GF \psi_2$, there exists a smallest index $n$ such that $w_k \not\models \psi_2$ for all $k \geq n$.   Further, since $w \models \varphi_1[\psi_1 \U \psi_2] \W \varphi_2$ by hypothesis,  we have in particular $w_k \models \varphi_1[\psi_1 \U \psi_2]$ for every $k < n$. So in order to prove $w \models \varphi_1[\psi_1 \U \psi_2] \U (\varphi_2 \vee (\G \varphi_1[\false]))$ it suffices to show $w_n \models \G \varphi_1[\false]$, or, equivalently, that $w_k \models \varphi_1[\false]$ holds for every $k \geq n$. For this, observe that $w_k \models \varphi_1[\psi_1 \U \psi_2]$ and $ \psi_1 \U \psi_2 \equiv^{w_k} \false $ hold for every $k \geq n$, the latter because $w_k \not\models \psi_2$ holds for every $k \geq n$. Apply now  \cref{lemma:placeholder}(\ref{lemma:replace}). 
\end{itemize}

\smallskip\noindent \textbf{Second premise}: The formula for context $\emptyset$ entails the one for context $ \{ \GF \psi_2 \}$.
	 \[\begin{array}{lr}
		\varphi_1[\psi_1 \U \psi_2] \U (\varphi_2 \vee (\G \varphi_1[\false])) \\
		\kern2em \models \varphi_1[\psi_1 \W \psi_2] \U (\varphi_2 \vee (\G \varphi_1[\false])) &\text{(\cref{lemma:placeholder}(\ref{lemma:nnf}) and $\psi_1 \U \psi_2 \models \psi_1 \W \psi_2$)} \\
		\kern2em \models \varphi_1[\psi_1 \W \psi_2] \U (\varphi_2 \vee (\G \varphi_1[\psi_1 \W \psi_2])) &\text{(\cref{lemma:placeholder}(\ref{lemma:nnf}) and $\false \models \psi_1 \W \psi_2$)} \\
		\kern2em \models \varphi_1[\psi_1 \W \psi_2] \W \varphi_2 &\text{(semantics of LTL)}
	\end{array}\]
	}
\end{proof}

\begin{proposition} \label{prop:23norm}
For every LTL formula $\varphi$ there exists an equivalent formula $\varphi'$ in \fnf{} such that $\nnodes{\varphi'} \leq 4^{2\nnodes{\varphi}} \cdot \nnodes{\varphi}$. Moreover, for every subformula $\GF \psi$ of $\varphi'$ the formula $\psi$ is a subformula of $\varphi$, and every \FG-subformula of $\varphi'$ is also a subformula of $\varphi$.
\end{proposition}

\begin{proof}

We associate to each formula a rank, defined by $\rank{\varphi} = \nnodes{\varphi} + \ubw{\varphi}$. Observe that a formula $\varphi$ is in \fnf{} if{}f $\rank{\varphi} = \nnodes{\varphi}$.
Throughout the proof we say that a formula $\varphi'$ \emph{satisfies the limit property} if for every subformula $\GF \psi$ of $\varphi'$ the formula $\psi$ is a subformula of $\varphi$ and every \FG-subformula of $\varphi'$ is also a subformula of $\varphi$ (notice the asymmetry). Further, we say that a formula $\varphi'$ \emph{satisfies the size property} if $\nnodes{\varphi'} \leq 4^{\mathit{rank}(\varphi)} \cdot \nnodes{\varphi}$ from which the claimed size bound immediately follows.

We prove by induction on $\rank{\varphi}$ that $\varphi$ is equivalent to a formula $\varphi'$ in \fnf{} satisfying the limit and size properties. Within the inductive step we proceed by a case distinction of $\varphi$:

\medskip\noindent If $\varphi = \true, \false, \GF \psi, \FG \psi$ then $\varphi$ is already in \fnf{}, and satisfies the limit and size properties.

\medskip\noindent If $\varphi= \varphi_1 \wedge \varphi_2, \varphi_1 \vee \varphi_2, \varphi_1 \U \varphi_2$ then by induction hypothesis $\varphi_1$ and $\varphi_2$  can be normalized into formulas $\varphi_1'$ and $\varphi_2'$ satisfying the limit and size properties. The formulas $\varphi_1'  \wedge \varphi_2'$, $\varphi_1' \vee \varphi_2'$,  $\varphi_1'  \U \varphi_2'$ are then in \fnf{} (the latter because the additional \U-node is above any \W-node) and satisfy the limit property. The size property holds because:

\[\begin{array}{rlr}
\nnodes{\varphi_1'} + \nnodes{\varphi_2'} + 1
	& \leq 4^{\rank{\varphi_1}} \cdot \nnodes{\varphi_1} + 4^{\rank{\varphi_2}} \cdot \nnodes{\varphi_2} + 1
	& (\text{induction hypotheses}) \\
	& \leq 4^{\rank{\varphi_1}+\rank{\varphi_2}} \cdot (\nnodes{\varphi_1} + \nnodes{\varphi_2} + 1) \\
	& \leq 4^{\rank{\varphi}} \cdot \nnodes{\varphi}
	& (\nnodes{\varphi_i} \leq \nnodes{\varphi} \text{ and } \sum_{i=1}^2 \rank{\varphi_i} < \rank{\varphi})
\end{array}\]

\medskip\noindent If $\varphi  = \X \varphi_1$, then by induction hypothesis there is a formula $\varphi_1'$ equivalent to $\varphi_1$ in \fnf{}, and so $\varphi$ is equivalent to $\X \varphi_1'$, which is in \fnf{} and satisfies the limit and size properties.

\medskip\noindent If $\varphi = \varphi_1 \W \varphi_2$ and $\ubw{\varphi} = 0$, then $\varphi_1 \W \varphi_2$ is already in \fnf{} and satisfies the limit and size properties.

\medskip\noindent If $\varphi = \varphi_1 \W \varphi_2$ and $\ubw{\varphi} > 0$, then we proceed by a case distinction:

\begin{itemize}
\item $\varphi_2$ contains at least one \U-node that is not under a limit node. Let $\psi_1 \U \psi_2$ be such a \U-node. We derive $\varphi_2\hole$ from $\varphi_2$ by replacing each \U-node labeled by $\psi_1 \U \psi_2$ by the special atomic proposition $\hole$. By \cref{lem:UWandWU}(\ref{eqWU}) we have:
\[ \varphi_1 \W \varphi_2[\psi_1 \U \psi_2] \equiv \varphi_1 \U \varphi_2[\psi_1 \U \psi_2] \vee \varphi_1 \W \false\]
Since $\rank{\varphi_1} < \rank{\varphi}$, $\rank{\varphi_2} <  \rank{\varphi}$, and $\rank{\varphi_1 \W \false} < \rank{\varphi}$ (the latter because $\varphi_2$ contains at least one \U-node), by induction hypothesis $\varphi_1$,  $\varphi_2$, and $\varphi_1 \W \false$ can be normalized into formulas $\varphi_1'$, $\varphi_2'$, and $\varphi_3'$ satisfying the limit and size properties. So $\varphi$ can be normalized into $\varphi' = \varphi_1' \U \varphi_2' \vee \varphi_3'$. Moreover, $\varphi'$ satisfies the limit property, because all \GF- and \FG-subformulas of $\varphi'$ are subformulas of $\varphi_1'$, $\varphi_2'$, or $\varphi_3'$. For the size property we calculate:

\[\begin{array}{rlr}
\nnodes{\varphi'} & = \nnodes{\varphi_1'} + \nnodes{\varphi_2'} + \nnodes{\varphi_3'} + 2 \\ 
                  & \leq 4^{\rank{\varphi_1}} \cdot \nnodes{\varphi_1} + 4^{\rank{\varphi_2}} \cdot \nnodes{\varphi_2} + 4^{\rank{\varphi_1\W \false}} \cdot \nnodes{\varphi_1\W\false} + 2
                  & (\text{induction hypotheses}) \\
                  & \leq 4^{\rank{\varphi} - 1} \cdot (\nnodes{\varphi_1} + \nnodes{\varphi_2} + \nnodes{\varphi_1\W\false} + 2) \\
                  & \leq 4^{\rank{\varphi} - 1} \cdot 4 \cdot \nnodes{\varphi} = 4^{\rank{\varphi}} \cdot \nnodes{\varphi}
\end{array}\]

\item Every \U-node of $\varphi_2$ is under a limit node, and $\varphi_1$ contains at least one \U-node that is not under any  limit node. Then $\varphi_1$ contains a maximal subformula $\psi_1 \U \psi_2$ (with respect to the subformula order) that is not under a limit node. We derive $\varphi_1\hole$ from $\varphi_1$ by replacing each \U-node labeled by $\psi_1 \U \psi_2$ that does not appear under a limit node by the special atomic proposition $\hole$. By \cref{lem:UWandWU}(\ref{eqUW}), we have
\begin{equation*}
\varphi_1[\psi_1 \U \psi_2] \W \varphi_2 \equiv \big(\GF \psi_2 \wedge \underbrace{\varphi_1[\psi_1 \W \psi_2] \W \varphi_2}_{\rho_1}\big) \vee \big( \underbrace{\varphi_1[\psi_1 \U \psi_2]}_{\rho_2} \U (\underbrace{\varphi_2 \vee (\varphi_1[\false] \W \false}_{\rho_3}) \big)
\end{equation*}
In order to apply the induction hypothesis we argue that $\rho_1$, $\rho_2$, and $\rho_3$ have rank smaller than $\varphi$, and thus can be normalized to $\rho_1'$, $\rho_2'$ and $\rho_3'$ satisfying the limit and size properties. The formula $\rho_1$ has the same number of nodes as $\varphi$, but fewer $\U$-nodes under $\W$-nodes; so $\ubw{\rho_1} < \ubw{\varphi}$ and thus $\rank{\rho_1} < \rank{\varphi}$. The same argument applies to $\rho_3$.
Finally, $\rank{\rho_2} < \rank{\varphi}$ follows from the fact that $\rho_2$ has fewer nodes than $\varphi$. So $\varphi$ can be normalized to $\varphi'=(\GF\psi_2 \wedge \rho_1') \vee (\rho_2' \U \rho_3')$.

We show that $\varphi'$ satisfies the limit property. Let $\GF \psi$ be a subformula of $\varphi'$. If  $\GF \psi = \GF \psi_2$, then we are done, because $\psi_2$ is a subformula of $\varphi$. Otherwise $\GF\psi$ is a subformula of $\rho_1'$, $\rho_2'$, or $\rho_3'$. Since all of them satisfy the limit property, $\psi$ is a subformula of $\varphi$, and we are done. Further, every $\FG$-subformula of $\varphi'$ belongs to $\rho_1'$, $\rho_2'$, or $\rho_3'$ and so it is also subformula of $\varphi$. For the size property we calculate:
\begin{align*}\nnodes{\varphi'} & = \nnodes{\rho_1'} + \nnodes{\rho_2'} + \nnodes{\rho_3'} + \nnodes{\psi_2} + 4 \\ 
                  & \leq 4^{\rank{\rho_1}} \cdot \nnodes{\rho_1} + 4^{\rank{\rho_2}} \cdot \nnodes{\rho_2} + 4^{\rank{\rho_3}} \cdot \nnodes{\rho_3} + \nnodes{\varphi_1} + 4
                  & (\text{induction hypotheses}) \\
                  & \leq 4^{\rank{\varphi} - 1} \cdot (\nnodes{\varphi} + \nnodes{\varphi} + \nnodes{\varphi})  + \nnodes{\varphi} + 4
                  & (\nnodes{\rho_i}, \nnodes{\varphi_1} \leq \nnodes{\varphi}
 \text{ and } \rank{\rho_i} < \rank{\varphi}) \\
                  & \leq 4^{\rank{\varphi} - 1} \cdot 4 \cdot \nnodes{\varphi} = 4^{\rank{\varphi}} \cdot \nnodes{\varphi}
\end{align*}
\end{itemize}
\end{proof}

\subsection{Stage 2: Moving $\GF$- and $\FG$-subformulas up.}
\label{subsec:stage2}

In this section, we address the second property of the normal form. The following lemma allows us to pull limit subformulas out of any temporal formula. (Note that the second rule is only necessary if the formula before stage 1 contained \FG-subformulas, since stage 1 only creates new \GF-formulas.)

\begin{lemma}
\label{lemma:GF1}
\begin{eqnarray} \varphi[\GF \psi] & \equiv & (\GF \psi \wedge \varphi[\true] ) \vee \varphi[\false]  \label{eqGF1} 
\\ \varphi[\FG \psi] & \equiv & (\FG \psi \wedge \varphi[\true] ) \vee \varphi[\false]  \label{eqFG1}
\end{eqnarray} \end{lemma}

\begin{proof}
We prove that $\varphi[\psi] \equiv (\psi \wedge \varphi[\true]) \vee \varphi[\false]$ for every $\psi$ such that $\G \psi \equiv \psi$, i.e., $w \models \psi$ if{}f $w_k \models \psi$ for all $k \in \N$.
This is a generalization of (\ref{eqGF1}) and (\ref{eqFG1}), since both $\G \GF \psi' \equiv \GF \psi'$ and $\G \FG \psi' \equiv \FG \psi'$ for any $\psi'$. The new statement follows from the Weak Contextual Equivalence Lemma (Lemma~\ref{lem:conteq}) with basis $\{\psi\}$ and formulas $\flat{\varphi[\psi]}{\{\psi\}} = \varphi[\true]$ and $\flat{\varphi[\psi]}{\emptyset} = \varphi[\false]$. It remains to prove that the premises of the lemma hold:
\begin{enumerate}
\item[(i)] Since $\psi$ is satisfied by either all or no suffix of a word $w$, then either $\psi \equiv^w \true$ or $\psi \equiv^w \false$. Hence,
\[ w \models \psi \stackrel{\G \psi \strut\equiv \psi}\iff \psi \equiv^w \true \stackrel{\strut\text{\cref{lemma:placeholder}(\ref{lemma:replace})}}\Rightarrow \varphi[\psi] \equiv^w \varphi[\true] \;\Rightarrow\; (w \models \varphi[\psi] \iff w \models \varphi[\true]) \]
and identical arguments can be used to prove $w \models \varphi[\psi] \iff w \models \varphi[\false]$ from $w \not\models \psi$.
\item[(ii)] Since $\false \models \true$ and by \cref{lemma:placeholder}(\ref{lemma:nnf}), $\varphi[\false] \models \varphi[\true]$.
\end{enumerate}
\end{proof}

We show using (\ref{eqGF1}) and (\ref{eqFG1}) that every formula in \fnf{} can be transformed into an equivalent formula in \snf.

\begin{proposition} \label{prop:snf}
Every LTL formula $\varphi$ in \fnf{} is equivalent to a formula $\varphi'$ in \snf{} such that $\nnodes{\varphi'} \leq 3^{\gfba{\varphi}} \cdot \nnodes{\varphi}$. Moreover, the size of the limit subformulas does not increase: for every $b > 0$, if $\nnodes{\psi} \leq b$ for every limit subformula $\psi$ of $\varphi$, then $\nnodes{\psi'} \leq b$ for every limit subformula $\psi'$ if $\varphi'$.
\end{proposition}

\begin{proof}
We proceed by induction on the number of proper limit sub\-formulas of $\varphi$. If $\varphi$ does not contain any, then it is already in \snf. Assume there exists such a proper limit subformula $\psi$ that is smaller (or incomparable) to all other limit subformulas of $\varphi$ according to the subformula order. We derive $\varphi\hole$ from $\varphi$ by replacing each limit-node labeled by $\psi$ by the special atomic proposition $\hole$. We then apply \cref{lemma:GF1} to obtain:
\[ \varphi[\psi] \equiv (\psi \wedge \varphi[\true] ) \vee \varphi[\false] \qquad \mbox{ where } \psi = \GF \psi', \FG \psi'  \ .\]
Note that $\psi$ does not properly contain any  limit subformula, and so it is in \snf{}. Both $\varphi[\true]$ and $\varphi[\false]$ are still in \fnf{} and they have one limit operator less than $\varphi$. Thus they can be normalized by the induction hypothesis into $\varphi_1'$ and $\varphi_2'$ in \snf. Finally, $\varphi' = (\psi \wedge \varphi_1') \vee \varphi_2'$ is a Boolean combination of formulas in \snf, so it is in \snf. The number of nodes of $T_{\varphi'}$ can be crudely bounded as follows:
\begin{equation*}
\begin{array}{rll}
	\nnodes{\varphi'} & \leq \nnodes{\varphi_1'} + \nnodes{\varphi_2'} + \nnodes{\psi} + 2 \\[1ex]
	                  & \leq 3^{\gfba{\varphi[\true]}} \cdot \nnodes{\varphi[\true]} + 3^{\gfba{\varphi[\false]}} \cdot \nnodes{\varphi[\false]} + \nnodes{\psi} + 2
	                  & (\text{induction hypotheses})  \\[1ex]
	                  & \leq 3^{\gfba{\varphi[\true]}} \cdot 2 \cdot \nnodes{\varphi[\true]} + \nnodes{\psi} + 2
	                  & (\nnodes{\varphi[\true]} = \nnodes{\varphi[\false]}\text{, idem for } n_{\mathrm{lim}})  \\[1ex]
	                  & \leq 3^{\gfba{\varphi[\true]}} \cdot \big(2 \cdot (\nnodes{\varphi} - \nnodes{\psi} + 1) + \nnodes{\psi} + 2\big)
	                  & (\gfba{\varphi[\true]} \leq \gfba{\varphi} - 1) \\[1ex]
	                  & \leq 3^{\gfba{\varphi} - 1} \cdot (2	 \nnodes{\varphi} - \nnodes{\psi} + 4) \\[1ex]
	                  & \leq 3^{\gfba{\varphi} - 1} \cdot (3 \nnodes{\varphi}) = 3^{\gfba{\varphi}} \cdot \nnodes{\varphi}
	                  & (\nnodes{\psi}, \nnodes{\varphi} \geq 2)
\end{array}
\end{equation*}

To show that the size of the limit subformulas does not increase, let $b$ be a bound on the size of the \GF-subformulas of $\varphi$. We claim that the size of each \GF-subformula of $\varphi'$ is also bounded by $b$ (the case of $\FG \psi$ is analogous). Indeed, the \GF-subformulas of $\varphi'$ are $\psi$ (which is already in $\varphi$) and the \GF-subformulas of $\varphi_1'$ and $\varphi_2'$. Since the \GF-subformulas of $\varphi[\true]$ and $\varphi[\false]$ can only have decreased in size, by induction hypothesis the number of nodes of any \GF-subformula of $\varphi_1'$ and $\varphi_2'$ is bounded by $b$, and we are done.
\end{proof}

\subsection{Stage 3: Removing $\W$-nodes ($\U$-nodes)  under $\GF$-nodes ($\FG$-nodes)}
\label{subsec:stage3}

The normalization of LTL formulas is completed in this section by fixing the problems within limit subformulas. In order to do so, we introduce two new rewrite rules that allow us to pull \W-subformulas out of \GF-formulas, and \U-subformulas out of \FG-formulas.

\begin{lemma}
\label{lemma:GF2}
\begin{eqnarray} 
\GF \varphi [\psi_1 \W \psi_2]  & \equiv & \GF \varphi [\psi_1 \U \psi_2] \vee  (\FG \psi_1 \wedge \GF \varphi [\true])  \;\;\; \label{eqGF2} \\
\FG \varphi[\psi_1 \U \psi_2] & \equiv & (\GF \psi_2 \wedge \FG \varphi[\psi_1 \W \psi_2]) \vee \FG \varphi[\false]  \;\;\; \label{eqFG2}
\end{eqnarray}
\end{lemma}

\begin{proof}
Notice that (\ref{eqGF2}) and (\ref{eqFG2}) are instances of the Weak Contextual Equivalence Lemma (Lemma~\ref{lem:conteq}) with bases $\{\FG \psi_1\}$ and $\{\GF \psi_2\}$, respectively. In the case of (\ref{eqGF2}), the contextual formulas are $\flat{(\GF \varphi [\psi_1 \W \psi_2])}{\{\FG \psi_1\}} = \GF \varphi[\true]$ and $\flat{(\GF \varphi [\psi_1 \W \psi_2])}{\emptyset} = \GF \varphi[\psi_1 \U \psi_2]$. The premises of the lemma are satisfied:
{\renewcommand\theenumi{\roman{enumi}}
\begin{enumerate}
	\item Given $w \models \FG \psi_1$, we should first prove $w \models \GF \varphi [\psi_1 \W \psi_2]$ if{}f $w \models \GF \varphi [\true]$. Since $w \models \FG \psi_1$, there is a suffix $v$ of $w$ such that $v_k \models \psi_1$ for all $k \in \N$. By the semantics of LTL, $\psi_1 \W \psi_2 \equiv^v \true$, so $\GF \varphi[\psi_1 \W \psi_2] \equiv^v \GF \varphi[\true]$ by \cref{lemma:placeholder}(\ref{lemma:replace}). Because \GF-formulas are satisfied in all suffixes of a word or in none of them, $w \models \GF \varphi [\psi_1 \W \psi_2]$ if{}f $w \models \GF \varphi [\true]$.

	Given $w \not\models \FG \psi_1$, we claim $w \models \GF \varphi [\psi_1 \W \psi_2]$ if{}f $w \models \GF \varphi [\psi_1 \U \psi_2]$. We know $\psi_1 \W \psi_2 \equiv \psi_1 \U \psi_2 \vee \G \psi_1$, but $\G \psi_1 \equiv^w \false$ because of the hypothesis, so $\psi_1 \W \psi_2 \equiv^w \psi_1 \U \psi_2$. By \cref{lemma:placeholder}(\ref{lemma:replace}), $\GF \varphi [\psi_1 \W \psi_2] \equiv^w \GF \varphi [\psi_1 \U \psi_2]$ and this implies the claim.
	\item Since $\psi_1 \U \psi_2 \models \true$ and by \cref{lemma:placeholder}(\ref{lemma:nnf}), $\GF \varphi[\psi_1 \U \psi_2] \models \GF \varphi[\true]$.
\end{enumerate}}

	For (\ref{eqGF2}), the formulas are $\flat{(\FG \varphi[\psi_1 \U \psi_2])}{\{\GF \psi_2\}} = \FG \varphi[\psi_1 \W \psi_2]$ and $\flat{(\FG \varphi[\psi_1 \U \psi_2])}{\emptyset} = \FG \varphi[\false]$ and the premises of the lemma also hold:
{\renewcommand\theenumi{\roman{enumi}}
\begin{enumerate}
	\item Given $w \models \GF \psi_2$, we prove $w \models \FG \varphi [\psi_1 \U \psi_2]$ if{}f $w \models \FG \varphi [\psi_1 \W \psi_2]$. Since $\psi_1 \U \psi_2 \equiv \psi_1 \W \psi_2 \wedge \F \psi_2$ and $\F \psi_2 \equiv^w \true$ by the hypothesis, $\psi_1 \U \psi_2 \equiv^w \psi_1 \W \psi_2$. By \cref{lemma:placeholder}(\ref{lemma:replace}), $\FG \varphi [\psi_1 \U \psi_2] \equiv^w \FG \varphi [\psi_1 \W \psi_2]$ and this entails the desired equivalence.

	Given $w \not\models \GF \psi_2$, we have to prove $w \models \FG \varphi [\psi_1 \U \psi_2]$ if{}f $w \models \FG \varphi [\false]$. Since $w \not\models \GF \psi_2$, there is a suffix $v$ of $w$ such that $v_k \not\models \psi_2$ for all $k \in \N$. In this case, by the semantics of LTL, $\psi_1 \U \psi_2 \equiv^v \false$, so $\FG \varphi [\psi_1 \U \psi_2] \equiv^v \FG \varphi [\false]$ by \cref{lemma:placeholder}(\ref{lemma:replace}). The conclusion follows again from the limit properties of $\FG$.
	\item Since $\false \models \psi_1 \W \psi_2$ and by \cref{lemma:placeholder}(\ref{lemma:nnf}), $\FG \varphi[\false] \models \FG \varphi[\psi_1 \W \psi_2]$.
\end{enumerate}}
\end{proof}

The following proposition repeatedly applies these rules to show that limit formulas can be normalized with an exponential blowup.

\begin{proposition} \label{prop:gffg}
For every LTL formula $\varphi$ without limit operators, $\GF \varphi$ and $\FG \varphi$ can be normalized into formulas with at most $\nnodes{\varphi'} \leq 3^{\nnodes{\varphi}} \cdot \nnodes{\varphi}$ nodes.
\end{proposition}

\begin{proof}
A \GF-obstacle of a formula is a \W-node or a \U-node under a \W-node inside a \GF-node. Similarly, a \FG-obstacle is a \U-node or a \W-node under a \U-node inside a \FG-node. Finally, an obstacle is either a \GF-obstacle or an \FG-obstacle. We proceed by induction on the number of obstacles of $\GF \varphi$ or $\FG \varphi$. If they have no obstacles, then they are already in normal form (\cref{def:normalform}).

Assume $\GF \varphi$ has at least one obstacle. Then $\varphi$ contains at least one maximal \W-node $\psi_1 \W \psi_2$. We derive $\GF \varphi\hole$ from $\GF \varphi$ by replacing each \W-node labeled by $\psi_1 \W \psi_2$ by the special atomic proposition $\hole$. By \cref{eqGF2}, $\GF \varphi[\psi_1 \W \psi_2]$ is equivalent to
\[  \GF \varphi [\psi_1 \U \psi_2] \vee  (\FG \psi_1 \wedge \GF \varphi [\true]) \]
We claim that each of $\GF \varphi [\psi_1 \U \psi_2]$, $\GF \varphi [\true]$, and $\FG \psi_1$ has fewer obstacles than 
$\GF \varphi[\psi_1 \W \psi_2]$, and so can be normalized by induction hypothesis.  Indeed, $\GF \varphi [\psi_1 \U \psi_2]$, and $\GF \varphi[\true]$ have at least one \W-node less than $\varphi$, and the number of \U-nodes under a \W-node, due to the maximality of $\psi_1 \W \psi_2$, has not increased, and as a consequence it has fewer \GF-obstacles (and by definition no \FG-obstacles). For $\FG \psi_1$, observe first that every \FG-obstacle of $\FG \psi_1$ is a \GF-obstacle of $\GF \varphi[\psi_1 \W \psi_2]$. Indeed, the obstacles of $\FG \psi_1$ are the \U-nodes and the \W-nodes under \U-nodes; the former were \U-nodes under \W-nodes in $\varphi$, and the latter were \W-nodes of $\varphi$, and so both \GF-obstacles of $\GF\varphi$. Moreover, $\psi_1 \W \psi_2$ is a \GF-obstacle of $\varphi$, but not a \FG-obstacle of $\FG \psi_1$. Hence, the number of obstacles has decreased. 

Assume now that $\FG \varphi$ has at least one obstacle. Then $\varphi$ contains at least one maximal \U-node $\psi_1 \U \psi_2$. We derive $\FG \varphi\hole$ from $\FG \varphi$ by replacing each \U-node labeled by $\psi_1 \U \psi_2$ by the special atomic proposition $\hole$. By \cref{eqFG2}, $\FG \varphi[\psi_1 \U \psi_2]$ is equivalent to
\[ (\GF \psi_2 \wedge \FG \varphi[\psi_1 \W \psi_2]) \vee \FG \varphi[\false] \]
Each of $\GF \psi_2$, $\FG \varphi[\psi_1 \W \psi_2]$, and $\FG \varphi[\false]$ has fewer obstacles as $\FG \varphi[\psi_1 \U \psi_2]$, and can be normalized by induction hypothesis. The proof is as above.

	The size of the formula increases at most by a factor of 3 on each step, and the number of steps is bounded by the number of both \W-nodes and \U-nodes in $\varphi$, which is bounded by the total number of nodes in $\varphi$. So the formula has at most $3^{\nnodes{\varphi}} \nnodes{\varphi}$ nodes.
\end{proof}

\subsection{The Normalization Theorem}
\label{sec:overview-rs}

	The main result directly follows from the previous propositions.

\begin{theorem}\label{thm:rewriting:main}
Every formula $\varphi$ of LTL is normalizable into a formula with at most $4^{6\nnodes{\varphi}}$ nodes.
\end{theorem}

\begin{proof}
Any LTL formula $\varphi$ can be transformed into an equivalent $\varphi'$ in \fnf{} of size $\nnodes{\varphi'} \leq 4^{2\nnodes{\varphi}} \cdot \nnodes{\varphi}$ by \cref{prop:23norm}. Moreover, $\gfba{\varphi'} \leq 2 \cdot \nnodes{\varphi}$, since every \FG-subformula of $\psi'$ and every argument $\psi$ of a \GF-subformula of $\varphi'$ is a subformula of $\varphi$. In addition, $\nnodes{\psi} \leq \nnodes{\varphi}$ for every $\GF \psi$ subformula of $\varphi$.

According to \cref{prop:snf}, for every formula $\varphi'$ in \fnf{} there is an equivalent formula $\varphi''$ in \snf{} with
\begin{equation}
	\label{eq:phi2nd}
	\nnodes{\varphi''} \leq 3^{\gfba{\varphi'}} \cdot \nnodes{\varphi'} \leq 3^{2 \nnodes{\varphi}} \cdot (4^{2\nnodes{\varphi}} \cdot \nnodes{\varphi}) \leq  3^{2\nnodes{\varphi}} \cdot 4^{3\nnodes{\varphi}}
\end{equation}
This formula is a Boolean combination of  limit formulas with at most $\nnodes{\varphi}$ nodes, not containing any proper limit node, and other temporal formulas containing neither limit nodes nor \U-nodes under \W-nodes. The latter are in $\Sigma_2$ and \cref{prop:gffg} deals with the former. Notice that every $\GF \psi$ and $\FG \psi$ subformula has at most $\nnodes{\varphi}$ nodes and thus  can be normalized into a formula with at most $3^{\nnodes{\varphi}} \nnodes{\varphi}$ nodes. The result $\varphi'''$ of replacing these limit subformulas by their normal forms within $\varphi''$ is a Boolean combination of normal forms, and so we are done. The number of nodes in the resulting formula $\varphi'''$ is at most:
\begin{align*}
\nnodes{\varphi'''} \leq & ~ \nnodes{\varphi''} + \gfba{\varphi''} \cdot 3^{\nnodes{\varphi}} \cdot \nnodes{\varphi} \\
		    \leq & ~ \nnodes{\varphi''} \cdot 3^{\nnodes{\varphi}} \cdot (\nnodes{\varphi} + 1)
                    & (\gfba{\varphi''} \leq \nnodes{\varphi''}) \\
		    \leq & ~ 4^{3\nnodes{\varphi}} \cdot 3^{3\nnodes{\varphi}} \cdot (\nnodes{\varphi} + 1)
		    & (\text{\ref{eq:phi2nd}}) \\
		    \leq & ~ 4^{3\nnodes{\varphi}} \cdot 4^{(3 \log_4 3 + \frac12) \nnodes{\varphi}} \leq 4^{6\nnodes{\varphi}}
		    & (\nnodes{\varphi} + 1 \leq 4^{\nnodes{\varphi}/2})
\end{align*}
\end{proof}

\subsection{Summary of the normalization algorithm} \label{sec:summary}

We summarize the steps of the normalization algorithm described and proven in this section. Recall that a formula is in normal form if{}f it satisfies the following properties:
\normalFormConds
The normalization algorithm applies the rules in~\cref{tab:allrules} as follows to fix any violation of these properties:
\begin{enumerate}
	\item \U-nodes under \W-nodes and not under limit nodes are removed using rules (\ref{eqWU}) and (\ref{eqUW}). This may introduce new \GF-subformulas. By applying (\ref{eqUW}) only to highest \U-nodes of $\varphi_1$ the number of new \GF-subformulas is only linear in the size of the original formula.
	\item Limit nodes under other temporal nodes are pulled out using rules (\ref{eqGF1}) and (\ref{eqFG1}). By applying the rules only to the lowest limit nodes, it only needs to be applied once for each limit subformula. \item \W-nodes under \GF-nodes are removed using rule (\ref{eqGF2}), and \U-nodes under \FG-nodes are removed using rule (\ref{eqFG2}). This may produce new limit nodes of smaller size that are handled recursively. Choosing highest \W- and \U-nodes ensures that the process produces only a single exponential blowup over the initial size of the formula.
\end{enumerate}
{
\begin{table}[tb]
\begin{equation*}
\begin{array}{lrl@{\;\;\;}c}
\multirow{3}*{\text{Stage 1:}~} & \varphi_1 \W \varphi_2[\psi_1 \U \psi_2] \equiv & \varphi_1 \U \varphi_2[\psi_1 \U \psi_2] \vee \G \varphi_1 & (\ref{eqWU}) \\[1ex]
	                  & \varphi_1[\psi_1 \U \psi_2] \W \varphi_2 \equiv & (\GF \psi_2 \wedge \varphi_1[\psi_1 \W \psi_2] \W \varphi_2) \vee \varphi_1[\psi_1 \U \psi_2] \U (\varphi_2 \vee \G \varphi_1[\false]) & (\ref{eqUW}) \\[1em] \hline \\
\multirow{2}*{\text{Stage 2:}~} & \varphi[\GF \psi] \equiv & (\GF \psi \wedge \varphi[\true]) \vee \varphi[\false] & (\ref{eqGF1}) \\[1ex]
	& \varphi[\FG \psi] \equiv & (\FG \psi \wedge \varphi[\true]) \vee \varphi[\false] & (\ref{eqFG1}) \\[1em] \hline \\
\multirow{2}*{\text{Stage 3:}~} & \GF \varphi [\psi_1 \W \psi_2] \equiv & \GF \varphi [\psi_1 \U \psi_2] \vee  (\FG \psi_1 \wedge \GF \varphi [\true]) & (\ref{eqGF2}) \\[1ex]
	& \FG \varphi[\psi_1 \U \psi_2] \equiv & (\GF \psi_2 \wedge \FG \varphi[\psi_1 \W \psi_2]) \vee \FG \varphi[\false] & (\ref{eqFG2}) \\[1ex]
\end{array}
\end{equation*}
\caption{Normalization rules.} \label{tab:allrules}
\end{table}
}
After the three steps, a formula in normal form is obtained with a single exponential blowup in the number of nodes.

Moreover, notice that $\psi_1$ itself does not play any role in rules (\ref{eqWU}) and (\ref{eqFG2}), and neither does $\psi_2$ in (\ref{eqGF2}). Hence, the application of (\ref{eqWU}) can be made more efficient by replacing not only every occurrence of $\psi_1 \U \psi_2$ outside a limit subformula with $\psi_1 \W \psi_2$ and $\false$, but also every occurrence of $\psi \U \psi_2$ for any formula $\psi$ by $\psi \W \psi_2$ and by $\false$. The same holds for rules (\ref{eqGF2}) and (\ref{eqFG2}).

\begin{example}	
Let us apply the procedure to the formula $\varphi_n = (\cdots ((((a_0 \U a_1) \W a_2) \U a_3) \U a_4) \cdots \U a_n)$ in~\cref{example:exp}. In stage 1, rule (\ref{eqUW}) matches the subformula $(a_0 \U a_1) \W a_2$ and rewrites it to $\GF a_1 \wedge (a_0 \W a_1) \W a_2 \vee (a_0 \U a_1) \U (a_2 \vee \false \W \false)$, where $\false \W \false$ can be simplified to $\false$ and removed from the disjunction. The rewritten formula is in 1-form, because there is no \U-node under a \W-node, so we can continue to stage 2. Now, we must pull the \GF-node $\GF a_1$ out the cascade of \U-nodes using rule (\ref{eqGF1}). This yields
$$\begin{array}{rcl}
\varphi_n & \equiv & (\GF a_1 \wedge (\cdots ((((a_0 \W a_1) \W a_2 \vee (a_0 \U a_1) \U a_2) \U a_3) \U a_4) \cdots \U a_n)) \\[0.1cm]
& & \vee \; (\cdots ((((a_0 \U a_1) \U a_2) \U a_3) \U a_4) \cdots \U a_n)
\end{array}$$
Since the only remaining  limit node is outside any temporal formula, we have obtained a formula in 1-2-form and the procedure arrives to stage 3. Again, the only limit subformula is $\GF a_1$, and $a_1$ does not contain any \W-node, so the formula is completely normalized and we have finished. Observe that $\varphi_n$ has been normalized by exactly two rule applications for all $n \geq 3$, so the algorithm proceeds in linear-time for this family of formulas. The result is not identical, but very similar to the one in~\cref{example:exp}.
\end{example}

\begin{table}[tb]
\[\begin{array}{lr@{\;\;\equiv\;\;}l}
\multirow{6}*{\text{Stage 1:}\;\;} &
	\varphi_1[\psi_1 \M \psi_2] \W \varphi_2 & (\GF \psi_1 \wedge \varphi_1[\psi_1 \R \psi_2] \W \varphi_2) \vee \varphi_1[\psi_1 \M \psi_2] \U (\varphi_2 \vee \G \varphi_1[\false]) \\[1ex]
	& \varphi_1 \W \varphi_2[\psi_1 \M \psi_2] & \varphi_1 \U \varphi_2[\psi_1 \M \psi_2] \vee \G \varphi_1 \\[1ex]
	& \varphi_1[\psi_1 \U \psi_2] \R \varphi_2 & \varphi_1[\psi_1 \U \psi_2] \M \varphi_2 \vee \G \varphi_2 \\[1ex]
	& \varphi_1[\psi_1 \M \psi_2] \R \varphi_2 & \varphi_1[\psi_1 \M \psi_2] \M \varphi_2 \vee \G \varphi_2 \\[1ex]
	& \varphi_1 \R \varphi_2[\psi_1 \U \psi_2] & (\GF \psi_2 \wedge \varphi_1 \R \varphi_2[\psi_1 \W \psi_2]) \vee (\varphi_1 \vee \G \varphi_2[\false]) \M \varphi_2[\psi_1 \U \psi_2] \\[1ex]
	& \varphi_1 \R \varphi_2[\psi_1 \M \psi_2] & (\GF \psi_1 \wedge \varphi_1 \R \varphi_2[\psi_1 \R \psi_2]) \vee (\varphi_1 \vee \G \varphi_2[\false]) \M \varphi_2[\psi_1 \M \psi_2]  \\[1em] \hline \\
\multirow{2}*{\text{Stage 3:}\;\;} &
 	\GF \varphi [\psi_1 \R \psi_2] & \GF \varphi [\psi_1 \M \psi_2] \vee  (\FG \psi_2 \wedge \GF \varphi [\true]) \\[1ex]
	& \FG \varphi[\psi_1 \M \psi_2] & (\GF \psi_1 \wedge \FG \varphi[\psi_1 \R \psi_2]) \vee \FG \varphi[\false] \\[1ex]
\end{array}\]
\caption{Normalization rules for $\R$ and $\M$.} \label{tab:allrulesRM}
\end{table}

\subsection{Extensions and fragments}
\label{sec:extensions}

\paragraph{The operators $\R$ and $\M$.} We have omitted these operators from the proof and the normalization procedure, since they can be expressed in terms of the subset of operators we have considered. However, this translation exponentially increases the number of nodes of the formula, so handling them directly is convenient for efficiency. Their role at every step of the procedure is analogous to that of the \U\ and \W\ operators, i.e. we treat \R\ in the same way as \W\ and we treat \M\ in the same way as $\U$. The corresponding rules are shown in~\cref{tab:allrulesRM}.

\paragraph{Dual normal form.} Recall that a formula is in dual normal form if it satisfies conditions 2 and 3 of Definition \ref{def:normalform} and no \W-node is under a \U-node. Given a formula $\varphi$, let $\psi$ be a formula in primal normal form equivalent to $\overline{\varphi}$. Since $\varphi \equiv \neg \overline{\varphi} \equiv \neg \psi$, pushing the negation into $\psi$ yields a formula equivalent to $\varphi$ in dual normal form.

\paragraph{Past LTL}  Past LTL is an extension of LTL with past operators like yesterday ($\mathbf{Y}$), since ($\mathbf{S}$), etc. In an appendix of \cite{gabbay87}, Gabbay introduced eight rewrite rules to pull future operators out of past operators. Combining these rules with ours yields a procedure that transforms a Past LTL formula into a normalized LTL formula, where past operators are gathered in past-only subformulas, and so can be considered  atomic propositions.

\paragraph{LTL$[\F,\G,\X]$.} Note that LTL is equivalent to first-order logic (FO) over words \cite{GastingDiekertFO-definableLanguages}. In particular a translation from LTL to first-order logic that uses only three variables can be obtained by replacing LTL operators by their respective semantic definitions which are already FO-formulas (See Definition 2.2.). Thus, in fact LTL is equivalent to FO[3], where 3 denotes the number of variables that are used. If we restrict LTL to the unary temporal operators ($\F$, $\G$, and $\X$) and apply the same idea, we obtain formulas with at most 2 variables (FO[2]). Indeed, FO[2] is equivalent to LTL$[\F,\G,\X]$ \cite{EtessamiVW97}. However, our normalization procedure does not take this into account: if we start with a formula from LTL$[\F,\G,\X]$ and apply \cref{thm:rewriting:main} we do not necessarily obtain a normalized formula from LTL$[\F,\G,\X]$. Is the introduction of $\U$ that are not equivalent to $\F$ unavoidable? 

Luckily, no. We can update our rewrite rules such that if we start with a formula from LTL$[\F,\G,\X]$ we also obtain a normalized formula from LTL$[\F,\G,\X]$. The rules are in \Cref{tab:fo2rules} and the correctness proofs proceed as before. The idea is now that we proceed by a case distinction over each $\F\psi$ and split into three cases: (a) $\F\psi$ never holds, (b) $\F\psi$ holds at least once, but finitely often, (c) $\F\psi$ always holds. Moreover, since $\X\F\psi \equiv \F\X\psi$ and $\X\G\psi \equiv \G\X\psi$, every $\G$-subformula with a nested $\F$-subformula can always be reduced to one or more formulas $\G (\F\psi \vee \G \varphi[\F\psi])$ by pushing next operators inside $\F$ and $\G$, calculating the conjunctive normal form of the argument, and splitting the $\G$ operator over the conjunctions.

{
\begin{table}[tb]
\begin{equation*}
\begin{array}{lrl@{\;\;\;}c}
\multirow{2}*{\text{Stage 1:}~} & \G (\varphi) \equiv & \bigwedge_{C \in \textsf{cnf}(\varphi)} \G (\bigvee_{\psi \in C} \psi) \text{ if } \varphi \text{ is not a disjunction of temporal nodes} & (\ref{eqWU}^*) \\[1ex]
                                & \G (\F \psi \vee \varphi[\F \psi]) \equiv & \GF \psi \vee \F (\psi \wedge \X \G \varphi[\false]) \vee \G (\varphi[\false]) & (\ref{eqUW}^*) \\[1em] \hline \\
\multirow{2}*{\text{Stage 2:}~} & \varphi[\GF \psi] \equiv & (\GF \psi \wedge \varphi[\true]) \vee \varphi[\false] & (\ref{eqGF1}) \\[1ex]
	& \varphi[\FG \psi] \equiv & (\FG \psi \wedge \varphi[\true]) \vee \varphi[\false] & (\ref{eqFG1}) \\[1em] \hline \\
\multirow{2}*{\text{Stage 3:}~} & \GF \varphi [\G \psi] \equiv & \GF \varphi [\false] \vee  (\FG \psi \wedge \GF \varphi [\true]) & (\ref{eqGF2}^*) \\[1ex]
	& \FG \varphi[\F \psi] \equiv & (\GF \psi \wedge \FG \varphi[\true]) \vee \FG \varphi[\false] & (\ref{eqFG2}^*) \\[1ex]
\end{array}
\end{equation*}
\caption{LTL$[\F,\G,\X]$-Normalization rules.} \label{tab:fo2rules}
\end{table}
}

 \section{A Novel Translation from LTL to Deterministic Rabin Automata (DRW)}
\label{sec:LTLtoDRW}

We apply our $\Delta_2$-normalization procedure to derive a new translation from LTL to
DRW via weak alternating automata (AWW).

It is well-known that
an LTL formula $\varphi$ of length $n$ can be translated into an AWW with $O(n)$ states  \cite{MullerSS88,Vardi94}.
We show that, if $\varphi$ is in normal form, then the AWW can be chosen so that every path through the automaton switches at most once between accepting and non-accepting states. We then prove that determinizing AWWs satisfying this additional property is much simpler than the general case. Finally, we give an upper bound on the size of the resulting DRW: if we use the normal form of \cref{thm:normthm}, then the number of states of the final DRW is double exponential in the length of $\varphi$, which is asymptotically optimal.

The section is structured as follows: \Cref{subsec:aww} introduces basic definitions,
\Cref{subsec:LTLtoAWW2} shows how to translate an $\Delta_2$-formula into AWWs with at most one alternation between accepting and non-accepting states, and \Cref{subsec:det} presents the determinization procedure for this subclass of AWWs.

\subsection{Weak and Very Weak Alternating Automata}
\label{subsec:aww}

Let $X$ be a finite set. The set of positive Boolean formulas over $X$, denoted by $\mathcal{B}^+(X)$, is the closure of $X \cup \{ \true, \false \}$ under disjunction and conjunction.
A set $S \subseteq X$ is a model of $\theta \in B^+(X)$ if the truth assignment that assigns true to the elements of $S$ and false to the elements of $X \setminus S$ satisfies $\theta$.
Observe, that if $S$ is a model of $\theta$ and $S \subseteq S'$ then $S'$ is also a model of $\theta$.
A model $S$ of $\theta$ is minimal if no proper subset of $S$ is also a model of $\theta$. The set of minimal models of $\theta$ is denoted $\mathcal{M}_\theta$. Observe that $\mathcal{M}_\true=\{\emptyset\}$ and $\mathcal{M}_\false=\emptyset$.
Two formulas are equivalent, denoted $\theta \equiv \theta'$, if their sets of minimal models are equal, i.e., $\mathcal{M}_\theta = \mathcal{M}_{\theta'}$.

\subsubsection*{Alternating Büchi and co-Büchi  automata.}
An alternating Büchi automaton over an alphabet $\Sigma$ is a tuple $\mathcal{A} = \langle \Sigma, Q, \theta_0, \delta, \alpha \rangle$, where $Q$ is a finite set of states, $\theta_0 \in \mathcal{B}^+(Q)$ is an initial formula,\footnote{Many papers define alternating automata with an initial state instead of an initial formula. Both models are equivalent, and an initial formula is better for our purposes.} $\delta \colon Q \times \Sigma \mapsto \mathcal{B}^+(Q)$ is the transition function, and $\alpha \subseteq Q$ is the acceptance condition \cite{ChandraKS81,Vardi96}.

\begin{example}
\label{ex:running}
We use as running example the alternating Büchi automaton $\mathcal{A} = \langle \Sigma, Q, \theta_0, \delta, \alpha \rangle$ over the alphabet $\Sigma= 2^{\{a, b, c\}}$, where $Q = \{q_0, q_1, q_2\}$, $\theta_0 = q_0$, $\alpha = \{q_1\}$, and $\delta$ is given by:
$$\delta(q_0, \sigma) = \begin{cases} q_0 \vee q_1 & \text{if } \sigma=\{a\} \\ q_2 & \text{if } \sigma=\{b\} \\ \false            & \text{otherwise,} \end{cases}  \quad 
\delta(q_1, \sigma) = \begin{cases} q_1          & \text{if } \sigma = \{b\} \\ q_1 \wedge q_2 & \text{otherwise,} \end{cases}  \quad \mbox{ and }
\delta(q_2, \sigma) = \begin{cases} \true   & \text{if }  \sigma = \{c\}  \\ q_2            & \text{otherwise.} \end{cases}  $$
\end{example}

A \emph{run} of $\mathcal{A}$ on the word $w$ is a directed acyclic graph $G=(V, E)$ satisfying the following properties:
\begin{itemize}
\item $V \subseteq Q \times \mathbb{N}_0$, and $E \subseteq \{ ((q, l), (q', l +1)) \mid q, q' \in Q \mbox{ and } l \geq 0 \}$.
We call the elements of $V$ \emph{nodes}.  For every $l \geq 0$ we let $Q_l :=\{ q \colon (q,l) \in V \}$.  
\item $Q_0$ is a minimal model of $\theta_0$.
\item For every $l \geq 0$ and every $q \in Q_l$, either $\delta(q, w[l]) \equiv \false$, or the set $\{q' \colon ((q, l), (q',l+1)) \in E\}$ is a minimal model of $\delta(q, w[l])$. (Observe that if $\delta(q, w[l]) \equiv \true$ then this set is empty.) 
\item For every $l \geq 1$, if $q \in Q_l$ then there exists $q' \in Q_{l-1}$ such that $((q',l-1),(q,l)) \in E$.
\end{itemize}
\noindent Runs can be finite or infinite. A run $G$ is \emph{accepting} if
\begin{itemize}
\item[(a)] $\delta(q, w[l]) \not\equiv \false$ for every $q \in Q_l$, and
\item[(b)] every infinite path of $G$ visits infinitely often nodes $(q, l) \in V$ such that $q \in \alpha$.
\end{itemize}
In particular, every finite run satisfying (a) is accepting. $\mathcal{A}$ accepts a word $w$ if{}f it has an accepting run $G$ on $w$. The language $\lang(\mathcal{A})$ recognized by $\mathcal{A}$ is the set of words accepted by $\mathcal{A}$. Two alternating Büchi automata are equivalent if they recognize the same language.

\begin{example}
\label{ex:runs}
Figure \ref{fig:runs} shows (initial parts of) four runs of the alternating Büchi automaton of \cref{ex:running} on the words $\{a\} \{a\}\{c\}^\omega$ (top two runs), $\{a\} \{a\}\{b\}^\omega$ (third run), and $\{b\} \{c\}^\omega$ (fourth run). Let us check that the first run satisfies the properties of the definition. We have  $Q_0=\{q_0\}$, which is indeed a minimal model of $\theta_0$. Further, we have:
\begin{itemize} 
\item The set $\{q' \colon ((q_0, 0),(q',1)) \in E\}= \{q_1\}$ is a minimal model of $\delta(q_0,\{a\})= q_0 \vee q_1$. 
\item For every $l \geq 1$, the set $\{q' \colon ((q_1, l),(q',l+1)) \in E\}=\{q_1, q_2\}$ is a minimal model of $\delta(q_1,\{a\})= \delta(q_1,\{c\})= q_1 \wedge q_2$.
\item For every $l \geq 2$, the set $\{q' \colon ((q_2,l,(q',l+1)) \in E\} = \emptyset$ is a minimal model of $\delta(q_2, \{c\})= \true$.
\end{itemize}
The first run is infinite and accepting, because the only infinite path of the run visits $q_1$ infinitely often. The second run is finite and not accepting, because $\delta(q_0, \{c\}) = \false$. The third run is infinite and also not accepting, because  the infinite path $(q_0, 0)(q_1, 1) (q_2, 2)(q_2, 3) (q_2, 4) \cdots$ only visits one node with states of $\alpha$, namely $(q_1, 1)$. The fourth run is finite and accepting, because $\delta(q_2, \{c\}) = \true$.

\begin{figure}[h]
  \begin{center}
	\begin{tikzpicture}
\newcommand{\SetToWidestCases}[1]{\mathmakebox[1.1cm][l]{#1}}\begin{scope}[yshift=5cm]
	\node[state,initial]   (00) at (0,0) {$(q_0,0)$};
	\node[state,accepting] (11) at (2,0) {$(q_1,1)$};
	\node[state,accepting] (12) at (4,0) {$(q_1,2)$};
	\node[state]           (22) at (4,-1) {$(q_2,2)$};
	\node[state,accepting] (13) at (6,0) {$(q_1,3)$};
	\node[state]           (23) at (6,-1) {$(q_2,3)$};
	\node[state,accepting] (14) at (8,0) {$(q_1,4)$};
	\node[state]           (24) at (8,-1) {$(q_2,4)$};
	\node          (dots) at (9.5,-0.5) {{\Large$\cdots$}};
	
	\node          (l1) at (1,0.6) {$\{a\}$};
	\node          (l2) at (3,0.6) {$\{a\}$};
	\node          (l3) at (5,0.6) {$\{c\}$};
	\node          (l4) at (7,0.6) {$\{c\}$};
	\node          (l5) at (9.5,0.6) {{\Large$\cdots$}};

	\path[->]
	(00) edge node[above]{} (11)
	(11) edge node[above]{} (12)
	(11) edge node[above]{} (22)
	(12) edge node[above]{} (13)
	(12) edge node[above]{} (23)
	(13) edge node[above]{} (14)
	(13) edge node[above]{} (24)
	;
	\end{scope}
	
	\begin{scope}[yshift=2cm]
	\node[state,initial]   (00) at (0,0) {$(q_0,0)$};
	\node[state] (01) at (2,0) {$(q_0,1)$};
	\node[state] (02) at (4,0) {$(q_0,2)$};
	
	\node          (l1) at (1,0.6) {$\{a\}$};
	\node          (l2) at (3,0.6) {$\{a\}$};
	\node          (l3) at (5,0.6) {$\{c\}$};
	\node          (l4) at (7,0.6) {$\{c\}$};
	\node          (l5) at (9.5,0.6) {{\Large$\cdots$}};

	\path[->]
	(00) edge node[above]{} (01)
	(01) edge node[above]{} (02)
	;
	\end{scope}
	\begin{scope}
	\node[state,initial]   (00) at (0,0) {$(q_0,0)$};
	\node[state,accepting] (11) at (2,0) {$(q_1,1)$};
	\node[state,accepting] (12) at (4,0) {$(q_1,2)$};
	\node[state]           (22) at (4,-1) {$(q_2,2)$};
	\node[state,accepting] (13) at (6,0) {$(q_1,3)$};
	\node[state]           (23) at (6,-1) {$(q_2,3)$};
	\node[state,accepting] (14) at (8,0) {$(q_1,4)$};
	\node[state]           (24) at (8,-1) {$(q_2,4)$};
	\node          (dots) at (9.5,-0.5) {{\Large$\cdots$}};
	
	\node          (l1) at (1,0.6) {$\{a\}$};
	\node          (l2) at (3,0.6) {$\{a\}$};
	\node          (l3) at (5,0.6) {$\{b\}$};
	\node          (l4) at (7,0.6) {$\{b\}$};
	\node          (l5) at (9.5,0.6) {{\Large$\cdots$}};

	\path[->]
	(00) edge node[above]{} (11)
	(11) edge node[above]{} (12)
	(11) edge node[above]{} (22)
	(12) edge node[above]{} (13)
	(22) edge node[above]{} (23)
	(13) edge node[above]{} (14)
	(23) edge node[above]{} (24)
	;
	\end{scope}
	\begin{scope}[yshift=-3cm]
	\node[state,initial]   (00) at (0,0) {$(q_0,0)$};
	\node[state] (01) at (2,0) {$(q_2,1)$};
	
	\node          (l1) at (1,0.6) {$\{b\}$};
	\node          (l2) at (3,0.6) {$\{c\}$};
	\node          (l3) at (5,0.6) {$\{c\}$};
	\node          (l4) at (7,0.6) {$\{c\}$};
	\node          (l5) at (9.5,0.6) {{\Large$\cdots$}};

	\path[->]
	(00) edge node[above]{} (01)
	;
	\end{scope}
	\end{tikzpicture}
  \end{center}
\caption{Initial parts of four runs of the alternating Büchi automaton of \cref{ex:running} on the words $\{a\} \{a\}\{c\}^\omega$, $\{a\} \{a\}\{b\}^\omega$, and $\{b\}\{c\}^\omega$.}
\label{fig:runs}
\end{figure}
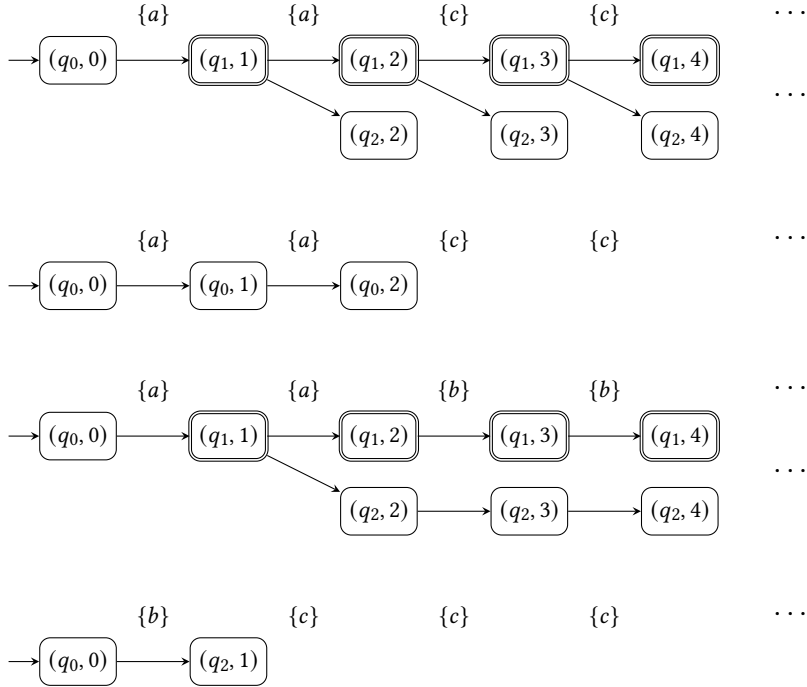
\end{example}

Alternating \emph{co-Büchi} automata are defined as Büchi automata, replacing condition (b) by the co-Büchi condition (every infinite path of $G$ only visits $\alpha$-nodes finitely often).

An alternating automaton is \emph{deterministic} if for every state $q \in Q$ and every letter $a \in \Sigma$ there exists $q' \in Q$ such that $\delta(q,a) = q'$, and \emph{non-deterministic} if for every $q \in Q$ and every $a \in \Sigma$ there exists $Q' \subseteq Q$ such that $\delta(q,a) = \bigvee_{q' \in Q'} q'$.

\subsubsection*{Weak and very weak automata.} Let $\mathcal{A} = \langle \Sigma, Q, \theta_0, \delta, \alpha \rangle$ be an alternating (co-)Büchi automaton. We write $q \trans{} q'$ if there is $a \in \Sigma$ such that $q'$ belongs to some minimal model of $\delta(q, a)$. An automaton $\mathcal{A}$
is \emph{weak} if there is a partition $Q_0$, \dots, $Q_m$ of $Q$ such that
\begin{itemize}
\item for every $q, q' \in Q$, if $q \trans{} q'$  then there are $i \leq j$ such that $q \in Q_i$ and $q' \in Q_j$, and
\item for every $0 \leq i \leq m$: $Q_i \subseteq \alpha$ or $Q_i \cap \alpha = \emptyset$.
\end{itemize}
$\mathcal{A}$ is \emph{very weak} or \emph{linear} if it is weak and every class $Q_i$ of the partition is a singleton (that is, $|Q_i| = 1$) \cite{MullerSS86,KupfermanV01}.
We let $\aww$ and $\alw$ denote the set of weak and very weak alternating automata, respectively.

\begin{example}
Figure \ref{fig:ex:aww:1} shows the relation $\trans{}$ between the states of the automaton of \Cref{ex:running}. For example, we have $q_1 \trans{} q_2$ because $\{q_1, q_2\}$ is a minimal model of $\delta(q_1, \{a\}) = q_1 \wedge q_2$. The automaton is very weak because of the partition $Q_0 = \{q_0\}$, $Q_1 = \{q_1\}$, $Q_2 = \{q_2\}$.
\begin{figure}[h]
  \begin{center}
	\begin{tikzpicture}
\newcommand{\SetToWidestCases}[1]{\mathmakebox[1.1cm][l]{#1}}

	\node[state]   (0) at (0,0) {$q_0$};
	\node[state,accepting] (1) at (1.25,0) {$q_1$};
	\node[state]           (2) at (2.5,0) {$q_2$};

	\path[->]
	(0) edge[loop above]  (0)
	(0) edge  (1)
	(1) edge[loop above]  (1)
	(1) edge (2)
	(2) edge[loop above]  (2);
	\end{tikzpicture}
  \end{center}
  \caption{Relation $\trans{}$ between the states of the automaton of \Cref{ex:running}.}
\label{fig:ex:aww:1}
\end{figure}
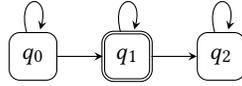

\end{example}

Observe that for every \emph{weak} automaton with a co-Büchi acceptance condition we can define a Büchi acceptance condition on the same structure recognizing the same language.
Thus, we will from now on assume that every weak automaton is equipped with a Büchi acceptance condition.

We define the alternation height of a weak alternating automaton. The definition is very similar, but not identical, to the standard one as presented in e.g.\ \cite{Vardi96}.
A weak automaton $\mathcal{A} = \langle \Sigma, Q, \theta_0, \delta, \alpha \rangle$ has \emph{alternation height} $n$ if every path $q \rightarrow q' \rightarrow q'' \cdots $ of $\mathcal{A}$ alternates at most $n - 1$ times between $\alpha$ and $Q \setminus \alpha$. The automaton of \Cref{ex:running} (see also \Cref{fig:ex:aww:1}) has height 3 because of the path $q_0 \rightarrow q_1 \rightarrow q_2$, which exhibits two alternations. 

\begin{definition}
\label{def:awwn}
$\aww[n]$ (resp. $\alw[n]$) denote the set of all weak (resp. very weak) alternating automata with height at most $n$. Further, we let $\aww[n,\A]$ (resp. $\aww[n,\R]$) denote the set of automata of $\aww[n]$ whose initial formula $\theta_0$ satisfies $\theta_0 \in \mathcal{B}^+(\alpha)$ (resp. $\theta_0 \in \mathcal{B}^+(Q \setminus \alpha)$). For instance, the automaton of \Cref{fig:ex:aww:1} belongs to $\alw[3,\R]$.
\end{definition}

 \subsection{Translation of LTL to $\alw[2]$}
\label{subsec:LTLtoAWW2}
\newcommand{\Pairs}[1]{\mathit{An}({#1})}
\newcommand{\bsta}[1]{[{#1}]_{\leq \Gamma}}
\newcommand{\bstatwo}[2]{[{#1}]_{\leq {#2}}}
\newcommand{\sta}[2]{{#1}_{\scriptscriptstyle #2}}
\newcommand{\atom}[2]{\langle#1,#2\rangle}

In the standard translation of LTL to $\alw$, the states of the $\alw$ for a formula $\varphi$ are subformulas of $\varphi$  \cite{MullerSS86,Vardi96}.  We show that, at the price of a slightly more complicated translation, the resulting $\alw$ for a $\Delta_i$-formula belongs to $\alw[i]$. Thus, by the Normalization Theorem every LTL formula can be translated into an equivalent $\alw[2]$.
The idea of the construction is to use subformulas of the formulas as states, ensuring that 
\begin{enumerate}
	\item transitions can only lead from a subformula to itself or to another of a smaller in the syntactic-future hierarchy (\Cref{fig:syntactic-future}); and
	\item accepting states are subformulas of the $\Pi$-classes of the hierarchy.
\end{enumerate}
These two conditions immediately imply that the alternating automaton for a formula of $\Sigma_i$ has  alternation height $i$; indeed, every change of state
involves going down in the hierarchy. There are two small technical problems. First, at the bottom of the hierarchy we have $\Sigma_0=\Pi_0$, and so it is not well defined whether bottom formulas are accepting states or not. Fortunately, in our automata such states are not reachable, and so this question is irrelevant. Second, the hierarchy level of a formula is not always well-defined, because some formulas do not belong to one single lowest level of the hierarchy.  For example, the formula $\X a$ belongs to both $\Pi_1$ and $\Sigma_1$. To solve this problem the states of the automaton will be formulas marked with one of the classes they belong to. For example, the automaton for $\X a$ will contain two states $\sta{(\X a)}{\Sigma_1}$ and $\sta{(\X a)}{\Pi_1}$.

Formally, we proceed as follows. A formula $\varphi$ is \emph{proper} if it is neither a Boolean constant ($\true$ or $\false$) nor a conjunction or a disjunction. An \emph{atom} is a pair $\atom{\varphi}{\Gamma}$, where $\varphi$ is a proper formula, and $\Gamma$ is a smallest class of the syntactic-future hierarchy such that $\varphi \in \Gamma$. A marked formula is a positive Boolean combination of atoms (including $\true$ and $\false$). Given a formula $\varphi$ and a class $\Gamma$ such that $\varphi \in \Gamma$, we define the marked formula $\sta{\varphi}{\Gamma}$ as follows: 
$\sta{\true}{\Gamma} \coloneq \true$; $\sta{\false}{\Gamma} \coloneq \false$;
$\sta{(\varphi \vee \psi)}{\Gamma} \coloneq \sta{\varphi}{\Gamma} \vee \sta{\psi}{\Gamma}$; $\sta{(\varphi \wedge \psi)}{\Gamma} \coloneq \sta{\varphi}{\Gamma} \wedge \sta{\psi}{\Gamma}$; if $\varphi$ is proper, then $\sta{\varphi}{\Gamma}$ is the disjunction of all atoms $\atom{\varphi}{\Gamma_\varphi}$ such that $\Gamma_\varphi \subseteq \Gamma$ is a smallest class containing $\varphi$. Observe, in particular, that $\sta{\varphi}{\Gamma}=\sta{\varphi}{\Gamma'}$ for every $\Gamma' \supseteq \Gamma$.

\begin{example}
If $\varphi = \false \vee (\X a \wedge c)$, then $\sta{\varphi}{\Delta_1} = \false \vee (\sta{(\X a)}{\Delta_1} \wedge \sta{c}{\Delta_1}) = (\atom{\X a}{\Sigma_1} \vee \atom{(\X a)}{\Pi_1}) \wedge \atom{c}{\Delta_0}$.
\end{example}

\begin{definition}
\label{def:aawforltl}
Let $\varphi \in \Delta_i$ for some $i \geq 0$. We define the automaton $\mathcal{A}_\varphi = \langle 2^{Ap}, Q, \theta_0, \delta, \alpha \rangle$ as follows, where the symbol $\Gamma$ ranges over the classes of the syntactic hierarchy.
\begin{itemize}
\item $Q$ is the set of atoms $\atom{\psi}{\Gamma}$ where $\psi$ is a proper subformula of $\varphi$ and $\Gamma \subseteq \Delta_i$.
\item $\theta_0 = \sta{\varphi}{\Delta_i}$ (recall that $\theta_0$ is a positive Boolean combination of states).
\item $\alpha = \{ (\psi, \Gamma) \in Q \colon \Gamma = \Pi_j \mbox{ for some $j \geq 0$}\}$. 
\item $\delta \colon Q \times \Sigma \to \mathcal{B}^+(Q)$ is the restriction to $Q \times \Sigma$ of a function
$\delta$ that assigns to each marked formula $\sta{\varphi}{\Gamma}$ (notice that, abusing language, we also call this function $\delta$),
a positive Boolean combination of states:
\begin{align*}
\delta(\sta{\true}{\Gamma}, \sigma) &  = \true  &	 
\delta(\sta{\false}{\Gamma}, \sigma) &  = \false 
\\
\delta(\sta{a}{\Gamma}, \sigma) & =  \begin{cases}  \true & \mbox{if $a \in \sigma$} \\ \false & \mbox{otherwise} \end{cases} &	
\delta(\sta{\overline{a}}{\Gamma}, \sigma) & =  \begin{cases}  \false & \mbox{if $a \in \sigma$} \\ \true & \mbox{otherwise} \end{cases} \\
\delta(\sta{(\varphi \vee \psi)}{\Gamma}, \sigma) & =  \, \delta(\sta{\varphi}{\Gamma}, \sigma) \vee \delta(\sta{\psi}{\Gamma}, \sigma) &
\delta(\sta{(\varphi \wedge \psi)}{\Gamma}, \sigma) & =  \, \delta(\sta{\varphi}{\Gamma}, \sigma) \wedge \delta(\sta{\psi}{\Gamma}, \sigma)   
\\
\delta(\sta{(\X \psi)}{\Gamma}, \sigma) & =  \, \sta{\psi}{\Gamma} 
\\
\delta(\sta{(\varphi \U \psi)}{\Gamma}, \sigma)  & =  \,  \delta(\sta{\psi}{\Gamma}, \sigma) \vee \big(\delta(\sta{\varphi}{\Gamma}, \sigma) \wedge \sta{(\varphi \U \psi)}{\Gamma} \big) & 
\delta(\sta{(\varphi \R \psi)}{\Gamma}, \sigma)  & =  \, \delta(\sta{\psi}{\Gamma}, \sigma) \wedge \big(\delta(\sta{\varphi}{\Gamma}, \sigma) \vee \sta{(\varphi \R \psi)}{\Gamma}\big) 
\\
\delta(\sta{(\varphi \W \psi)}{\Gamma}, \sigma) &  =  \, \delta(\sta{\psi}{\Gamma}, \sigma) \vee \big(\delta(\sta{\varphi}{\Gamma}, \sigma) \wedge \sta{(\varphi \W \psi)}{\Gamma} \big)
&
\delta(\sta{(\varphi \M \psi)}{\Gamma}, \sigma) & =  \, \delta(\sta{\psi}{\Gamma}, \sigma) \wedge \big(\delta(\sta{\varphi}{\Gamma}, \sigma) \vee \sta{(\varphi \M \psi)}{\Gamma}\big) 
\end{align*}
\end{itemize}
\end{definition}

\begin{example}
Consider the formula $\varphi = a \U (\X\G (b \vee \X\F c)) \in \Sigma_3$. The automaton $\mathcal{A}_\varphi$ is the one shown in 
\Cref{ex:running}, with the correspondence $q_0 = \atom{\varphi}{\Delta_3}$, $q_1 = \atom{\G (b \vee \X\F c)}{\Pi_2}$, and 
$q_2 = \atom{\F c}{\Sigma_1}$. Let us compute $\delta(q_1, \sigma)$. We get:
\begin{align*}
\delta(q_1, \sigma) 
& = \delta(\atom{\G (b \vee \X\F c)}{\Pi_2}, \sigma) 
= \delta(\sta{(\G (b \vee \X\F c))}{\Pi_2}, \sigma)  
\\
&= \delta(\sta{(b \vee \X\F c)}{\Pi_2}, \sigma) \wedge \atom{\G (b \vee \X\F c)}{\Pi_2} \\
& = (\delta(\sta{b}{\Pi_2}, \sigma) \vee \atom{\F c}{\Sigma_1}) \wedge  \atom{\G (b \vee \X\F c)}{\Pi_2}
= \begin{cases} q_1 & \mbox{ if $b \in \sigma$} \\ q_1 \wedge q_2 & \mbox{otherwise} \end{cases}
\end{align*}
\end{example}

\begin{lemma}
\label{lem:aww:translation}
Let $\varphi$ be a formula of $\Delta_i$ with $n$ proper subformulas. The automaton $\mathcal{A}_\varphi$ belongs to $\alw[i]$, has $2n$ states, and recognizes $\lang(\varphi)$.
\end{lemma}

\begin{proof}
Let us first show that $\mathcal{A}_\varphi$ belongs to $\alw[i]$. It follows immediately from the definition of $\mathcal{A}_\varphi$  that 
for every two states $\atom{\psi}{\Gamma}, \atom{\psi'}{\Gamma'}$ of $\mathcal{A}_\varphi$, if 
$\atom{\psi}{\Gamma} \trans{} \atom{\psi'}{\Gamma'}$ then $\Gamma' \subseteq \Gamma$. So in every path
there are at most $(i-1)$ alternations between $\Sigma$ and $\Pi$ classes.  Since the states of $\alpha$ are those 
annotated with $\Pi$ classes, there are also at most $(i-1)$ alternations between $\alpha$ and non-$\alpha$ states in a path.

To show that $\mathcal{A}_\varphi$ has at most $2n$ states, observe that for every formula $\psi$
there are at most two smallest classes of the syntactic-future hierarchy containing $\psi$. So $\mathcal{A}_\varphi$ has at most
two states for each proper subformula of $\varphi$.

To prove that $\mathcal{A}_\varphi$ recognizes $\lang(\varphi)$ one shows by induction on $\psi$ that 
$\mathcal{A}_\psi$ recognizes $\lang(\psi)$ from every marked formula $\sta{\psi}{\Gamma}$. The proof is completely analogous to the one given in \cite{Vardi96}.
\end{proof}
 \subsection{Determinization of $\aww[2]$}
\label{subsec:det}

\newcommand{\Promising}{\mathit{Promising}}
\newcommand{\Levels}{\mathit{Levels}}

We present a determinization procedure for $\aww[2,\R]$ and $\aww[2,\A]$ inspired by the break-point construction from \cite{MiyanoH84}. More precisely, given an $\aww[2,\R]$, we construct an equivalent deterministic co-Büchi automaton, and given an $\aww[2,\A]$, we construct an equivalent deterministic Büchi automaton. We only describe the construction for $\aww[2,\R]$, as the one for $\aww[2,\A]$ is dual. 

The section is structured as follows. First, we introduce the notion of the level sequence of a run. Second, we present the fundamental property of $\aww[2,\R]$ that the procedure will exploit. Third,  we describe the procedure itself. Finally, we combine the procedures for $\aww[2,\R]$ and $\aww[2,\A]$ into a procedure that, given an arbitrary $\aww[2]$, constructs an equivalent deterministic Rabin automaton. 

\subsubsection{Level sequence of a run.} Let $\mathcal{A} = \langle \Sigma, Q, \theta_0, \delta, \alpha \rangle$ be an alternating  Büchi automaton.  A set $U \subseteq Q$ of states is called a \emph{level}. If $U \subseteq \alpha$, then $U$ is an \emph{$\alpha$-level}. Given two levels $U, U'$ and a letter $a \in \Sigma$, we say that $U'$ is a \emph{successor of $U$ w.r.t.\ $a \in \Sigma$}, also called an \emph{$a$-successor} of $U$, if for every $q \in U$ there is a minimal model $S_q$ of $\delta(q, a)$ such that $U' = \bigcup_{q \in U} S_q$. We make two observations:
\begin{itemize}
\item The empty set of states is a level, and moreover an $\alpha$-level for any $\alpha \subseteq Q$.  
The empty level has exactly one $a$-successor for every $a \in \Sigma$, namely the empty level itself.
\item A level $U$ has no $a$-successors if and only if it contains a state $q$ such that $\delta(q, a) = \false$. Indeed, if $\delta(q, a) = \false$ then there is no minimal model of $\delta(q, a)$. Conversely, if $\delta(q, a) \neq \false$ for every $q \in U$, then $\delta(q, a)$ has at least one minimal model $S_q$ for every state $q \in U$, and the set $U':=\bigcup_{q \in U} S_q$ is a (possibly empty) $a$-successor of $U$.
\end{itemize}

Recall that, given a run $G=(V, E)$ on a word $w$, we define for every $l \geq 0$ the set $Q_l := \{q \colon (q,l) \in V\}$. By the definitions of run and level, for every $l \geq 0$ either $\delta(q, a) = \false$ for some $q \in Q_l$, or $Q_{l+1}$ is a $w[l]$-successor of $Q_l$. We define the \emph{level sequence} of $G$ as a certain prefix of the infinite sequence $Q_0 \, Q_1 \, Q_2 \ldots$:
\begin{itemize}
\item  If there exists a smallest $l \geq 0$ such that $Q_l$ has no $w[l]$-successor, then the level sequence is the finite sequence $Q_0 \, Q_1 \cdots Q_l$. 
\item Otherwise, the level sequence of $G$ is the infinite sequence $Q_0 \, Q_1 \, Q_2 \cdots$ itself. 
\end{itemize}

\begin{example}
\label{ex:levels}
Let us compute the level sequences of the four runs of Figure \ref{fig:runs}. 
\begin{itemize}
\item First run: $\{q_0\} \, \{q_1\} \, \{q_1, q_2\}^\omega$. Since $\alpha = \{q_1\}$, the only $\alpha$-level  of the level sequence is $\{q_1\}$.
\item Second run: $\{q_0\} \, \{ q_0\} \, \{q_0\}$.  The level sequence ends at the second level because $\delta(q_0, \{c\}) = \false$.
\item Third run:  $\{q_0\} \, \{q_1\} \, \{q_1, q_2\}^\omega$. Observe that although the first and third runs have the same sequence of levels, the first run is accepting but the second one is not.
\item Fourth run: $\{q_0\} \, \{q_0\} \, \emptyset^\omega$. Indeed, $\emptyset$ is a $c$-successor of $\{q_0\}$ because $\delta(q_0, c)= \true$, and so $\emptyset$ is a minimal model of $\delta(q_0, c)$. While the run itself is finite, its level sequence is infinite.
\end{itemize}
\end{example}

\subsubsection{Fundamental property of $\aww[2,\R]$} 
The first and third runs of \Cref{ex:levels} have the same level sequence, but one is accepting and the other is not. Therefore, the level sequence of a run does not determine in general whether the run is accepting. However, the runs of \Cref{ex:levels} are runs of an automaton in $\aww[3,\R]$. The following lemma shows that for $\aww[2,\R]$ the level sequence does determine acceptance; in other words, in order to decide whether a run is accepting or not it suffices to examine its level sequence:

\begin{lemma}\label{lem:rundag:normalform}
Let $\mathcal{A} = \langle \Sigma, Q, \theta_0, \delta, \alpha \rangle$ be an $\aww[2,\R]$. A run $G = (V, E)$ of $\mathcal{A}$ on a word $w$ is accepting if{}f  
\begin{itemize}
\item[(a)] the level sequence of $G$ is the infinite sequence $Q_0 \, Q_1 \, Q_2 \cdots$, and
\item[(b)] there is a threshold $k \geq 0$ such that  $Q_l$ is an $\alpha$-level for every $l \geq k$. 
\end{itemize}
\end{lemma}
\begin{proof}
We first observe that the level sequence of $G$ satisfies the following property for every $l \geq 0$: If the sets $Q_l$ and $Q_{l+1}$ belong to the sequence, then $Q_{l+1}$ is a $w[l]$-successor of $Q_l$. 
Indeed, if $Q_{l+1}$ belongs to the level sequence, then $\delta(q, a) \neq \false$ for every $q \in Q_l$ and so, by the definition of a run, the set $Q_{l+1}$ is a $w[l]$-successor of $Q_l$. Now we prove the two directions of the lemma.

\medskip\noindent ($\Rightarrow$):  Let $G=(V, E)$ be an accepting run of $\mathcal{A}$ on $w$. By the definition of acceptance, we have $\delta(q, w[l]) \neq \false$ for every $l \geq 0$ and $q \in Q_l$, and so, by the observation above, $Q_{l+1}$ is a $w[l]$-successor of $Q_l$. By the definition of the level sequence of a run,  the level sequence of $G$ is infinite. This proves (a).

Let us now prove (b). Say that a node $(q, l) \in V$ is accepting if $q \in \alpha$, and rejecting otherwise. We prove two claims.

\smallskip\noindent Claim 1: All descendants of an accepting node of $G$ are also accepting. \\
Since  $\mathcal{A}$ is an $\aww[2,\R]$, the initial formula $\theta_0$ of $\mathcal{A}$ satisfies $\theta_0 \in \mathcal{B}^+(Q \setminus \alpha)$ by \Cref{def:awwn}, and so every node $(q, 0)$ of $G$ is rejecting. So every node of $G$ is a descendant of a rejecting node. Since $\mathcal{A}$ has alternation height 2, every path of $G$ has at most \emph{one} alternation of accepting and rejecting nodes. Therefore, no descendant of an accepting node is rejecting. 

\smallskip\noindent Claim 2:  The set $R \subseteq V$ of rejecting nodes of $V$ is finite. \\
 Assume that $R$ is infinite. Let $G_R = (R, E \cap (R\times R))$, and let $(q, l) \in R$. If $l \geq 1$, then by the definition of a run, $(q, l)$ has at least one predecessor in $G$. By Claim 1, all predecessors of a rejecting node are rejecting, and so $(q, l)$ also has at least one predecessor in $G_R$. It follows that the only nodes of $G_R$ without predecessors are those of the form $(q, 0)$, and so finitely many. Since $R$ is infinite and $G_R$ is finitely branching, $G_R$ contains an infinite path by König's lemma. By the definition of $G_R$, this path only contains non-accepting nodes, contradicting that $G$ is an accepting run. This proves the claim.

\smallskip\noindent By Claim 2, there exists a threshold $k$ such that $k \geq m$ for every $(q, m) \in R$. It follows that $(q, l) \in V$ is an accepting node for every $l \geq k$, and so that  $Q_l$ is an $\alpha$-level for every $l \geq k$. 

\medskip\noindent ($\Leftarrow$):  We prove this direction for arbitrary alternating automata. Let $G=(V, E)$ be a run satisfying (a) and (b). We show that it also satisfies the two properties of an accepting run:
\begin{itemize}
\item $\delta(q, w[l]) \not \equiv \false$ for every $l \geq 0$ and every $q \in Q_l$. \\
By (a), the level sequence of $G$ is infinite, and so, by the definition of the level sequence of a run, we are done.
\item  Every infinite path of $G$ contains infinitely many accepting nodes. \\
Every infinite path of $G$ contains exactly one node of the form $(q, l)$ for every $l \geq 0$. By (b), the node $(q, l)$ of the path is accepting for every $l \geq k$. 
\end{itemize}
\end{proof}

\subsubsection{A determinization procedure for $\aww[2,\R]$.} 

Let $\mathcal{A}$ be an $\aww[2,\R]$ with set of states $Q$ and $|Q|=n$. We construct an equivalent deterministic co-Büchi automaton $\mathcal{D}$ whose states  are pairs $(\Levels, \Promising)$, where $\Levels \subseteq 2^Q$ (recall that a level of an $\aww$ is a set of states) and $\Promising \subseteq 2^\alpha \cap \Levels$. Since there exist $3^{2^n}$ pairs satisfying these conditions, 
$\mathcal{D}$ has at most $3^{2^n}$ states. 

Let us first give an informal but hopefully intuitive description of the transitions and acceptance condition of the deterministic co-Büchi automaton. The transitions of $\mathcal{D}$ are chosen to ensure that, after reading a finite word $w_{0i} = a_0 \ldots a_i$, the automaton is in the state $(\Levels_i$, $\Promising_i)$, where 
\begin{itemize}
\item $\Levels_i$ contains the $i$-th levels of all runs of $\mathcal{A}$ on all words having $w_{0i}$ as prefix (when they exist); and 
\item $\Promising_i \subseteq \Levels_i$ contains some $\alpha$-levels of $\Levels_i$. These levels are \enquote{promising}, intuitively meaning that they could belong to the infinite tail of $\alpha$-levels of the level sequence of an accepting run. 
\end{itemize}
For this, when $\mathcal{D}$ reads $a_{i+1}$, it moves from $(\Levels_i$, $\Promising_i)$ to $(\Levels_{i+1}$, $\Promising_{i+1})$, where $\Levels_{i+1}$ contains all the $a_{i+1}$-successors of the levels of $\Levels_i$, and $\Promising_{i+1}$ is defined as follows:
\begin{itemize}
\item If $\Promising_i \neq \emptyset$, then $\Promising_{i+1}$ contains the $a_{i+1}$-successors of $\Promising_i$. (Recall that $\mathcal{A}$ is an $\aww[2,\R]$, and so if $\Promising_i$ contains $\alpha$-levels, then so does  $\Promising_{i+1}$.)
\item If $\Promising_i = \emptyset$, then $\Promising_{i+1}$ contains all $\alpha$-levels of $\Levels_{i+1}$.
\end{itemize}

\noindent Finally, the co-Büchi condition contains the states $(\Levels$, $\Promising)$ such that $\Promising=\emptyset$.

Intuitively, during its run on a word $w$, the automaton $\mathcal{D}$ tracks the promising levels, removing those without successors, because they can not belong to an accepting run. If some run $G$ of $\mathcal{A}$ accepts $w$, then by
\Cref{lem:rundag:normalform} the level sequence $Q_0 \, Q_1 \, Q_2 \cdots$ of $G$ is infinite, and there is $k \geq 0$ such that the levels $Q_k , Q_{k+1} ,Q_{k+2} \cdots$ are all promising. So the sets $\Promising_k, \Promising_{k+1}, \Promising_{k+2} \cdots $ are all nonempty, and $\mathcal{D}$ accepts. Conversely, assume there is $k$ such that  $\Promising_k, \Promising_{k+1}, \Promising_{k+2} \cdots $ are all nonempty.  Then we can construct a run of $\mathcal{A}$ on $w$ by picking a sequence of levels $U_k \, U_{k+1} \, U_{k+2} \, \cdots$ such that $U_k \in \Promising_k$ and $U_l$ is a $w[l]$-successor of $U_{l-1}$ for every $l > k$, which belong by definition to $\Promising_l$, and then picking $U_0 \, U_1 \cdots U_{k-1}$ such that $U_l \in \Levels_l$ for every $0 \leq l \leq k-1$. The level sequence of this run is infinite, and from $U_k$ onwards it only contains $\alpha$-levels, which implies that the run is accepting. 

\begin{lemma}\label{lem:seqkonig}
For every $\mathcal{A} \in \aww[2,\R]$ and infinite word $w$ on $\Sigma$, if there is $k \geq 0$ such that $\Promising_i \neq \emptyset$ for all $i \geq k$, there is an infinite sequence $U_k \, U_{k+1} \cdots U_i \, U_{i+1} \cdots$ such that $U_i \in \Promising_i$ and $U_{i+1}$ is a $w[i+1]$-successor of $U_i$ for all $i \geq k$.
\end{lemma}

\begin{proof}
Let $F = (V, E)$ be the graph with vertices $V = \bigcup_{i \geq k} \{(U_i, i) \mid U_i \in \Promising_i \}$ and edges $E = \bigcup_{i \geq k} \{((U, i),\allowbreak (U', i + 1)) \mid U \in \Promising_i \text{ and $U'$ is a $w[i]$-successor of $U$}\}$. This is well-defined since such a $U'$ belongs to $\Promising_{i+1}$ by definition. $V$ is infinite because $\Promising_i \neq \emptyset$ for all $i \geq k$, and the degree of each vertex is finite since $|\Promising_i| \leq 2^{|Q|}$. Moreover, $F$ is acyclic, because the second entry of the pair is monotonically increasing, and consists of as many connected components as $U_k \in \Promising_k$. Indeed, $U_{i+1} \in \Promising_{i+1}$ is a $w[i+1]$-successor of some $U_i \in \Promising_i$ by definition of the $\Promising$ sets, so inductively, we can go back to a set $U_k \in \Promising_k$. Hence, $F$ is the union of $|\Promising_k| \leq 2^{|Q|}$ trees, and at least one of them must be infinite because $V$ is.

Let $T$ be an infinite tree of $F$. Since it is both finite and finitary, by the König lemma, there is an infinite branch $(U_k, k), (U_{k+1}, k+1), \ldots$. We can then extract the infinite sequence $U_k \, U_{k+1} \cdots$ satisfying the conditions in the statement of the lemma.
\end{proof}

For the formal definition of $\mathcal{D}$ it is convenient to identify subsets of ${2^Q}$ and $2^\alpha$ with formulas of $\mathcal{B}^+(Q)$, $\mathcal{B}^+(\alpha)$ (i.e., we identify a  formula and its set of models). Further, we lift $\delta \colon Q \times \Sigma \mapsto \mathcal{B}^+(Q)$ to $\delta \colon \mathcal{B}^+(Q) \times \Sigma \mapsto \mathcal{B}^+(Q)$ in the canonical way. Finally, given $\varphi \in \mathcal{B}^+(Q)$ and $S \subseteq Q$, we let 
$\varphi[\false / S]$ denote the result of substituting $\false$ for every state of $S$ in $\varphi$. With these notations, the deterministic Büchi automaton $\mathcal{D}$ equivalent to $\mathcal{A}$ can be described in four lines: $\mathcal{D} = \langle \Sigma, Q', q_0', \delta', \alpha' \rangle$,  where $Q' = \mathcal{B}^+(Q) \times \mathcal{B}^+(\alpha)$, $q_0' = (\theta_0, \false)$, $\alpha' = \{(\theta, \false) \colon \theta \in \mathcal{B}^+(Q) \}$, and
\[\delta'((q, p), a) = \begin{cases}
	(\delta(q, a), \delta(p, a)) &  \text{if $p \not \equiv \false$} \\
	(\delta(q, a), \delta(q, a)[\false / Q \setminus \alpha]) & \text{otherwise.}
\end{cases}\]
\begin{lemma}\label{lem:det:1}
For every $\mathcal{A} \in \aww[2,\R]$ with $n$ states, the deterministic co-Büchi automaton $\mathcal{D}$ defined above satisfies $L(\mathcal{A}) = L(\mathcal{D})$, and has $3^{2^n}$ states. Dually, for every $\mathcal{A}' \in \aww[2,\A]$ with $n'$ states, there exists a deterministic Büchi automaton $\mathcal{D}'$ that has $3^{2^{n'}}$ states and that satisfies $L(\mathcal{A}') = L(\mathcal{D}')$.
\end{lemma}

\begin{proof}
Assume $\mathcal{A}$ accepts a word $w$.
Let $G = (V, E)$ be an accepting run of $\mathcal{A}$ on $w$. Since $\mathcal{A} \in \aww[2,\R]$, we can apply 
\Cref{lem:rundag:normalform}, and so there exists an index $k$ such that all levels of $G$ after the $k$-th one are $\alpha$-levels and have at least one successor.  It follows that the run $(\Levels_0, \Promising_0), (\Levels_1, \Promising_1) \ldots$ of $\mathcal{D}$ on $w$ satisfies $\Promising_l\neq \emptyset$ for every $l \geq k$, and so $\mathcal{D}$ accepts.

Assume $\mathcal{D}$ accepts a word $w$. Let $(\Levels_0$, $\Promising_0)$, $(\Levels_1$, $\Promising_1) \ldots$ be the run of $\mathcal{D}$ on $w$. By the definition of the acceptance condition of $\mathcal{D}$, there is $k \geq 0$ such that $\Promising_i \neq \emptyset$ for every $i \geq k$. By \cref{lem:seqkonig}, there is an infinite sequence $U_k \, U_{k+1} \cdots U_{i} \cdots$ such that $U_i \in \Promising_i$ and $U_{i+1}$ is a $w[i+1]$-successor of $U_i$ for every $i \geq k$. Then, choose levels $U_0 \, U_1 \cdots U_{k-1}$ such that for every $1 \leq l \leq k$, the level $U_{l-1}$ is a predecessor of $U_l$ that belongs to $\Levels_{l-1}$ (that is, $U_l$ is a $w[l]$-successor of $U_{l-1}$; this is always possible by the definition of $\delta'$).
Let $G=(V, E)$ be the graph given by
\begin{itemize}
\item for every $l \geq 0$, $(q, l) \in V$ if{}f $q \in U_l$; and
\item $((q,l), (q', l+1)) \in E$ if{}f $q \in U_l$ and $q' \in S_q$, where $S_q$ is a minimal
model of $\delta(q, w[l])$ satisfying $U_{l+1} = \bigcup_{q \in U_l} S_q$, which exists by definition.
\end{itemize}
Since $U_k \in \Promising_k$, and promising sets only contain $\alpha$-levels, $U_k$ is an $\alpha$-level. Since
$\mathcal{A} \in \aww[2,\R]$, successors of $\alpha$-levels are also $\alpha$-levels, and so all of $U_k, U_{k+1}, U_{k+2}, \ldots$ are $\alpha$-levels. By \Cref{lem:rundag:normalform}, $G$ is an accepting run of $\mathcal{A}$.

The dual part of the lemma is proven by complementing $\mathcal{A}'$, applying the construction we have just described, and replacing the co-Büchi condition by a Büchi condition.
\end{proof}

\subsubsection{A determinization procedure for $\aww[2]$.} 

\noindent \Cref{lem:det:1} leads to a determinization procedure for $\aww[2]$. Combining the procedures for $\aww[2,\R]$
and $\aww[2,\A]$, we show how to construct for every automaton in $\aww[2]$ an equivalent deterministic Rabin automata.

Rabin automata generalize both Büchi and co-Büchi  automata. 
A Rabin automaton over an alphabet $\Sigma$ is a tuple $\mathcal{A} = \langle \Sigma, Q, q_0, \delta, \alpha \rangle$, where $\Sigma, Q, q_0, \delta$ are defined as for Büchi or co-Büchi automata, and $\alpha$ is a set of \emph{Rabin pairs} $(F, I) \subseteq Q \times Q$. Deterministic Rabin automata are defined as for Büchi or co-Büchi automata (the definitions do not involve the accepting condition $\alpha$). The same applies to the runs of a Rabin automaton. The only new definition is that of an accepting run. A run $G=(V,E)$ of a Rabin automaton is \emph{accepting} if
\begin{itemize}
\item[(a)] $\delta(q, w[l]) \not\equiv \false$ for every $q \in Q_l$, and
\item[(b)] for every infinite path of $G$ there exists a Rabin pair $(F,I) \in \alpha$ such that the path visits infinitely often nodes $(q, l) \in V$ such that $q \in I$, and only finitely often nodes $(q, l) \in V$ such that $q \in F$.
\end{itemize}
We abbreviate deterministic Rabin automaton to DRW.

\begin{lemma}\label{lem:det:3}
For every $\mathcal{A} = \langle \Sigma, Q, \theta_0, \delta, \alpha \rangle \in \aww[2]$ with $n = |Q|$ states and $m = |\mathcal{M}_{\theta_0}|$ minimal models of $\theta_0$ there exists an equivalent DRW $\mathcal{D}$ with $2^{2^{n + \log_2 m + 2}}$ states and with $m$ Rabin pairs.
\end{lemma}

\begin{proof}
Let $\mathcal{A}=  \langle \Sigma, Q, \theta_0, \delta, \alpha \rangle$. Given $Q' \subseteq Q$, let $\mathcal{A}_{Q'}$ be the AWW[2] obtaining from $\mathcal{A}$ by substituting $\bigwedge_{q \in Q'} q$ for the initial formula $\theta_0$.  We claim that for each minimal model $S \in \mathcal{M}_{\theta_0}$ we can construct a deterministic Rabin automaton (DRW) $\mathcal{D}_S$ with at most $2^{2^{n+2}}$ states and a single Rabin pair, recognizing the same language as $\mathcal{A}_{S}$. Let us first see how to construct $\mathcal{D}$, assuming the claim holds. By the claim we have $\lang(\mathcal{A}) = \bigcup_{S \in \mathcal{M}_{\theta_0}} \lang(\mathcal{A}_{S})$.
So we define $\mathcal{D}$ as the union of all the automata $\mathcal{D}_S$. 
Recall that given two  DRWs with $n_1, n_2$ states and $p_1, p_2$ Rabin pairs we can construct a DRW for the union of their languages with $n_1 \times n_2$ states and $p_1 + p_2$ pairs. Since $\theta_0$ has $m$ models, $\mathcal{D}$ has at most $m$ Rabin pairs and $\left(2^{2^{n+2}}\right)^m = 2^{2^{n + \log_2 m + 2}}$ states.

It remains to prove the claim. Partition $S$ into $ S \cap \alpha$ and $S \setminus \alpha$. We have $\mathcal{A}_{S \cap \alpha} \in \aww[2,\A]$ and $\mathcal{A}_{S \setminus \alpha} \in \aww[2,\R]$. By \Cref{lem:det:1} there exists a deterministic Büchi automaton $\mathcal{D}_{S \cap \alpha}$ and a deterministic co-Büchi automaton $\mathcal{D}_{S \setminus \alpha}$ equivalent to  $\mathcal{A}_{S \cap \alpha}$ and
$\mathcal{A}_{S \setminus \alpha}$, respectively, both with at most $3^{2^n}$ states. Intersecting these two automata yields a deterministic Rabin automaton with at most $3^{2^{n+1}} \leq 2^{2^{n+2}}$ states and a single Rabin pair, and we are done.
\end{proof}

\begin{remark}
The construction of \Cref{lem:det:1} is close to Miyano and Hayashi's translation of alternating automata to non-de\-ter\-min\-ist\-ic automata \cite{MiyanoH84}, and to Schneider's
translation of $\Sigma_2$ formulas to deterministic co-Büchi automata \cite[p.219]{DBLP:series/txtcs/Schneider04}, all based on the break-point idea.
\end{remark}

\subsection{Translation of LTL to DRW}

Combining a normalization procedure LTL$\rightarrow\Delta_2$, the procedure $\Delta_2\rightarrow\alw[2]$ of \Cref{subsec:LTLtoAWW2} and the determinization procedure of 
\Cref{subsec:det}, we obtain a translation  LTL$\rightarrow$DRW.  Since the procedures involve a single-exponential, 
a linear, and a double-exponential blow-up, respectively, a straightforward composition only yields a triple-exponential bound. However, a closer examination of the closed-form $\Delta_2$-formula provided by \cref{thm:normthm} allows us to reduce the bound to double-exponential. 

Given a formula $\varphi$, by \cref{thm:normthm} we have 
\begin{align}
\label{eq:norm}
\varphi \equiv \bigvee_{\substack{\setmu \subseteq \mubasis\\\setnu \subseteq \nubasis}} \varphi_{\setmu,\setnu}
\quad \text{ where } \quad
 \varphi_{\setmu,\setnu} \coloneqq \bigg( \flatten{\varphi}{M} \wedge \bigwedge_{\F\G \psi \in N} \F\G (\eval{\psi}{M}) \wedge \bigwedge_{\G\F \psi \in M} \G\F (\evalgf{\psi}{N})  \bigg)\ .
 \end{align}
The key result is that each $\varphi_{\setmu,\setnu}$ has a small number of different proper subformulas. (In fact, the number is even linear in the number of subformulas of $\varphi$.) Invoking \cref{lem:aww:translation} we obtain an
$\alw[2]$ with $O(|\subf(\varphi)|)$ states that recognizes $\lang(\varphi_{\setmu,\setnu})$.

\begin{restatable}{lemma}{lemSizeSf}\label{lem:size:sf}
Let $\varphi$ be a formula. For every $M \subseteq \mubasis$ and $N \subseteq \nubasis$, there exists an $\alw[2]$ with $O(|\subf(\varphi)|)$ states that recognizes $\lang(\varphi_{\setmu,\setnu})$.
\end{restatable}
\begin{proof}
By \Cref{lem:aww:translation}, some $\alw[2]$ with 
$O(|\subf(\varphi_{\!\setmu,\setnu})|)$ states recognizes $\lang(\varphi_{\setmu,\setnu})$. So it suffices to show that $|\subf(\varphi_{\setmu,\setnu}) | \in O(|\subf(\varphi)|)$, but this is \cref{lem:subfbound}.
\end{proof}

\begin{theorem}
\label{thm:mainLTLDRW}
Let $\varphi$ be a formula with $n$ proper subformulas. Let  $\mathcal{A}_\varphi$ be the DRW obtained by
\begin{itemize}
\item normalizing $\varphi$ into  $\displaystyle \bigvee_{\setmu \subseteq \mubasis, \setnu \subseteq \nubasis} \varphi_{\setmu,\setnu}$ as in \cref{eq:norm};\\[-0.1cm]
\item constructing a $\alw[2]$ $\mathcal{A}_{M,N}$ for each $\varphi_{M, N}$ as in Definition \ref{def:aawforltl};
\item transforming each $\mathcal{A}_{M,N}$ into an equivalent DRW $\mathcal{A}'_{M,N}$ as in \Cref{lem:det:3}; and
\item constructing a DRW recognizing the union of the languages of all $\mathcal{A}'_{M,N}$.
\end{itemize}
The DRW  $\mathcal{A}_\varphi$  recognizes $\lang(\varphi)$, and has $2^{2^{\mathcal{O}(n^2)}}$ states and $2^n$ Rabin pairs.
\end{theorem}

\begin{proof}
Due to \cref{lem:size:sf}, the alternating automaton 
 $\mathcal{A}_{M,N}$ that recognizes $\lang(\varphi_{M,N})$ belongs to $\alw[2]$ and has $O(n)$ states. Applying the construction of \Cref{lem:det:3} we obtain a DRW with $2^{2^{O(n)}}$ states and a single Rabin pair. Using the union operation for DRWs we obtain a DRW for $\varphi$ with 
 $\left(2^{2^{O(n)}}\right)^{2^{O(n)}} = 2^{2^{O(n)}}$ states and $2^n$ Rabin pairs.
\end{proof}

\subsection{Determinization of Lower Classes}

We now determinize $\aww[1]$.
A deterministic automaton is \emph{terminal-accepting} if all states are rejecting except a single accepting sink with a self-loop, and \emph{terminal-rejecting} if all states are accepting except a single rejecting sink with a self-loop. It is easy to see that terminal-accepting and terminal-rejecting deterministic automata are closed under union and intersection. When applied to $\aww[1,\A]$, the construction of \Cref{lem:det:1}, yields automata whose states have a trivial $\Promising$ set (either the empty set or the complete level). Further, the successor of an $\alpha$-level is also an $\alpha$-level. From these observations we easily get:

\begin{corollary}\label{lem:det:4}
Let $\mathcal{A}$ be an automaton with $n$ states.
\begin{itemize}
   \item If $\mathcal{A} \in \aww[1,\R]$ (resp. $\mathcal{A} \in \aww[1,\A]$), then there exists a deterministic terminal-accepting (resp. termi\-nal-re\-jecting) automaton recognizing $\lang(\mathcal{A})$ with $2^{2^n}$ states.
    \item If $\mathcal{A} \in \aww[1]$, then there exists deterministic weak automaton recognizing $\lang(\mathcal{A})$ with $2^{2^{n + \log_2 |\mathcal{M}_{\theta_0}| + 1}}$ states.
\end{itemize}
\end{corollary}

 \section{A Hierarchy of Alternating Weak and Very Weak Automata}
\label{sec:hierarchy-automata}

\newcommand{\awwG}{\aww_{\text{G}}}
\newcommand{\alwPS}{\alw_{\text{PS}}}

The expressive power of weak and very weak alternating automata has been studied
by Gurumurthy \textit{et al.} in \cite{GurumurthyKSV03} and 
by Pel{\'{a}}nek and Strejcek in \cite{PelanekS05},
respectively.  Both papers identify the number of alternations 
between accepting and non-accepting states as an important parameter, and define
a hierarchy of automata classes based on it.  Let $\awwG[k]$ denote the class of $\aww$ with at most $(k-1)$ alternations defined in \cite{GurumurthyKSV03}. Similarly, let $\alwPS[k,\A]$
and $\alwPS[k,\R]$ denote the classes of $\alw$ with at most $(k-1)$ alternations and 
accepting or non-accepting initial state, respectively, defined in \cite{PelanekS05}. Finally, define $\alwPS[k] = \alwPS[k,\A] \cup \alwPS[k,\R]$\footnote{In \cite{PelanekS05} the classes have different names.}.
\Cref{fig:hierarchies} shows the results of \cite{GurumurthyKSV03}  
and \cite{PelanekS05}. We abuse language, and, for example, write $\Pi_2 = \alwPS[2, \A]$
to denote that the class of languages satisfying formulas in $\Pi_2$ and the class of languages recognized by automata 
in $\alwPS[2, \A]$ coincide.

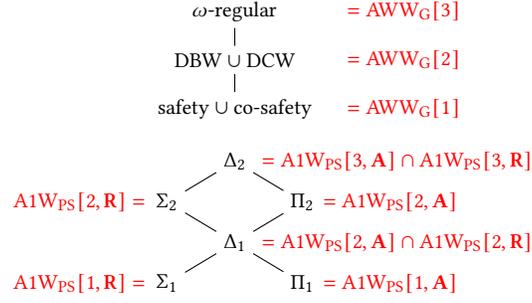
\begin{figure}[bt]
  \small
  \begin{center}
	\begin{tikzpicture}[x=1cm,y=0.60cm,outer sep=0pt,scale=0.9]
	
	\node (m3) at (0, 3.7) {$\omega$-regular}; \node[color=red] (m3l) at (2.5, 3.7) {$= \awwG[3]$};
	\node (m2) at (0, 2.5) {DBW $\cup$ DCW}; \node[color=red] (m2l) at (2.5, 2.5) {$= \awwG[2]$};
	\node (m1) at (0, 1.3) {safety $\cup$ co-safety}; \node[color=red] (m2l) at (2.5, 1.3) {$= \awwG[1]$};
	
	\path[-]
	(m2) edge node{} (m1)
     (m3) edge node{} (m2);

    \node (padding1) at ( 0,0.3) {};
	\node (1) at ( 0, 0) {$\Delta_2$}; \node[color=red] (1l) at (2.4, 0) {$= \alwPS[3,\A] \cap \alwPS[3,\R]$};
	\node (2) at ( 1,-1) {$\Pi_2$}; \node[color=red] (2l) at (2.3, -1) {$= \alwPS[2,\A]$};
	\node (3) at (-1,-1) {$\Sigma_2$}; \node[color=red] (3l) at (-2.3, -1) {$\alwPS[2,\R] = $};
    \node (4) at ( 0,-2) {$\Delta_1$}; \node[color=red] (4l) at (2.4, -2) {$= \alwPS[2,\A] \cap \alwPS[2,\R]$};
	\node (5) at ( 1,-3) {$\Pi_1$}; \node[color=red] (5l) at (2.3, -3) {$= \alwPS[1,\A]$};
	\node (6) at (-1,-3) {$\Sigma_1$}; \node[color=red] (6l) at (-2.3, -3) {$\alwPS[1,\R] = $};

	\path[-]
	(2) edge node{} (1)
    (3) edge node{} (1)
    
    (4) edge node{} (2)
    (4) edge node{} (3)
    
    (5) edge node{} (4)
    (6) edge node{} (4);
	\end{tikzpicture}
  \end{center}
\caption{Expressive power of AWWs after Gurumurthy \textit{et al.} \cite{GurumurthyKSV03}, and 
of A1Ws after Pel{\'{a}}nek and Strejcek \cite{PelanekS05}.}
\label{fig:hierarchies-aww}
\end{figure}

Unfortunately, the results of \cite{GurumurthyKSV03} and \cite{PelanekS05} do not \enquote{match}. Due to slight differences in the definitions of height, e.g.\ the treatment of $\delta(\cdot) = \false$ and $\delta(\cdot) = \true$, the restriction to very weak automata of $\awwG[k]$ does not match any class $\alwPS[k']$ (that is, $\awwG[k] \cap \alw \neq \alwPS[k']$) and, vice versa, extending $\alwPS[k]$ does not yield any $\awwG$ $[k']$. We show that our new definition of height unifies the two hierarchies, yielding the pleasant result shown in \Cref{fig:hierarchies2}.

\begin{figure}[bt]
  \small
  \begin{center}
	\begin{tikzpicture}[x=1cm,y=0.60cm,outer sep=0pt,scale=0.9]
	
	\def\d{1.1}
	
	\node (m3) at (0, 4*\d) {$\omega$-regular}; \node[color=red] (m3l) at (1.6, 4*\d) {$= \aww[2]$};
	\node (m2le) at (-0.7, 3*\d) {\dcw}; \node[color=red] (m2l) at (-2.2, 3*\d) {$\aww[2, \R] =$};
	\node (m2ri) at (0.7, 3*\d) {\dbw}; \node[color=red] (m2l) at (2.2, 3*\d) {$= \aww[2, \A]$};
	\node (m1) at (0, 2*\d) {DWW}; \node[color=red] (m2l) at (1.4, 2*\d) {$= \aww[1]$};
	\node (m0le) at (-0.7, 1*\d) {co-safety}; \node[color=red] (m2l) at (-2.5, \d) {$\aww[1,\R] =$};
	\node (m0ri) at (0.7, 1*\d) {safety}; \node[color=red] (m2l) at (2.3, \d) {$= \aww[1,\A]$};
	
	\path[-]
	(m1) edge node{} (m0ri)
     (m1) edge node{} (m0le)
	(m2le) edge node{} (m1)
	(m2ri) edge node{} (m1)
     (m3) edge node{} (m2ri)
     (m3) edge node{} (m2le);

    \node (padding1) at ( 0,0.3) {};
	\node (1) at ( 0, 0) {$\Delta_2$}; \node[color=red] (1l) at (1, 0) {$= \alw[2]$};
	\node (2) at (0.7,-1*\d) {$\Pi_2$}; \node[color=red] (2l) at (1.9, -1*\d) {$= \alw[2,\A]$};
	\node (3) at (-0.7,-1*\d) {$\Sigma_2$}; \node[color=red] (3l) at (-1.9, -1*\d) {$\alw[2,\R] = $};
    \node (4) at ( 0,-2*\d) {$\Delta_1$}; \node[color=red] (4l) at (1, -2*\d) {$= \alw[1]$};
	\node (5) at ( 0.7,-3*\d) {$\Pi_1$}; \node[color=red] (5l) at (1.9, -3*\d) {$= \alw[1,\A]$};
	\node (6) at (-0.7,-3*\d) {$\Sigma_1$}; \node[color=red] (6l) at (-1.9, -3*\d) {$\alw[1,\R] = $};

	\path[-]
	(2) edge node{} (1)
     (3) edge node{} (1)
    
    (4) edge node{} (2)
    (4) edge node{} (3)
    
    (5) edge node{} (4)
    (6) edge node{} (4);

	\end{tikzpicture}
  \end{center}
\caption{Expressive power of AWWs and A1Ws}
\label{fig:hierarchies2}
\end{figure}
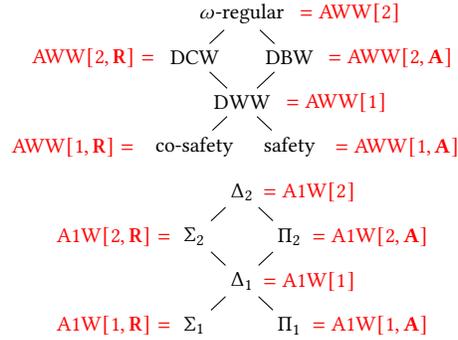

\begin{proposition}\label{prop:hierarchy:aww} 
We have:
\begin{enumerate}
\item $\!\aww[2] =$ $\omega$-regular, $\aww[2,\A]$ $=$ $\dbw$, and $\aww[2,\R]$ $=$ $\dcw$.
\item $\aww[1]$ $\!=$ \emph{DWW}, $\aww[1,\A] =$ safety, and $\aww[1,\R] =$ co-safety.
\item For $i\in \{1,2\}$,  $\alw[i] = \Delta_i$, $\alw[i,\R] = \Sigma_i$, and $\alw[i,\A] = \Pi_i$.
\end{enumerate}
\end{proposition}
\begin{proof} We sketch the proof of (1) and (3). The proof of (2) is analogous to that of (1). 

\medskip

\noindent (1): The $\subseteq$-inclusion follows immediately from \Cref{lem:det:1,lem:det:3} and \Cref{lem:det:4}. The $\supseteq$-inclusion is a slight adaptation of similar proofs in \cite{GurumurthyKSV03}. In order to translate a \dcw{} into a $\aww[2,\R]$ we duplicate the set of states into two sets of marked and unmarked states. We remove from the marked states all rejecting states, and add transitions that allow unmarked states to nondeterministically choose to move to another unmarked state, or to its marked copy. Finally, we define all unmarked states to be rejecting and all marked states to be accepting. The proof of $\aww[2,\A] \supseteq \dbw{}$ is dual. Finally, the inclusion $\aww[2] \supseteq \text{$\omega$-regular}$ follows from the previous two results; indeed, every $\omega$-regular language is recognized by a \drw{} \cite{Safra88}, and every \drw{} is equivalent to a Boolean combination of {\dbw}s and {\dcw}s, which we can express in the initial formula $\theta_0$ of the $\aww[2]$. 

\medskip

\noindent (3): The $\supseteq$-inclusion for $\Delta_i$ is proven in \Cref{lem:aww:translation}. For a formula $\varphi$ that belongs to $\Sigma_i$ ($\Pi_i$) we also rely on \Cref{lem:aww:translation}, but add a new initial state, $\sta{\varphi}{\Sigma_i}$ ($\sta{\varphi}{\Pi_i}$) that is marked as rejecting (accepting) such that the automaton belongs to $\alw[i,\R]$ ($\alw[i,\A]$). 

For the $\subseteq$-inclusion, let $\mathcal{A} = \langle \Sigma, Q, \theta_0, \delta, \alpha \rangle$ be a very weak alternating automaton with $\Sigma = 2^{Ap}$. We use the translation from $\alw$ to LTL presented in \cite[Thm. 6]{LodingT00}, with minimal modifications, to define a formula $\psi_\mathcal{A}$ such that $\lang(\psi_\mathcal{A})=\lang(\mathcal{A})$. Then, we show that when $\mathcal{A}$ belongs to one of the classes in the hierarchy,
$\psi_\mathcal{A}$ belongs to the corresponding class of formulas. For the proof of correctness of the translation we refer the reader to \cite{LodingT00}.

For the definition of $\psi_\mathcal{A}$, we assign to every $\theta \in \mathcal{B}^+(Q)$ an LTL formula $\psi(\theta)$ such that $\lang(\psi(\theta)) = \lang(\mathcal{A}_\theta)$, where $\mathcal{A}_\theta$ denotes $\mathcal{A}$ with $\theta$ as initial formula, and set $\psi_\mathcal{A} := \psi(\theta_0)$. Similarly, for the definition of $\psi(\theta)$, we first assign a formula $\psi(q)$ to every state $q$, and then define $\psi(\theta)$ as the result of substituting $\psi(q)$ for $q$ in $\theta$, for every state $q$. It remains to define $\psi(q)$. Using that $\mathcal{A}$ is very weak, we proceed inductively, i.e., we assume that $\psi(q')$ has already been defined for all $q'$ such that $q \rightarrow q'$ and $q \neq q'$.  For every $q \in Q$ and $\sigma \in 2^{Ap}$, let $\theta_{q,\sigma}$ and $\theta'_{q,\sigma}$ be formulas such that  $\delta(q, \sigma) \equiv (q \wedge \theta_{q,\sigma}) \vee \theta'_{q,\sigma}$ (it is easy to see that they exist). Define
\[\psi(q) := \begin{cases}
	\varphi_q \U \varphi_q' & \text{if $q \notin \alpha$} \\
	\varphi_q \W \varphi_q' & \text{if $q \in \alpha$} \\
\end{cases}\]
\noindent where  \[ \varphi_q  := \bigvee_{\sigma \subseteq \Sigma} \left(\bigwedge_{a \in \sigma} a \wedge \bigwedge_{a \notin \sigma} \neg a \wedge \X\, \psi(\theta_{q,\sigma})\right)  \qquad \varphi_q'  := \bigvee_{\sigma \subseteq \Sigma} \left(\bigwedge_{a \in \sigma} a \wedge \bigwedge_{a \notin \sigma} \neg a \wedge \X\, \psi(\theta'_{q,\sigma})\right)  \ . \]

Since this translation assigns to each $\U$-formula a rejecting state and to each $\W$-formula an accepting state, the syntax tree of $\psi_\mathcal{A}$ has an alternation between $\U$ and $\W$ exactly when there is an alternation between accepting and non-accepting states. This yields all the desired inclusions in $\Sigma_1$, $\Pi_1, \ldots,$ $\Delta_2$. 
\end{proof}

\noindent Moreover, our single exponential normalization procedure for LTL transfers to a single exponential normalization procedure for $\alw$:
 
\begin{lemma}
Let $\mathcal{A}$ be an $\alw$ with $n$ states over an alphabet with $m$ letters. There exists $\mathcal{A}' \in \alw[2]$ with $2^{\mathcal{O}(nm)}$ states such that $\lang(\mathcal{A}) = \lang(\mathcal{A}')$.
\end{lemma}

\begin{proof}
The translation from $\alw$ to LTL used in \Cref{prop:hierarchy:aww} (an adaption of \cite{LodingT00}) yields a formula $\psi_\mathcal{A}$ with at most $\mathcal{O}(mn)$ proper subformulas. Applying our normalization procedure to $\psi_\mathcal{A}$ yields an equivalent formula in $\Delta_2$ with at most $2^{\mathcal{O}(mn)}$ proper subformulas (\Cref{lem:size:sf}). 
Applying \Cref{lem:aww:translation} we obtain the postulated automaton $\mathcal{A}'$.
\end{proof} 
 \section{Conclusion}
\label{sec:concl}

We have presented two purely syntactic normalization procedures for LTL that transform a given formula into an equivalent formula in $\Delta_2$, i.e., a formula with at most one alternation
between least- and greatest-fixpoint operators. The procedure has single exponential blow-up,
improving on the prohibitive non-elementary cost of previous constructions. The much better complexity of the new procedure (recall that normalization procedures for CNF and DNF are
also exponential) makes it attractive for its implementation and use in tools. We have presented 
a first promising application, namely a novel translation from LTL to DRW with double exponential blow-up. Finally, we have shown that the normalization procedure for LTL can be transferred to a 
normalization procedure for very weak alternating automata.

We think that these results demystify the Normalization Theorem of Chang, Manna, and Pnueli, which heavily relied on automata-theoretic results, and involved a nonelementary blowup. Indeed, the only conceptual difference between our rewrite system and the one for bringing Boolean formulas in CNF is the use of rewrite rules with contexts.

Our normalization procedure  has already found applications to the translation of LTL formulas into deterministic or limit-deterministic $\omega$-automata \cite{KretinskyMS18,MeyerSL18}. Until now normalization had not been considered, because of the non-elementary blow-up, much higher than the double exponential blow-up of existing constructions. With our new procedure, translations that first normalize the formula, and then apply efficient formula-to-automaton procedures specifically designed for formulas in normal form, have become competitive. Our system of rewriting rules makes this even more attractive. More generally, we think that the design of analysis procedures for formulas in normal form (to check satisfiability, equivalence, or other properties) should be further studied in the coming years.

\begin{acks}
We thank Nathana{\"{e}}l Fijalkow, J\"urgen Giesl, Marcin Jurdcinski, Jan K\v{r}et{\'{\i}}nsk{\'{y}}, and Orna Kupferman for very helpful comments and remarks. This work was partially supported by the \grantsponsor{DFG}{Deutsche Forschungsgemeinschaft (DFG)}{https://doi.org/10.13039/501100001659} under projects \grantnum{DFG1}{183790222}, \grantnum{CAVA2}{317422601}, and \grantnum{DFG3}{436811179}; by the \grantsponsor{ERC}{European Research Council (ERC)}{https://doi.org/10.13039/501100000781} under the European Union's Horizon 2020 research and innovation programme under grant agreement No~\grantnum{PaVeS}{787367} (PaVeS); by the \grantsponsor{AEI}{Agencia Estatal de Investigación (AEI)}{https://doi.org/10.13039/501100011033} under project \grantnum{ProCode}{PID2019-108528RB-C22}; and by the \grantsponsor{MU}{Spanish Ministry of Universities}{https://doi.org/10.13039/100014440} under grants FPU17/02319 and EST21/00536.
\end{acks}

\bibliographystyle{ACM-Reference-Format}

\end{document}